\documentclass{article}
\usepackage[english]{babel}
\usepackage[utf8x]{inputenc}
\usepackage[T1]{fontenc}
\usepackage[a4paper,top=2cm,bottom=2cm,left=3cm,right=3cm,marginparwidth=1.75cm]{geometry}
\usepackage{amsmath,amsthm}
\usepackage{graphicx}
\usepackage{float}
\usepackage{amsmath,amssymb,amsthm}
\usepackage{bm}
\usepackage{tikz}
\usepackage{comment}
\usepackage[numbers,sort]{natbib}
\usepackage{algorithm,algorithmic,float}
\usepackage{enumitem}
\usepackage{comment}
\usepackage{mathtools}
\usepackage{url}
\usepackage{hyperref}
\usepackage{cleveref}
\usepackage{adjustbox}
\usepackage{times}
\usepackage{multirow}
\usepackage{makecell}
\usepackage{siunitx}
\usepackage{longtable}
\usepackage[normalem]{ulem}

\usepackage{booktabs}

\hypersetup{
    colorlinks,
    breaklinks=true,
    linkcolor={blue},
    citecolor={blue},
    urlcolor={blue}
}

\newtheorem{lemma}{Lemma}[section]
\newtheorem{theorem}[lemma]{Theorem}

\newtheorem{proposition}[lemma]{Proposition}
\newtheorem{corollary}[lemma]{Corollary}

\theoremstyle{definition}

\theoremstyle{definition}
\newtheorem{definition}[lemma]{Definition}
\newtheorem{example}[lemma]{Example}
\newtheorem{remark}[lemma]{Remark}
\newtheorem{notation}[lemma]{Notation}
   
\newcommand{\CC}{\mathbb{C}}
\newcommand{\trdeg}{\operatorname{trdeg}}
\newcommand{\QQ}{\mathbb{Q}}
\newcommand{\ZZ}{\mathbb{Z}}

\newcommand{\balpha}{\bm{\alpha}}
\newcommand{\bbeta}{\bm{\beta}}
\newcommand{\bgamma}{\bm{\gamma}}
\newcommand{\bomega}{\bm{\omega}}
\newcommand{\bsigma}{\bm{\sigma}}

\newcommand{\bxi}{\bm{\xi}}
\newcommand{\bx}{\mathbf{x}}  
\newcommand{\ba}{\mathbf{a}}  
\newcommand{\bb}{\mathbf{b}}  
  
\newcommand{\by}{\mathbf{y}}  
\newcommand{\bh}{\mathbf{h}}

\newcommand{\bg}{\mathbf{g}}

\newcommand{\bmu}{\boldsymbol{\mu}}

\newcommand{\cE}{\mathcal{E}}

\newcommand{\grobner}{Gr{\"o}bner}
\newcommand{\NFp}{\ensuremath{\operatorname{NF}^{+}}}
\newcommand{\OMS}{\ensuremath{\operatorname{OMS}}}
\newcommand{\eOMS}{\ensuremath{\operatorname{eOMS}}}
\newcommand{\blackbox}{{\tt BB}}
\newcommand{\exname}[1]{{\tt #1}}

\newcommand{\favoritelabelsep}{4pt}
\newcommand{\favoriteitemsep}{2pt}

\newcommand{\Zplus}{\mathbb{N}_{0}}

\usepackage[dvipsnames]{xcolor}

\makeatletter
\newcommand*\rel@kern[1]{\kern#1\dimexpr\macc@kerna}
\newcommand*\widebar[1]{%
  \begingroup
  \def\mathaccent##1##2{%
    \rel@kern{0.8}%
    \overline{\rel@kern{-0.8}\macc@nucleus\rel@kern{0.2}}%
    \rel@kern{-0.2}%
  }%
  \macc@depth\@ne
  \let\math@bgroup\@empty \let\math@egroup\macc@set@skewchar
  \mathsurround\z@ \frozen@everymath{\mathgroup\macc@group\relax}%
  \macc@set@skewchar\relax
  \let\mathaccentV\macc@nested@a
  \macc@nested@a\relax111{#1}%
  \endgroup
}
\makeatother

\usepackage{tikz}
\usetikzlibrary{cd} 
\tikzset{
    lablvert/.style={anchor=south, rotate=90, inner sep=.5em}
}

\usepackage[textsize=tiny,colorinlistoftodos]{todonotes}
\usepackage{marginnote}
\usepackage{etoolbox}

\newcommand{\code}[1]{{#1}}

\title{Simple generators of rational function fields}
\author{Alexander Demin\footnote{{National Research University Higher School of Economics, Moscow, Russia, \url{asdemin_2@edu.hse.ru}}}, ~Gleb Pogudin\footnote{LIX, CNRS, \'Ecole polytechnique, Institute Polytechnique de Paris, Paris, France, \url{gleb.pogudin@polytechnique.edu}}}
\date{\today}

\begin{document}

\maketitle

\makeatletter
\begingroup
\renewcommand{\thefootnote}{}
\renewcommand{\@makefntext}[1]{\noindent #1}

\footnotetext{%
\vspace*{0.6em}%
\begin{minipage}[t]{0.85\linewidth}
\vspace{0pt}%
\itshape
Funded by the European Union. Views and opinions expressed are however
those of the author(s) only and do not necessarily reflect those of the
European Union or the European Research Council Executive Agency.
Neither the European Union nor the granting authority can be held
responsible for them.
\end{minipage}%
\hfill
\begin{minipage}[t]{0.13\linewidth}
\vspace{0.20em}
\includegraphics[height=0.78cm]{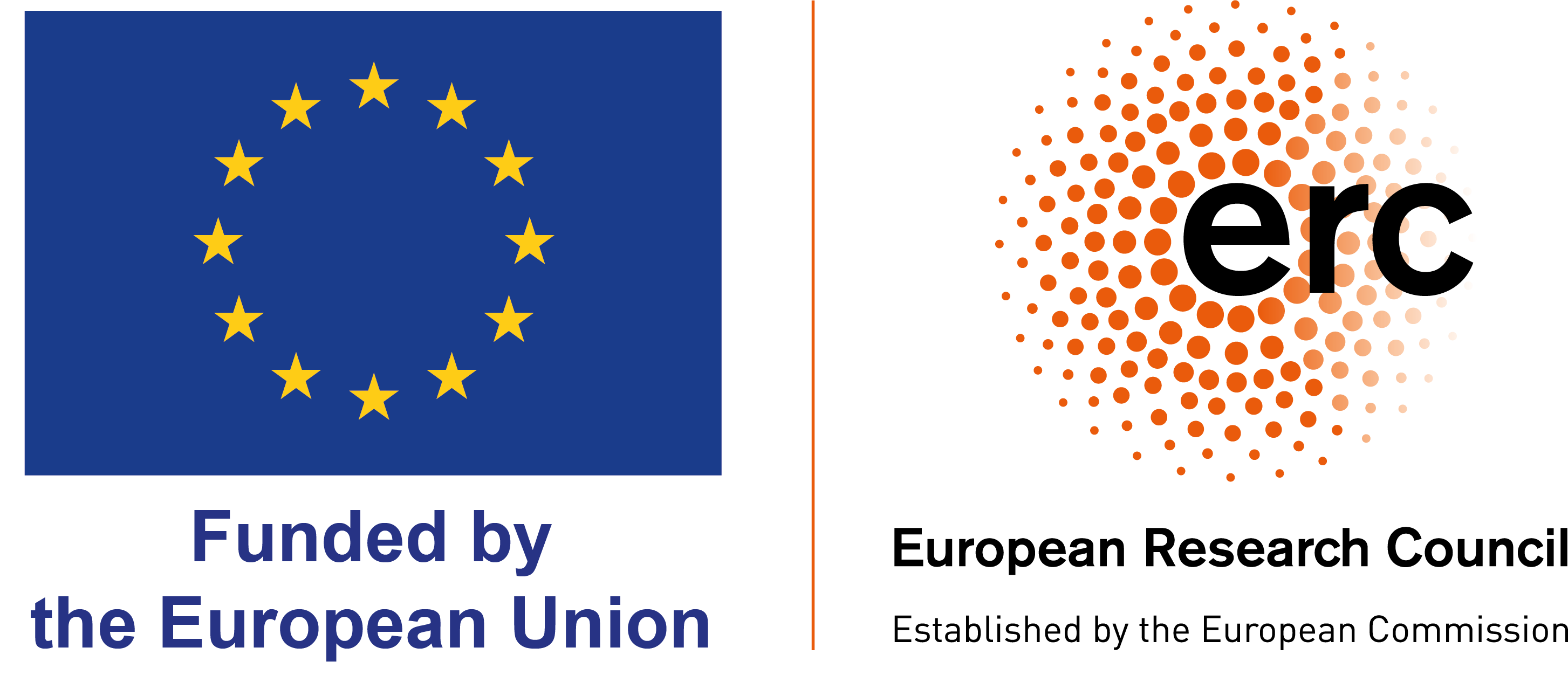}
\end{minipage}%
}
\endgroup
\makeatother

\begin{abstract}
Consider a subfield of the field of rational functions in several indeterminates. 
We present an algorithm that, given a set of generators of such a subfield, finds a simple generating set. 
We provide an implementation of the algorithm and show that it improves upon the state of the art both in efficiency and the quality of the results.
Furthermore, we demonstrate the utility of simplified generators through several case studies 
from different application domains, such as structural parameter identifiability.
The main algorithmic novelties include performing only partial Gr\"obner basis computation via sparse interpolation and efficient search for polynomials of a fixed degree in a subfield of the rational function field.
\end{abstract}

\section{Introduction}

Let $k$ be a field, and consider the field of rational functions $k(\bx)$, where $\bx = (x_1, \ldots, x_n)$ are indeterminates.
Intermediate subfields $k \subset \mathcal{E} \subset k(\bx)$ are ubiquitous in algebra and algebraic geometry.
Their theoretical properties have been extensively studied during the 20-th century, in particular, due to the central role they play in the L\"uroth problem~\cite{Beauville2016} and Hilbert 14-th problem~\cite{Nagata1959,Kuroda2005}.

A starting point for the computational study of the subfields of the rational function field is the fact that every such field can be generated over $k$ by finitely many elements~\cite[Ch.~VIII, Ex.~4]{Lang};
this contrasts sharply with the situation for intermediate subalgebras of the polynomial algebra $k \subset \mathcal{A} \subset k[\bx]$~\cite{Robbiano1990}.
Such a finite generating set is then a natural data structure for representing subfields of $k(\bx)$.
Within this framework, a number of algorithmic results have
been obtained in the 90-s~\cite{OllivierPhD,Sweedler1993,Kemper1996,MULLERQUADE1999143,MllerQuade2000}, in particular, a solution to the membership problem.
One of the key technical tools introduced first by Ollivier~\cite[Ch.~III, \S2]{OllivierPhD} and then by M\"uller-Quade and Steinwandt~\cite{MULLERQUADE1999143,MllerQuade2000} 
has been
certain polynomial ideals in $k(\bx)[\by]$ with $\by = (y_1, \ldots, y_n)$ which we call \emph{OMS ideals} (\Cref{sec:OMS}); they play a prominent role in this paper as well.
Another line of research 
has been
devoted to subfields of transcendence degree one which have a particularly nice structure due to the L\"uroth theorem~\cite{Binder1996,Gutierrez2001}.

More recent advances in computing with rational function fields are often motivated by particular applications such as computer vision~\cite{dissertation} or structural parameter identifiability of dynamical models~\cite{allident,stident}.
In the latter case, the field $\mathcal{E}$ of interest consists of quantities that
may be inferred from the input-output data of the model.
The existing algorithms for computing the generators of $\mathcal{E}$ typically produce extremely complicated expressions, which cannot be further analyzed and interpreted by a modeler.
Therefore, one of the key questions in this context is the question of \emph{generator simplification}, which can be stated as follows:
\begin{description}
  \item[Given:] $\bg \coloneqq (g_1,\ldots,g_m) \in k(\mathbf{x})$, such that $\mathcal{E} = k(\bg)$;
  \item[Find:] $\bh \coloneqq (h_1,\ldots, h_\ell) \in k(\mathbf{x})$, such that $\mathcal{E} = k(\mathbf{h})$ and $\mathbf{h}$ are {\em simpler} than~$\mathbf{g}$.
\end{description}
We will deliberately not give a formal definition of what being simple means, as we do not expect that there is a single mathematical definition capable of capturing the intuitive notions of simplicity used in applications.
In general, it may be preferable to have
a generating set of small cardinality, where each element is a sparse rational function of low degree with small coefficients.
To give some intuition on how the simplification may look like and what benefits it may bring, let us give an illustrative example.

\begin{example}
    Consider a triangle with side lengths $a, b, c$.
    Using Heron's formula, we can write the length of the altitude from $a$ as
    \[
    h_a^2 = \frac{4s(s - a)(s - b)(s - c)}{a^2}, \quad \text{where } s = \frac{a + b + c}{2},
    \]
    and analogous formulas hold for $h_b$ and $h_c$.
    By computing a simplified generating set of $\mathcal{E} := \QQ(h_a^2, h_b^2, h_c^2)$, we find that $\mathcal{E} = \QQ(a^2, b^2, c^2)$.
    Since the side lengths are always positive, we conclude that they can be determined from the altitude lengths confirming a well-known fact from elementary geometry.
\end{example}

Already the L\"uroth problem can be viewed as a version of the simplification question focusing solely on the cardinality of the generating set.
Specifically, any generating set of $k \subset \mathcal{E} \subset k(\bx)$ over $k$ must contain at least $\trdeg_k \mathcal{E}$ elements.
On the other hand, if $\operatorname{char} k = 0$, the primitive element theorem implies that $\trdeg_k \mathcal{E} + 1$ generators are always sufficient.
Whether or not using only $\trdeg_k \mathcal{E}$ generators is always possible is the L\"uroth problem, and the answer is negative starting from $n = 3$.
Furthermore, deciding whether $\trdeg_k \mathcal{E} + 1$ generators are necessary is equivalent to determining if a unirational variety is rational, which is an old open problem in algebraic geometry~\cite{smith1997rationalnonrationalalgebraicvarieties}.
For a related Noether problem, an approach based on searching for generators of small degree was used to give a computational solution in several interesting cases~\cite{Kemper1996}.

The question of simplification in its broader, more intuitive, form was already considered in~\cite{MULLERQUADE1999143} where it was proposed to use the coefficients of a reduced \grobner{} basis of the OMS ideal of $\mathcal{E}$ as a ``canonical generating set''.
In the context of structural identifiability, this approach was further refined~\cite{allident} by taking this canonical generating set and filtering out redundant generators.
The resulting algorithm has been implemented in Maple~\cite{allident}, and the implementation was later incorporated into 
a web-application~\cite{Ilmer2021}.
A construction reminiscent of OMS ideals was used in~\cite{combinations_combos_algo,combos,Meshkat2011} to search for a simple transcendence basis of the relative algebraic closure of $\mathcal{E}$ inside $k(\bx)$ also in connection to the structural identifiability problem.

The main contribution of the present paper is a new algorithm 
that
takes as input a set of generators of a subfield $\mathcal{E} \subset k(\bx)$ and returns a simplified 
set of generators, assuming $\operatorname{char} k = 0$.
We provide an open-source implementation of the algorithm for the case $k = \QQ$ in Julia.
Our approach also relies on OMS ideals but improves upon the state of the art 
both in terms of \emph{efficiency} and \emph{quality of simplification}.
Let us describe our contributions along these two axes separately.

From the \emph{efficiency} perspective, the key question is how to perform \grobner{} basis computations with OMS ideals.
These ideals belong to the ring $k(\bx)[\by]$ over the rational function field which makes computations with them expensive due to the intermediate expression swell.
Our algorithm follows the evaluation-interpolation approach: 
it computes
multiple \grobner{} bases by specializing $\bx$ and then reconstructs the coefficients by using sparse rational function interpolation~\cite{cuyt-lee,vdhl}.
There are two rationale behind this choice of methodology.

First, interpolation
is one of the standard ways to alleviate intermediate expression swell.
It has been 
successfully
employed for similar problems over the rational function field such as GCD computation~\cite{gcd-interpolation,vdhl} and linear system solving~\cite{monagan-parametric-linsys} 
and explored for
\grobner{} basis computation~\cite[Ch.~5]{Boku2016} and, recently, for 
generic parametrized polynomial system solving~\cite{corniquel:hal-05487354}.
Our approach includes further optimizations such as computing over finite fields and using tracing to speed up repeated \grobner{} basis computations (see \Cref{sec:tracing}).
The resulting implementation is competitive for computing \grobner{} bases in $\mathbb{Q}(\bx)[\by]$, significantly improving upon the prior interpolation-based
algorithms, but may be slower than other, not interpolation-based methods (see \Cref{sec:other_gb}).

And this is where the second key advantage of our method comes into play: it does not require computing the full \grobner{} basis. Using evaluation-interpolation allows us to stop 
early,
once only low-degree coefficients of the \grobner{} basis have been interpolated, and these coefficients are 
the ones
of interest in the simplification problem.
So, we gradually increase the degree of interpolation stopping once the computed coefficients generate the whole field. 
Thus, compared to the earlier approach from~\cite{allident}, instead of computing all the coefficients and then filtering out the large redundant ones, we now avoid computing them at the first place.
This strategy allows to reduce the number of 
evaluations
and, consequently, the runtime dramatically and solve the problems which are out of reach otherwise (see~\Cref{subsec:small-degree}).

We improve the \emph{quality of simplification} by creating a pool of potential generators (which contains the coefficients of the OMS ideal) and employing a refined heuristic for selecting the final result from this pool.
The main tool for adding extra candidates to the pool
is our algorithm for finding polynomial elements of the subfield up to a fixed degree, which is of independent interest.
While a theoretical algorithm for this task has appeared already in~\cite{MULLERQUADE1999143}, it required computing the full \grobner{} basis of the OMS ideal.
In the spirit of the evaluation-based approach we follow, our algorithm performs only \grobner{} basis computation 
on the specialized ideal
which makes it efficient for practical computation.

Since we do not fix a formal notion of simplicity, assessing the quality of simplification is a delicate task.
First, we run our algorithm and the one from~\cite{allident} on a large suite of benchmarks (53 problems) coming from applications to structural identifiability and perform a manual inspection of results.
We show that our algorithm is more efficient and argue that it produces simpler results
(see~\Cref{sec:comparison_simplification}).
Furthermore, we perform a more detailed analysis of the output of our algorithm on several case studies coming from structural identifiability of ODE models, observability of discrete dynamics, and fields of invariants.
In addition to syntactic simplicity, we discuss
the usefulness of the simplification performed by our algorithm.

The rest of the paper is organized as follows.
We review and adapt the main tools (e.g., sparse interpolation, OMS ideals) in~\Cref{sec:prelim}.
The main building blocks of our algorithm are presented in~\Cref{sec:blocks}, and we use them to describe and justify our main algorithm in~\Cref{sec:main_algo}.
\Cref{section:experimental} is devoted to implementation details including discussion of the impact of the main design choices on the performance.
We analyze the quality of simplification by our algorithm on a large benchmark suite in~\Cref{sec:experiment} (the details of the suite are summarized in~\Cref{app:simplify_everyone}).
We finish with a more in-depth discussion of a set of case studies coming from three different application domains in~\Cref{sec:app}.

\paragraph{Acknowledgements.}
The authors would like to thank Hoon Hong, Joe Kileel, Gr\'egoire Lecerf, Nicolette Meshkat, and Fran\c{c}ois Ollivier for helpful discussions.
The authors are also greatful to Jan Flusser and Evelyne Hubert for stimulating discussions of~\Cref{sec:app_invariants}.
AD has been supported by an ERC-2023-ADG grant for the ODELIX project (number 101142171) and partially supported by the NSF grant CCF-2212460.
GP was supported by the French ANR-22-CE48-0008 OCCAM and ANR-22-CE48-0016 NODE projects.

\section{Preliminaries}\label{sec:prelim}

\subsection{Notation}

Let $k$ denote an effective field, meaning that there exist algorithms for the arithmetic in $k$. 
For a field $k$, let $\overline{k}$
denote its algebraic closure.
We will use shorthands for tuples, such as $\bx \coloneqq (x_1,\ldots,x_n)$.
For a multivariate polynomial $f \in k[\bx]$, we will write $d_f \coloneqq \deg f$ for the total degree of $f$ and $s_f$ for the number of nonzero terms in $f$. 
We define the degree of zero polynomial to be $-\infty$. 
We have $\Zplus \coloneqq 0, 1, 2,\ldots$.

Products, sums, and powering of tuples are performed component-wise, more precisely, for every $\bgamma, \bomega, \bsigma \in k^n$ and every $\balpha \in \Zplus^n$, we write $\bgamma \bomega^{\balpha} + \bsigma \coloneqq (\gamma_1 \omega_1^{\alpha_1} + \sigma_1, \ldots, \gamma_n \omega_n^{\alpha_n} + \sigma_n)$. We extend this notation to scalars by treating a scalar as a tuple of $n$ identical elements. In place of $\bgamma, \bomega, \bsigma $ we may use indeterminates or tuples thereof. 
Given a tuple $\ba$ (say, $\ba \in k^n$), we write $\Pi(\ba) \coloneqq \prod_{i=1,\ldots,n} a_i$.

The degrevlex monomial ordering is a total ordering on the set of monomials, which we may denote by $<_{\operatorname{drl}}$. For $\bx^{\balpha}, \bx^{\bbeta}$ with ${\balpha},{\bbeta} \in \Zplus^n$, we have $\bx^{\balpha} <_{\operatorname{drl}} \bx^{\bbeta}$ whenever the sum of elements of $\balpha$ is smaller than that of $\bbeta$, or, them being equal, the right-most nonzero entry of ${\balpha} - {\bbeta}$ is positive.

For polynomials $f_1,\ldots,f_\ell$ from $k[\bx]$, let $\langle f_1,\ldots,f_\ell \rangle$ denote the ideal in $k[\bx]$ generated by $f_1,\ldots,f_\ell$.
For an ideal $I$ of $k[\bx]$ and an element $a \in k[\bx]$, denote $I\colon a^\infty \coloneqq \{r \in k[\bx] \mid \exists \ell: r a^\ell \in I\}$. This set is also an ideal.

\subsection{DeMillo-Lipton-Schwartz-Zippel lemma}

Let $k$ be an effective field and $\bx := (x_1, \ldots, x_n)$ be an $n$-tuple of variables. 
Let the coordinates of $\bgamma \coloneqq (\gamma_1,\ldots,\gamma_n)$ be chosen independently and uniformly from a finite subset $S \subseteq k$. 
Many of our probabilistic statements will rely on the following lemma:

\begin{lemma}[DeMillo-Lipton-Schwarz-Zippel~\cite{demillo-lipton,Schwartz,zippel-lemma}]
    \label{lemma:sz}
    Let $f \in k[\bx]$ be a nonzero polynomial of degree $d$.
    The probability that $f(\bgamma) = 0$ is at most $d/|S|$.
\end{lemma}

\begin{corollary}
    \label{corollary:diversification-new}
    Let $f \in k[\bx]$ be a nonzero polynomial of degree $d$.
    Let $\bsigma_1,\ldots,\bsigma_T \in k^n$ be arbitrary but fixed. Let $\bomega_1,\ldots,\bomega_T \in (k \setminus \{0\})^n$ be arbitrary but fixed.
    The probability that there exists
    $i = 1,\ldots,T$ 
    such that $f(\bomega_i \bgamma + \bsigma_i) = 0$ is at most $d T / |S|$.
\end{corollary}
\begin{proof}
    Let $p \coloneqq f(\bomega_1 \bx + \bsigma_1) \cdots f(\bomega_T \bx + \bsigma_T)$.
    Note $p \neq 0$, since for every $i=1,\ldots,T$ none of the components of $\bomega_i$ are zero. We apply~\Cref{lemma:sz} to $p$. 
\end{proof}

\subsection{Univariate rational function interpolation}
\label{sec:cauchy}

Let $F = A/B$ be a rational function in $k(x)$ with $A,B \in k[x]$, $\operatorname{gcd}(A,B) = 1$, and $d_A \coloneqq \deg A$ and $d_B \coloneqq \deg B$.
Let $D_A \geqslant d_A$ and $D_B \geqslant d_B$ be the bounds on the total degrees of $A$ and $B$, respectively.
Let $D$ be such that $D \geqslant D_A + D_B + 2$, and let $u_1,\ldots,u_D \in k$ be pairwise different
such that for every $i=1,\ldots,D$ we have $B(u_i) \neq 0$ (assuming $k$ is sufficiently large).
Univariate rational function interpolation is the problem of recovering $F$ from the evaluations $F(u_1),\ldots,F(u_D)$ and the bounds $D_A, D_B$.

The solution is unique up to a scaling of $A$ and $B$, and can be computed using univariate polynomial interpolation and the Extended Euclidean Algorithm~\cite[Corollary 5.18]{MCA}. This method is also referred to as Cauchy interpolation.

\begin{remark}
    When the bounds on the degrees $D_A$ and $D_B$ are not known, it is still possible to recover a solution using~\cite[Algorithm MQRFR]{mqrfr} instead of the EEA. Later in this paper, these two methods may be used interchangeably.
\end{remark}


\subsection{Sparse polynomial interpolation}
\label{sec:tiwari}

Let $f$ be a polynomial in $k[\bx]$ with $\bx \coloneqq (x_1,\ldots,x_n)$. Let $d_f$ be an integer, the total degree of $f$.
Sparse interpolation is the problem of recovering the sparse representation of $f$ in the form
\begin{equation}
\label{eq:sparse-repr-poly}
\begin{array}{l}
f = \sum\limits_{i=1,\ldots,s_f} a_i \Pi(\bx^{\balpha_i})
\end{array}    
\end{equation}
from a sufficient number of evaluations of $f$, where $s_f$ is the number of terms in $f$ and, for every $i = 1, \ldots, s_f$, we have $a_i \in k \setminus \{0\}$ and $\balpha_i \in \Zplus^n$ 
are pairwise different.

In the case an upper bound on the number of terms $s_f$ is known, the algorithm by Ben-Or and Tiwari~\cite{ben-or-tiwari} (which generalizes the Prony's method~\cite{Prony}) recovers a sparse representation of $f$ in the form~\eqref{eq:sparse-repr-poly} deterministically from the evaluations $f(\bomega^0),f(\bomega^1),\ldots,f(\bomega^{2T-1})$ with $T \geqslant s_f$ for some carefully chosen~$\bomega \in (k \setminus \{0\})^n$. We shall call such $\bomega$ an {\em admissible ratio}, following~\cite{demin-vdhl-factoring}. 
A variant of the algorithm by Ben-Or and Tiwari is summarized in \Cref{alg:interpolation-poly}.

There exist several strategies for choosing an admissible ratio $\bomega$. A common requirement is that for every $\balpha \in \Zplus^n$ such that the degree of $\Pi(\bx^{\balpha})$ does not exceed $d_f$ it is possible to determine the exponent 
$\balpha$
from the product $\Pi(\bomega^{\balpha})$.
In practice, when $k=\mathbb{Q}$ or when $k$ is a finite field of large characteristic, one may take $\bomega \coloneqq (\omega_1,\ldots,\omega_n) = (2,3,5,\ldots,p_n)$, where $p_n$ denotes the $n$-th prime number. In this case, the exponent can be recovered via integer factorization (more precisely, via trial division). 
Other approaches include reduction to univariate interpolation by Kronecker substitution~\cite[Section 8.4]{MCA} and picking the evaluation point of a large order in a suitable extension of $k$~\cite{vdhl-interpol}.

\begin{algorithm}[H]
\caption{Deterministic sparse polynomial interpolation~\cite[Figure 1]{ben-or-tiwari}, \cite[Algorithm 4]{vdhl-interpol}}
\label{alg:interpolation-poly}
\begin{description}[itemsep=0pt]
\item[Input:] an admissible ratio $\bomega \in (k \setminus \{0\})^n$; the evaluations $f(\bomega^0),f(\bomega^1),\ldots,f(\bomega^{2 T - 1})$.
\item[Output:] if $T \geqslant s_f$, the sparse representation of $f$ as in~\eqref{eq:sparse-repr-poly}; if $T < s_f$, an incorrect result.
\end{description}

\begin{enumerate}[label = \textbf{(Step~\arabic*)}, leftmargin=*, align=left, labelsep=\favoritelabelsep, itemsep=\favoriteitemsep]
    \item Compute (via P{\'a}de approximation~\cite[Section 5.9]{MCA}) the minimal $t \leqslant T$, a monic $\Lambda \in k[z]$ of degree $t$, and $N \in k[z]$ of degree less than $t$, such that the following identity holds modulo $O(z^{2T})$:
    \[
    \sum\limits_{i =0}^{\infty} f(\bomega^i) z^i 
    = \frac{N(z)}{\Lambda(z)}.
    \]
    \item Find the roots $\Pi(\bomega^{-\balpha_1}),\ldots,\Pi(\bomega^{-\balpha_t})$
    of $\Lambda(z)$ using polynomial root-finding.
    \item Compute the exponents 
    $\balpha_1,\ldots,\balpha_t$ 
    using integer factoring.
    \item Compute the coefficients $a_1,\ldots,a_t$ by solving the linear system (for example, using~\cite{transposed_vandermonde}):
    \[
    \begin{aligned}
    a_1 + \ldots + a_t &= f(\bomega^0)\\
    &\ldots\\
    a_1 \Pi(\bomega^{(t-1) \balpha_1}) + \ldots + a_t \Pi(\bomega^{(t-1) \balpha_t}) &= f(\bomega^{t-1})
    \end{aligned}
    \]
    \item {\bf Return} $\balpha_1,\ldots,\balpha_{t}$ and $a_1,\ldots,a_{t}$.
\end{enumerate}
\end{algorithm}

\begin{remark}
    The fact that a field is effective does not immediately imply that there is an algorithm for finding root of univariate polynomials over it.
    Over the rationals and over finite fields there exist efficient deterministic algorithms~\cite[Figure 2]{ben-or-tiwari},~\cite{cantor-zassenhaus-81},~\cite[Section 14.5.]{MCA},~\cite{gvdhl-det-roots}.
\end{remark}

\begin{remark}
    Computing an upper bound on the number of terms $s_f$ deterministically is a hard problem; see the discussion in~\cite[Section 8]{ben-or-tiwari}. For this reason, in~\Cref{alg:interpolation-poly}, when $T < s_f$, we allow an incorrect result instead of a proper {\tt FAIL}. An upper bound can be computed 
    probabilistically~\cite[Section 4.6]{vdhl-interpol}.
\end{remark}

\begin{remark}[Diversification~\cite{diversification}]
\label{remark:diversification}
In the following sections, we will be interested in estimating the probability that for a fixed nonzero $\tilde{f}(\bx) \coloneq f(\bx + \bsigma)$ any of the evaluations $\tilde{f}(\bomega_1),\ldots,\tilde{f}(\bomega_T)$ vanish at a fixed $\bomega_1,\ldots,\bomega_T \in (k\setminus \{0\})^{n}$ and $\bsigma \in k^n$. In theoretical arguments, we may replace $\tilde{f}(\bx)$ with $\tilde{f}(\bgamma \bx)$ and apply \Cref{corollary:diversification-new}, where the coordinates of $\bgamma \in k^n$ are chosen independently and uniformly from $S \subseteq k$. In practice, this amounts to evaluating $\tilde{f}(\bgamma \bomega_1),\ldots,\tilde{f}(\bgamma \bomega_T)$; when the end goal is to interpolate the original $\tilde{f}(\bx)$ from these evaluations, we would additionally require that the coordinates of $\bgamma$ are nonzero, so that the original polynomial could be restored by applying the change of variables $\bx \mapsto \bx / \bgamma$.
\end{remark}

\subsection{Sparse rational function interpolation}
\label{sec:prelim:ratinterpol}

Let $F$ be a rational function in $k(\bx)$ with $\bx \coloneqq (x_1,\ldots,x_n)$. 
Sparse interpolation is the problem of recovering the sparse representation of $F$ in the form
\begin{equation}
\label{eq:sparse-repr}
\begin{array}{lll}
F = \dfrac{A}{B}, & A = \sum\limits_{i=1,\ldots,s_A} a_i 
\Pi(\bx^{\balpha_i}),  & B = \sum\limits_{i=1,\ldots,s_B} b_i 
\Pi(\bx^{\bbeta_i}),\\
\end{array}    
\end{equation}
from a sufficient number of evaluations of $F$.
Here for every $i=1,\ldots,s_A$ we have $a_i \in k \setminus \{0\}$ and ${\balpha}_i \in \Zplus^n$ are pairwise different and for every $j=1,\ldots,s_B$ we have $b_j \in k \setminus \{0\}$ and ${\bbeta}_j \in \Zplus^n$ are pairwise different.
We also 
assume that
$\operatorname{gcd}(A, B) = 1$.
Let $d_A$ and $d_B$ be the total degrees of $A$ and $B$, respectively. 

Assume we have access to a blackbox that evaluates $F$ at a point in $k^n$.
The algorithm by van der Hoeven and Lecerf~\cite[Section 3.3.]{vdhl} is a probabilistic
algorithm for recovering the sparse representation of a rational function as in~\eqref{eq:sparse-repr} from a sequence of evaluations of the blackbox. 
We will show that the probability of failure in this algorithm can be made arbitrarily close to zero assuming we can sample elements independently and uniformly from an arbitrarily large subset $S \subseteq k \setminus \{0\}$. 
For completeness, we describe the algorithm here.

Introduce a new indeterminate $x_0$. Consider the rational function
\[
\hat{F}(x_0,\ldots,x_n) \coloneqq x_0^{d_A-d_B}\frac{A(x_1/x_0, \ldots, x_n / x_0)}{B(x_1/x_0, \ldots, x_n / x_0)},
\]
where we also introduce for convenience
\[
\begin{array}{l}
\hat{A} \coloneqq x_0^{d_A} A\left(x_1 / x_0, \ldots, x_n / x_0\right),\\
\hat{B} \coloneqq x_0^{d_B} B\left(x_1 / x_0, \ldots, x_n /x_0\right).
\end{array}
\]
Note that $\hat{A}$ and $\hat{B}$ are homogeneous and coprime. 

Introduce a new indeterminate $u$. Let $\bsigma \coloneqq (\sigma_0, \ldots, \sigma_n) \in k^{n+1}$, and let
\[
\begin{array}{l}
\tilde{F}_{\bsigma}(x_0, \ldots, x_n, u) \coloneqq \hat{F}(x_0 u + \sigma_0, \ldots, x_n u + \sigma_n),\\
\tilde{A}_{\bsigma}(x_0, \ldots, x_n, u) \coloneqq \hat{A}(x_0 u + \sigma_0, \ldots, x_n u + \sigma_n) \coloneqq \tilde{A}_0 + \tilde{A}_1u + \ldots +\tilde{A}_{d_A}u^{d_A},\\
\tilde{B}_{\bsigma}(x_0, \ldots, x_n, u) \coloneqq \hat{B}(x_0 u + \sigma_0, \ldots, x_n u + \sigma_n) \coloneqq \tilde{B}_0 + \tilde{B}_1u + \ldots +\tilde{B}_{d_B}u^{d_B},
\end{array}
\]
where for $i=0,\ldots,d_A$ we have $\tilde{A}_i \in k[x_0,\ldots,x_n]$ and for $i=0,\ldots,d_B$ we have $\tilde{B}_i \in k[x_0,\ldots,x_n]$.
If $\tilde{B}_0\neq0$, then $\tilde{A}_{\bsigma}$, $\tilde{B}_{\bsigma}$ are coprime by~\cite[Lemma~1]{vdhl}.
In the following, assume $\tilde{B}_0 \neq 0$.  

Let $\bomega \coloneqq (\omega_0, \ldots, \omega_n) \in (k \setminus \{0\})^{n+1}$ be an admissible ratio (as in~\Cref{sec:tiwari}). 
For $i \in \Zplus$, introduce notation:
\[
\begin{array}{l}
\tilde{F}_{\bsigma}^{[i]}(u) \coloneqq \tilde{F}_{\bsigma}(\omega_0^i,\ldots,\omega_n^i, u),\\
\tilde{A}_{\bsigma}^{[i]}(u) \coloneqq \tilde{A}_{\bsigma}(\omega_0^i,\ldots,\omega_n^i, u) \coloneqq \tilde{A}_0^{[i]} + \tilde{A}_1^{[i]} u + \ldots +\tilde{A}_{d_A}^{[i]} u^{d_A},\\
\tilde{B}_{\bsigma}^{[i]}(u) \coloneqq \tilde{B}_{\bsigma}(\omega_0^i,\ldots,\omega_n^i, u) \coloneqq \tilde{B}_0^{[i]} + \tilde{B}_1^{[i]}u + \ldots +\tilde{B}_{d_B}^{[i]} u^{d_B}.
\end{array}
\]
 Assume $\tilde{A}_{\bsigma}^{[i]}(u)$ and$\tilde{B}_{\bsigma}^{[i]}(u)$
are coprime. Let $D \coloneqq d_A + d_B + 2$. 
Let $u_1,\ldots,u_D \in k \setminus \{0\}$ be pairwise different. 
If $d_A$ and $d_B$ are known, then it is possible to obtain $\tilde{F}_{\bsigma}^{[i]}(u)$
from the evaluations 
$\tilde{F}_{\bsigma}^{[i]}(u_1), \ldots, \tilde{F}_{\bsigma}^{[i]}(u_D)$
using univariate rational interpolation (as in Section~\ref{sec:cauchy}), assuming none of the evaluations fail.
Then, for a sufficiently large $T$, it is possible to 
compute $\tilde{A}_{d_A}/\tilde{B}_0$ from  
$\tilde{A}_{d_A}^{[0]}/\tilde{B}_0,\ldots,\tilde{A}_{d_A}^{[2T-1]}/\tilde{B}_0$
using sparse multivariate polynomial interpolation (as in Section~\ref{sec:tiwari}).
Division by $\tilde{B}_0$ ensures that the trailing coefficient of the denominator of 
$\tilde{F}_{\bsigma}^{[i]}(u)$
is normalized to one across all $i=0,\ldots,2T-1$.
Similarly, it is possible to 
compute $\tilde{B}_{d_B}/\tilde{B}_0$.

Finally, note that $\tilde{A}_{d_A}, \tilde{B}_{d_B}$ coincide with $\hat{A}, \hat{B}$, since $\hat{A}$ and $\hat{B}$ are homogeneous. 
Then, by setting $x_0 \coloneq 1$, we obtain the sought 
$A = \hat{A}(1,x_1,\ldots,x_n)$ and $B = \hat{B}(1,x_1,\ldots,x_n)$.

\begin{algorithm}[H]
\caption{Sparse interpolation of a rational function given as a blackbox,~\cite[Section 3.3.]{vdhl}}
\label{alg:interpolation}
\begin{description}[itemsep=0pt]
\item[Input:] a blackbox $\blackbox$, such that for every $\ba \in k^n$ either $\blackbox(\ba) = F(\ba)$ or $\blackbox(\ba)$ returns {\tt FAIL}; a subset $S \subseteq k \setminus \{0\}$; 
integers $\delta_A, \delta_B$;
an admissible ratio $\bomega \in (k\setminus \{0\})^{n+1}$.

\item[Output:] $P,Q \in k[\bx]$ with $P/Q=F$ and $\operatorname{gcd}(P,Q)=1$, or {\tt FAIL}.
\end{description}

\begin{enumerate}[label = \textbf{(Step~\arabic*)}, leftmargin=*, align=left, labelsep=\favoritelabelsep, itemsep=\favoriteitemsep]
    \item Let $D \coloneqq \delta_A + \delta_B +2$.
    
    \item Choose the coordinates of $\bsigma \coloneqq (\sigma_0,\ldots,\sigma_n)$ independently and uniformly from $S \subseteq k \setminus \{0\}$;
    choose pairwise different
    $u_1,\ldots,u_{D} \in k \setminus \{0\}$.
    
    \item For $T = 1,2,4,\ldots$ do
    \begin{enumerate}[ref= \theenumi (\alph*)]
        \item\label{step:eval} For $i = 0, \ldots, 2T-1$ do
        \begin{enumerate}[ref= \theenumii(\roman*), label=(\roman*)]
        \item Compute 
        $\tilde{F}_{\bsigma}^{[i]}(u_1), \ldots, \tilde{F}_{\bsigma}^{[i]}(u_D)$ by evaluating the blackbox $\blackbox$ at the appropriate points; if any of the evaluations fail, \textbf{return} {\tt FAIL}.
        \item \label{step:eval-rfi} 
        Compute $\tilde{F}_{\bsigma}^{[i]}(u)$
        by applying univariate rational interpolation to $\tilde{F}_{\bsigma}^{[i]}(u_1), \ldots, \tilde{F}_{\bsigma}^{[i]}(u_D)$
        and $(D_A, D_B) \coloneqq (\delta_A, \delta_B)$ (as in \Cref{sec:cauchy}); if the interpolation fails, \textbf{return} {\tt FAIL}.
        \end{enumerate}
        \item\label{step:poly-interp} Let $P \in k[x_0,\ldots,x_n]$ 
        (resp., $Q$) 
        be the output of sparse interpolation (\Cref{alg:interpolation-poly}) applied to 
        $\tilde{A}_{d_A}^{[0]}/\tilde{B}_{0}, \ldots, \tilde{A}_{d_A}^{[2T-1]}/\tilde{B}_{0}$
        (resp., 
        $\tilde{B}_{d_B}^{[0]}/\tilde{B}_{0}, \ldots, \tilde{B}_{d_B}^{[2T-1]}/\tilde{B}_{0}$).
        \item\label{step:rat-check} If $(\delta_A,\delta_B) = (\operatorname{deg}P,\operatorname{deg}Q)$ and
        $P(\bsigma) = Q(\bsigma)\hat{F}(\bsigma)$, where $\hat{F}({\bsigma})$ is computed by evaluating the blackbox $\blackbox$, then \textbf{return} 
        $P(1,x_1,\ldots,x_n)$ and $Q(1,x_1,\ldots,x_n)$.
        \item\label{step:duct-tape} If $T > \binom{n+ \delta_A + \delta_B}{n}$, then {\bf return} {\tt FAIL}.
    \end{enumerate}
\end{enumerate}
\end{algorithm}

\begin{remark}
    One difference between the presentation here and in~\cite[Section 3.3.]{vdhl} is that here we allow for the input 
    $(\delta_A, \delta_B)$
    to be different from the actual $(d_A, d_B)$ 
    (and the algorithm will 
    either fail of produce an incorrect result
    in this case). 
    The condition in~\ref{step:duct-tape} is to make sure that the algorithm terminates when $\delta_A \neq d_A$ or $\delta_B \neq d_B$.
\end{remark}

\begin{proposition}
\label{prop:proba-interpol}
Let $g \in k[\bx]$ be a nonzero polynomial such that for all $\ba \in k^n$ where $\blackbox(\ba)$ returns {\tt FAIL} we have $g(\ba) = 0$. Let $h \coloneqq \operatorname{deg} g$.
Let $s \coloneq \operatorname{max}(s_A, s_B)$ and $D \coloneqq d_A + d_B + 2$. 

\Cref{alg:interpolation} terminates for any inputs. 
If $\delta_A = d_A$ and $\delta_B = d_B$, then \Cref{alg:interpolation} 
does not fail and returns a correct result
with probability at least
\[
1 - \frac{D(\lceil \log_2 s \rceil + 8s + 4sh + 8sD)}{|S|}.
\]
\end{proposition}

\begin{proof}
Correctness is proven in~\cite[Theorem 2]{vdhl}. 
For the termination claim, note that the algorithm terminates for any inputs in~\ref{step:duct-tape} for large enough $T$, in the case it did not terminate prior to that.

We treat probability in the case $\delta_A = d_A$ and $\delta_B = d_B$. 
Following~\Cref{remark:diversification}, we shall replace $\tilde{F}_{\bsigma}(x_0, \ldots,x_n, u)$ with $\tilde{F}_{\bsigma}(\gamma_0 x_0, \ldots, \gamma_n x_n, u)$, where the coordinates of $\bgamma \coloneqq (\gamma_0,\ldots,\gamma_n)$ are chosen independently and uniformly from $S \subseteq k \setminus \{0\}$.
The algorithm terminates no later than the first iteration when $T \geqslant s$. 
Therefore, upon termination, we have $T \leqslant 2s$.
\Cref{alg:interpolation} returns {\tt FAIL} only in one of the following cases:
\begin{itemize}
    \item We have $\tilde{B}_0 = 0$. 
    This is equivalent to $\hat{B}(\sigma_0,\ldots,\sigma_n) = 0$. This can be checked when performing rational function interpolation in~\ref{step:eval-rfi}, assuming the interpolation does not fail. 
    Note $\deg \hat{B} \leqslant D$. Therefore, the probability that $\hat{B}(\sigma_0,\ldots,\sigma_n) = 0$ does not exceed $D / |S|$ by~\Cref{lemma:sz}.
    \item For any of $i=0,\ldots, 2T-1$, any of the evaluations 
    $\tilde{F}_{\bsigma}^{[i]}(u_1), \ldots, \tilde{F}_{\bsigma}^{[i]}(u_D)$
    fail. 
    For every $i=0,\ldots, 2T-1$ and $j=1,\ldots,D$, evaluating 
    $\tilde{F}_{\bsigma}^{[i]}(u_j)$
    is equivalent to evaluating $\blackbox(\bxi^{[i,j]})$ with
    \[
   \bxi^{[i,j]} \coloneqq \left(\frac{\gamma_1 \omega_1^i u_j + \sigma_1}{\gamma_0 \omega_0^i u_j + \sigma_0}, \ldots, \frac{\gamma_n \omega_n^i u_j + \sigma_n}{\gamma_0 \omega_0^i u_j + \sigma_0}\right) \in k^n.
    \]
    Failure occurs when either of the two happens: a division by zero when constructing $\bxi^{[i,j]}$, or the blackbox returning failure. The probability of the former is at most $4 s D / |S|$, which can be seen by applying the 
    Schwartz-Zippel lemma~(\Cref{lemma:sz}) to $p \in k[\sigma]$ with $\deg p \leqslant 4 s D$ and the evaluation point $\sigma_0$, where
    \[
    p \coloneqq \prod_{\substack{i=0,\ldots,2T-1\\j=1,\ldots,D}} (\gamma_0 \omega_0^i u_j + \sigma).
    \]
    For the latter, recall that for all $\ba \in k^n$ where the blackbox fails we have $g(\ba) = 0$ with $h = \deg g$. We apply~\Cref{corollary:diversification-new} to $g$ and $\bxi^{[i,j]}$, 
    which is justified since none of $\omega_1,\ldots,\omega_n$ and $u_1,\ldots,u_D$ are zero.
    Whence, and because of $2T \leqslant 4s$, the probability of failure does not exceed $4 s h D / |S|$.
    \item For any of $i=0,\ldots,2T-1$, rational univariate interpolation of 
    $\tilde{F}_{\bsigma}^{[i]}(u)$
    fails. 
    This can be detected by comparing the degrees of the results against $(\delta_A,\delta_B)$.
    As shown in~\cite[Section 3.3.]{vdhl}, there exists a nonzero polynomial $R_{\bsigma} \in k[x_0,\ldots,x_n]$ with $\operatorname{deg} R_{\bsigma} \leqslant 2D^2$ such that for every $i=0,\ldots,2T-1$ rational univariate interpolation of 
    $\tilde{F}_{\bsigma}^{[i]}(u)$
    fails whenever $R_{\bsigma}(\bgamma \bomega^i) = 0$. 
    Since $2T \leqslant 4 s$, 
    the probability of a failure does not exceed $8 s D^2 / |S|$ by~\Cref{corollary:diversification-new}. 
\end{itemize}

Provided~\Cref{alg:interpolation} does not return {\tt FAIL}, result may be incorrect. This happens when the check in~\ref{step:rat-check} passes for incorrect $P,Q$. 
Note $P, Q$ are correct if and only if $\operatorname{deg} P = \operatorname{deg}A$ and $\operatorname{deg} Q = \operatorname{deg} B$ and $P = Q\hat{F}$.
 Therefore, it suffices to estimate the probability that $P\hat{B} - Q\hat{A} \neq 0$ vanishes at $\bsigma$. We apply the Schwartz-Zippel lemma~(\Cref{lemma:sz}),
which is well-defined because
the evaluation $\hat{F}(\bsigma)$ does not fail since $\tilde{B}_0 \neq 0$.
There are at most $\lceil \log_2 s \rceil$ iterations where the Schwarz-Zippel lemma is applied, thus the probability of producing an incorrect result is bounded by $D\lceil \log_2 s \rceil / |S|$.
Putting everything together, we arrive at the desired probability.
\end{proof}

\subsection{The degree of a blackbox function}

Let $F = A/B$ be a rational function in $k(\bx)$ with $A,B \in k[\bx]$ and $\operatorname{gcd}(A,B) = 1$, with $\bx = (x_1,\ldots,x_n)$. 
As before, assume we have access to a blackbox that evaluates $F$ at a point in $k^n$. Let $d_A \coloneqq \deg A$ and $d_B \coloneqq \deg B$. We describe an algorithm for computing $d_A$ and $d_B$.

In the notation of the previous section, let $\bsigma \coloneqq (\sigma_0,\ldots,\sigma_n) \in k^{n+1}$ be arbitrary and let
\[
R_{\bsigma} \coloneqq \operatorname{Res}_u (\tilde{A}_{\bsigma}(x_0, \ldots,x_n, u), \tilde{B}_{\bsigma}(x_0, \ldots, x_n, u)),
\]
where $\operatorname{Res}_u$ denotes the resultant with respect to the variable $u$.
Note this is the same $R_{\bsigma}$ that appears in the proof of~\Cref{prop:proba-interpol}.
By~\cite[Theorem 2]{vdhl}, $R_{\bsigma}$ is not identically zero.
By definition, for every $(a_0,\ldots,a_n) \in k^{n+1}$, the degrees of the numerator and denominator of $\tilde{F}_{\bsigma}(a_0,\ldots,a_n,u)$ coincide with that of $F$ if and only if $R_{\bsigma}(a_0,\ldots,a_n) \neq 0$. 

For convenience, introduce notation
\[
\tilde{f}_{\bsigma}(u) \coloneqq \tilde{F}_{\bsigma}(1,\ldots,1, u).
\]
When $R_{\bsigma}(1,\ldots,1) \neq 0$, the degrees of the numerator and denominator of $\tilde{f}_{\bsigma}$ and $\tilde{F}_{\bsigma}$ coincide, and are equal to $d_A$ and $d_B$, respectively.
Let $T\geqslant \operatorname{max}(d_A,d_B)$ and let $u_1,\ldots,u_{2T+2} \in k \setminus \{0\}$ be pairwise different.
Then, we 
can
recover $\tilde{f}_{\bsigma}$ by applying univariate rational interpolation to the evaluations $\tilde{f}_{\bsigma}(u_1),\ldots,\tilde{f}_{\bsigma}(u_{2T+2})$ and the degree bounds $D_A \coloneqq D_B \coloneqq T$ (as in \Cref{sec:cauchy}). 
This leads to the following algorithm.

\begin{algorithm}[H]
\caption{The degree of a rational function given as a blackbox}
\label{alg:estimate_degrees}
\begin{description}[itemsep=0pt]
\item[Input:] a blackbox $\blackbox$, such that for every $\ba \in k^n$ either $\blackbox(\ba) = F(\ba)$ or $\blackbox(\ba)$ returns {\tt FAIL}; an integer $d$; a subset $S \subseteq k \setminus \{0\}$.
\item[Output:] 
if $d_A + d_B \leqslant d$, the integers $d_A, d_B$ or {\tt FAIL}; if $d_A + d_B > d$, {\tt STOPPED} or {\tt FAIL}.
\end{description}

\begin{enumerate}[label = \textbf{(Step~\arabic*)}, leftmargin=*, align=left, labelsep=\favoritelabelsep, itemsep=\favoriteitemsep]

    \item Choose the coordinates of $\bsigma \coloneqq (\sigma_0,\ldots,\sigma_n)$ independently and uniformly from $S \subseteq k \setminus \{0\}$. 
    
    \item Choose $u_0$ uniformly from $S \subseteq k \setminus \{0\}$. If the evaluation $\tilde{f}_{\bsigma}(u_0)$ fails (computed by evaluating the blackbox $\blackbox$), {\bf return} {\tt FAIL}.
        
    \item For $T = 0, 1,  2,\ldots$ do
    \begin{enumerate}[ref= \theenumi (\alph*)]
        \item \label{step:duct-tape-2} If $T > d$, then {\bf return} {\tt STOPPED}.
        \item Compute $\tilde{f}_{\bsigma}(u_1),\ldots,\tilde{f}_{\bsigma}(u_{2T+2})$ for pairwise different $u_1,\ldots,u_{2T+2} \in k \setminus \{0\}$ by evaluating the blackbox $\blackbox$ at the appropriate points; if any of the evaluations fail, {\bf return} {\tt FAIL}.
        \item Compute $p/q \in k(u)$ with $\operatorname{gcd}(p,q)=1$ by applying univariate rational interpolation to 
        $\tilde{f}_{\bsigma}(u_1),\ldots,\tilde{f}_{\bsigma}(u_{2T+2})$ with the degree bounds $D_A \coloneqq D_B \coloneqq T$ (as in \Cref{sec:cauchy}).
        \item\label{step:rat-check-univ} If 
        $p(u_0) = q(u_0)\tilde{f}_{\bsigma}(u_0)$, where $\tilde{f}_{\bsigma}(u_0)$ is computed by evaluating the blackbox $\blackbox$, then \textbf{return}
        $d_p, d_q$ if $d_p + d_q \leqslant d$ and {\tt STOPPED}, otherwise.
    \end{enumerate}
\end{enumerate}
\end{algorithm}

\begin{proposition}
\label{prop:estimate_degrees}
Let $g \in k[\bx]$ be a nonzero polynomial such that for all $\ba \in k^n$ where $\blackbox(\ba)$ returns {\tt FAIL} we have $g(\ba) = 0$. Let $h \coloneqq \operatorname{deg} g$. Let $D \coloneqq d_A + d_B + 2$. 

\Cref{alg:estimate_degrees} terminates for any inputs.
The probability that~\Cref{alg:estimate_degrees} 
does not fail and returns a correct result is at least
\[
1 - \frac{D(2 + 3 h + 3 D)}{|S|}.
\]
\end{proposition}

\begin{proof}
    Termination is clear due to \ref{step:duct-tape-2}.
    
    We prove correctness and probability.
    Following~\Cref{remark:diversification}, we replace $\tilde{F}_{\bsigma}(x_0, \ldots,x_n, u)$ with $\tilde{F}_{\bsigma}(\gamma_0 x_0, \ldots, \gamma_n x_n, u)$, where the coordinates of $\bgamma \coloneqq (\gamma_0,\ldots,\gamma_n)$ are chosen independently and uniformly from $S \subseteq k \setminus \{0\}$.
    Note that at every iteration we have $T + 1 \leqslant D$. In the case $d_A + d_B > d$, this is clear; in the case $d_A + d_B \leqslant d$, this is because the algorithm terminates no later than the first iteration when $T = \operatorname{max}(d_A, d_B)$.

    The algorithm returns {\tt FAIL} only in one of the following cases:
    \begin{itemize}
        \item Any of the evaluations $\tilde{f}_{\bsigma}(u_1), \ldots, \tilde{f}_{\bsigma}(u_{2T+2})$ fail. Similarly to the proof of~\Cref{prop:proba-interpol}, failure occurs either because of division by zero when constructing the evaluation point or when the evaluation of the blackbox returns fail. 
  
        The probability of these does not exceed $(2 D + 2 h D) / |S|$ by~\Cref{corollary:diversification-new}.
        \item The evaluation $\tilde{f}_{\bsigma}(u_0)$ fails. By the Schwartz-Zippel lemma (\Cref{lemma:sz}), the probability of this does not exceed $h / |S|$.
    \end{itemize}

    In the following, assume the algorithm did not fail.
    We treat the case $d_A + d_B > d$. Incorrect result is returned only if the check in \ref{step:rat-check-univ} ever passes. If $R_{\bsigma}(1,\ldots,1) \neq 0$, then at every iteration $p - q \tilde{f}_{\bsigma} \neq 0$. 
    We appeal to \Cref{remark:diversification}, so it suffices to estimate the probability that either 
    $R_{\bsigma}(\gamma_0,\ldots,\gamma_n)$
    vanishes or that at any iteration $p - q \tilde{f}_{\bsigma}$ vanishes at $u_0$.
    The former is at most $2 D^2 / |S|$ by the Schwartz-Zippel lemma (\Cref{lemma:sz}), since $\operatorname{deg} R_{\bsigma} \leqslant 2D^2$.
    The latter is at most $D^2/|S|$, since we apply the Schwartz-Zippel lemma (\Cref{lemma:sz}) at most $D$ times with the same $u_0$, each time to a polynomial of degree at most $D$.

    When the algorithm does not fail and does not return an incorrect result, the output is correct. Thus, we obtain the desired probability of correctness in the case $d_A + d_B > d$. The case $d_A + d_B \leqslant d$ is similar: the output in \ref{step:rat-check-univ} is guaranteed to be correct at the first iteration when $T = \operatorname{max}(d_A, d_B)$ in the case $R_{\bsigma}(1,\ldots,1) \neq 0$ and assuming the algorithm did not fail and did not terminate with incorrect result.
\end{proof}

\subsection{OMS Ideals}\label{sec:OMS}

Let $k(\mathbf{x})$ be the field of rational functions in the indeterminates $\mathbf{x} \coloneq (x_1,\ldots,x_n)$. 
We will be interested in subfields $k \subset \mathcal{E} \subseteq k(\mathbf{x})$.
By~\cite[Ch.~VIII, Ex.~4]{Lang} every such subfield is finitely generated over $k$, that is, $\mathcal{E} = k(\bg)$ for some $\bg = (g_1, \ldots, g_m)$ with $g_1, \ldots, g_m \in k(\bx)$.

One of the key tools in this paper are certain ideals associated to subfields $k \subset \mathcal{E} \subset k(\bx)$.
These ideals were introduced by Ollivier in~\cite[Ch.~III, \S 2]{OllivierPhD} and M\"uller-Quade and Steinwandt in~\cite{MULLERQUADE1999143}.
Thus, we will call them \emph{OMS ideals}.

\begin{definition}[OMS ideal]\label{def:OMS}
    Let $\mathcal{E} \supset k$ be a subfield of $k(\bx)$ represented as $\mathcal{E} = k(g_1,\ldots,g_m)$, where $g_i = \frac{p_i(\bx)}{q_i(\bx)}$ for $i = 1,\ldots,m$.
    We denote by $Q$ the $\operatorname{lcm}$ of $q_1,\ldots,q_m$.
    Introduce new indeterminates $\by \coloneqq (y_1,\ldots,y_n)$. 
    The following ideal in $k(\bx)[\by]$ will be called \emph{the OMS ideal} of $\mathcal{E}$:
\begin{equation}
    \label{eq:OMS_def}
    \OMS_{\mathcal{E}} := \langle p_i(\by) q_i(\bx) - q_i(\by) p_i(\bx)\mid 1 \leqslant i \leqslant m \rangle : (Q(\by))^\infty.
\end{equation}
\end{definition}

\begin{remark}[Dependence on the choice of generators]
    Formally, the definition~\eqref{eq:OMS_def} depends not only on $\mathcal{E}$ but also on the chosen set of generators.
    In fact, \cite[Lemma 1.2]{MULLERQUADE1999143} shows that $\OMS_\mathcal{E}$ is generated by
    $\langle y_1 - x_1,\ldots, y_n - x_n \rangle \cap \mathcal{E}[\by]$ and, thus, does not depend on the choice of generators.
    
    This also implies that, if $\operatorname{char} k = 0$, then $\OMS_{\mathcal{E}}$ is radical.
    It may be not the case in positive characteristic: $\OMS_{\mathbb{F}_p(x^p)}$ is not radical.
\end{remark}

OMS ideals are useful for characterizing the elements of rational function fields algorithmically by associating rational functions with polynomials of special form:

\begin{lemma}[{{\cite[Lemmas 1.2 and 1.6.]{MULLERQUADE1999143}}, \cite[p. 51, Theorem~3]{OllivierPhD}}]\label{lem:OMS}

Let $\mathcal{E} \supset k$ be a subfield of $k(\bx)$.
Then:
\begin{enumerate}
    \item for every polynomials $p,q \in k[\bx]$, $\frac{p}{q} \in \cE$ ~if and only if~ $p(\by) q(\bx) - q(\by) p(\bx) \in \OMS_{\mathcal{E}}$.
    \item \label{lem:OMS:item2} Let $f_1\ldots,f_s$ be the
    reduced \grobner{} basis of $\OMS_{\mathcal{E}}$ with respect to any ordering.
    Then the coefficients of $f_1,\ldots,f_s$ generate $\mathcal{E}$ over~$k$.
\end{enumerate}
\end{lemma}

The second property from the lemma provides, in particular, a way to compute a new (often, simpler) set of generators for $\mathcal{E}$ starting from some generators.

\begin{example}
\label{ex:simple-OMS-simplify}
    Let $\cE = \mathbb{Q}\left( x_1^2 + x_2^2, x_1^3 + x_2^3, x_1^4 + x_2^4 \right) \subset \mathbb{Q}(x_1, x_2)$. We have $Q(y_1, y_2) = 1$ and
    \[
    \OMS_{\mathcal{E}} = \left\langle
    \begin{array}{ll}
     y_1^2 + y_2^2 - (x_1^2 + x_2^2), \\
      y_1^3 + y_2^3 - (x_1^3 + x_2^3), \\
      y_1^4 + y_2^4  -(x_1^4 + x_2^4)
    \end{array}\right\rangle \subset \QQ(x_1,x_2)[y_1,y_2].
    \]
    The reduced \grobner{} basis of $\OMS_{\cE}$ with respect to the lexicographic ordering on $y_1,y_2$ with $y_2 < y_1$ is
    \[
    y_1 + y_2 - \dashuline{(x_1 + x_2)}, \; y_2^2 - \dashuline{(x_1 + x_2)} + \dashuline{~x_1x_2~}.
    \]
    By \Cref{lem:OMS}, the underlined coefficients generate $\cE$, so $\cE = \mathbb{Q}(x_1 + x_2, x_1x_2)$.
    Thus, the computation of the Gr\"obner basis of the OMS ideal allowed us to recognize that $\cE$ is the field of symmetric functions and provided us with the most natural set of generators.
\end{example}

Since the definition~\eqref{eq:OMS_def} involves saturation, practical computation with $\OMS_\mathcal{E}$ requires passing through an elimination ordering (if the saturation is done via the Rabinowitsch
trick) or using dedicated algorithms for computing \grobner{} bases of saturations.
It would naturally simplify the computations if we could employ the Rabinowitsch
trick with respect to a new variable $t$ 
without performing
elimination.
\begin{notation}[$\eOMS$]
    \label{notation:eOMS}
    In the notation of~\Cref{def:OMS}, denote
    \begin{equation}\label{eq:lazy_OMS}
        \eOMS_{\bg} := \langle p_1(\by) q_1(\bx) - q_1(\by) p_1(\bx), \ldots, p_m(\by)q_m(\bx) - q_m(\by)p_m(\bx), t Q(\by) - 1\rangle \subset k(\bx)[\by, t].
    \end{equation}
\end{notation}
Obviously, the first statement from \Cref{lem:OMS} holds if we replace $\OMS_{\mathcal{E}}$ with $\eOMS_{\bg}$.
As we prove in the following lemma, the same is true for the second statement as well.

\begin{lemma}
\label{lemma:OMS_coeffs_2}
    Let $\mathcal{E} \supset k$ be a subfield of $k(\bx)$ represented by $\mathcal{E} = k(\bg)$.
    Let $f_1, \ldots, f_s$ be the
    reduced \grobner{} basis of $\eOMS_{\bg}$ with respect to any ordering (see~\eqref{eq:lazy_OMS}).
    Then the coefficients of $f_1, \ldots, f_s$ generate $\mathcal{E}$ over~$k$.
\end{lemma}
\begin{proof}
    Let $\mathcal{F}$ be the field generated over $k$ by the coefficients of $f_1, \ldots, f_s$.
    Since the original ideal generators from~\eqref{eq:lazy_OMS} have coefficients in $\mathcal{E}$ and \grobner{} basis computation does not extend the ground field, we conclude that $\mathcal{F} \subset \mathcal{E}$.
    Assume that $\mathcal{F} \neq \mathcal{E}$.
    Without loss of generality, $g_1 \not\in \mathcal{F}$.
    We choose a basis $\{e_i\}_{i \in \Lambda}$ of $\mathcal{E}$ over $\mathcal{F}$, where $\Lambda$ is either $\{1, 2, \ldots, \dim_{\mathcal{F}}\mathcal{E}\}$ if $\dim_{\mathcal{F}}\mathcal{E} < \infty$ or $\mathbb{N}$ otherwise.
    The basis can be chosen so that $e_1 = 1$.
    We write $g_1$ in this basis as $\sum\limits_{i \in \Lambda} c_i e_i$, where $c_i \in \mathcal{F}$
    for every $i \in \Lambda$.
    Since $g_1 \not\in \mathcal{F}$, there exists $i_0 \in \Lambda \setminus \{1\}$ such that $c_{i_0} \neq 0$.
    We write $g_1 = \frac{p_1(\bx)}{q_1(\bx)}$ for $p_1, q_1 \in k[\bx]$.
    Let $r_1$ and $r_2$ be the normal forms of $p_1(\by) - c_{1}(\bx) q_1(\by)$ and $q_1(\by)$ with respect to $f_1, \ldots, f_s$.
    Since reductions do not extend the ground field, $r_1, r_2 \in \mathcal{F}[\by, t]$.
    Using the expansion for $g_1$ with respect to the basis, $p_1(\by) - g_1(\bx) q_1(\by)$ can be written as 
    \[
    (p_1(\by) - c_1(\bx)q_1(\by))e_1 -
    \sum\limits_{i \in \Lambda \setminus \{1\}} c_i(\bx) q_1(\by) e_i.
    \]
    Then the normal form of $p_1(\by) - g_1(\bx) q_1(\by)$ with respect to $f_1, \ldots, f_s$ is equal to
    \[
    r_1 e_1 -
    \sum\limits_{i \in \Lambda \setminus \{1\}} c_i r_2 e_i.
    \]
    This expression is equal to zero, so we have $c_i r_2 e_i = 0$ for every $i \in \Lambda \setminus\{1\}$.
    In particular, $c_{i_0} r_{2} e_{i_0} = 0$, so $r_2 = 0$.
    Then $q_1(\by) \in \eOMS_{\bg} \implies Q(\by) \in \eOMS_{\bg} \implies 1 \in \eOMS_{\bg}$.
    This contradicts the fact that $\eOMS_{\bg}$ has a solution $\by = \bx$.
    The contradiction implies that $\mathcal{F} = \mathcal{E}$ as desired. 
\end{proof}

\begin{example}[OMS may bear excess]\label{OMS:excess}
    While coefficients of the OMS ideal turn out to be very useful in finding simpler generating sets for subfields of $k(\bg)$, not all the coefficients are equally useful. 
    Fix a positive integer $\ell$ and consider a subfield in $k(x_1, x_2)$ generated by $g_1 = x_1 + x_2$ and $g_2 = \frac{x_1^\ell}{x_2^\ell}$.
    Then one can show that the degrevlex Gr\"obner basis of $\OMS_{k(g_1, g_2)}$ will consist of $y_1 + y_2 - (x_1 + x_2)$ (which gives one generator, $x_1 + x_2$) and the following polynomial:
    \[
    y_2^{\ell} - \underbrace{\frac{\ell x_2^\ell (x_1 + x_2)}{x_2^\ell + (-1)^{\ell + 1}x_1^{\ell}}}_{\text{second generator}} y_2^{\ell - 1} + \sum\limits_{k = 0}^{\ell - 2} (-1)^{\ell - k}\binom{\ell}{k} \underbrace{\frac{(x_1 + x_2)^{\ell - k - 1}b^\ell}{x_2^\ell + (-1)^{\ell + 1}x_1^{\ell}}}_{\text{useless coefficients}} y_2^k.
    \]
    Not only the majority of the coefficients are
    not needed to generate the field but also the degree of the numerator of the extra coefficients may reach $2 \ell - 1$, which is almost twice larger that of
    the original generators.
    One observes similar phenomenon when replacing $\OMS_{k(g_1, g_2)}$ with $\eOMS_{g_1, g_2}$.
\end{example}

\begin{example}[Dura lex, sed lex]
\label{ex:duralexsedlex}
    The generating set coming from the coefficients of the OMS ideal may depend significanlty on the ordering used for the Gr\"obner basis computation.
Consider a tuple $\bg := (x_1 (x_2 - x_3), x_2 + x_3, x_2 x_3)$ of elements of $k(x_1, x_2, x_3)$.
Computing the Gr\"obner basis of $\OMS_{k(\bg)}$ with respect to either the lexicographic or the degrevlex ordering with $y_1 > y_2 > y_3$ yields
\[
y_3^2 - \dashuline{(x_2 + x_3)} y_3 + \dashuline{x_2 x_3},\;\; y_2 + y_3 - \dashuline{(x_2 + x_3)},\;\; y_1 + 2 \dashuline{\frac{x_1}{x_2 - x_3}} y_3 - \dashuline{\frac{x_1(x_2 + x_3)}{x_2 - x_3}}.
\]
This basis gives back the original generators $x_2 + x_3$ and $x_2 x_3$ and provides additional $\frac{x_1}{x_2 - x_3}$ and $\frac{x_1(x_2 + x_3)}{x_2 - x_3}$, which do not look simpler than the original $x_1(x_2 - x_3)$.
On the other hand, the Gr\"obner basis of $\OMS_{k(\bg)}$ with respect to either the lexicographic or the degrevlex ordering with $y_3 > y_2 > y_1$ is equal to 
\[
y_1^2 - \dashuline{~x_1^2}, \;\; y_2 + \dashuline{\frac{x_3 - x_2}{2x_1}}y_1 - \dashuline{\frac{x_2 + x_3}{2}},\;\; y_3 + \dashuline{\frac{x_2 - x_3}{2x_1}}y_1 - \dashuline{\frac{x_2 + x_3}{2}}.
\]
Here, a simple generator $x_2 x_3$ has disappeared but another simple generator, namely $x_1^2$, is discovered.
\end{example}


\subsection{Randomized subfield membership testing}

\Cref{lem:OMS} provides a constructive way to check if a given function $f \in k(\bx)$ belongs to a subfield $k(g_1, \ldots, g_m)$.
However, this approach would involve computing Gr\"obner basis with coefficients in the field of rational functions $k(\bx)$.
This may be quite heavy and one of the leitmotifs of the present paper is to avoid such computations by working with specializations of the corresponding ideals.

\begin{notation}[Specialization of $\eOMS$]
\label{def:OMS:spec}
    In the notation of~\Cref{def:OMS}, consider $\ba \in k^n$.
    We define $\eOMS_{\bg}(\ba)$ by:
    \[
      \eOMS_{\bg}(\ba) := \langle p_1(\by) q_1(\ba) - q_1(\by) p_1(\ba), \ldots, p_m(\by)q_m(\ba) - q_m(\by)p_m(\ba), t Q(\by) - 1\rangle \subseteq k[\by, t].
    \]
\end{notation}

For membership testing, we will rely on the following fact.

\begin{theorem}[{\cite[Theorem 3.3]{stident}}]\label{thm:randomized_membership}
    Let $k$ be a field of characteristic zero.
    Let $p, p_1, \ldots, p_m, q \in k[\bx]$, where $\bx = (x_1, \ldots, x_n)$.
   Let $0 < \varepsilon < 1$ be a real number.
   We define 
   \[
     d := \max(\deg q + 1, \deg p, \deg p_1, \ldots, \deg p_m) \quad \text{ and }\quad M := \frac{6 d^{n + 3}}{\varepsilon}
     .
   \]
   Let $\ba \coloneqq (a_1, \ldots, a_n)$ be integers 
   sampled
   uniformly and independently from $[0, M]$.
   Consider the field $\mathcal{E} := k\left( \frac{p_1(\bx)}{q(\bx)}, \ldots, \frac{p_m(\bx)}{q(\bx)}  \right)$.
   Then,
   \begin{enumerate}
       \item if $p(\by)q(\mathbf{a}) - p(\mathbf{a})q(\by) \in \eOMS_{\bg}(\ba)$, then $\frac{p(\bx)}{q(\bx)} \in \mathcal{E}$ with probability at least $1 - \varepsilon$;
       \item if $p(\by)q(\mathbf{a}) - p(\mathbf{a})q(\by) \not\in \eOMS_\bg(\ba)$, then $\frac{p(\bx)}{q(\bx)} \not\in \mathcal{E}$ with probability at least $1 - \varepsilon$.
   \end{enumerate}
\end{theorem}

This theorem yields a randomized Monte-Carlo algorithm for checking field membership requiring Gr\"obner basis computation only over $k$.
We can go a bit further and use the approach from~\cite{ilmer2022efficientidentifiabilityverificationode} to ensure that the ideal
for which the \grobner{} basis computation is performed
is of dimension zero.
This is based on the following lemma implicitly contained in~\cite[Proof of Theorem 1]{ilmer2022efficientidentifiabilityverificationode} (also reminiscent to the use of cross-section in rational invariant computation~\cite[Section~3]{Hubert2007}).

\begin{lemma}\label{lem:transcendence_basis}
    Consider a subfield $\cE = k(g_1, \ldots, g_m) \subseteq k(\bx)$, and let $z_1, \ldots, z_s$ be a transcendence basis of $k(\bx)$ over $\cE$.
    Assume that $g \in k(\bx)$ is algebraic over $\cE$.
    Then $g \in \cE \iff g \in \cE(z_1, \ldots, z_s)$.
\end{lemma}

\begin{proof}
    The implication $g \in \cE \implies g \in \cE(z_1, \ldots, z_s)$ is trivial, let us prove the converse.
    Assume that $g \in \cE(z_1, \ldots, z_s)$ and $g$ is algebraic over $\cE$.
    Let $g_1, \ldots, g_\ell$ be all the roots of the minimal polynomial for $g$ over $\cE$ in $\overline{k(\bx)}$ such that $g_1 = g$.
    We write $g = G(z_1, \ldots, z_s)$, where $G$ is a rational function from $\cE(Z_1, \ldots, Z_s)$ with $Z_1, \ldots, Z_s$ being new indeterminates.
    Assume that $G$ is nonconstant and, without loss of generality, that it depends nontrivially on $Z_1$.
    Then $G(z_1 + X, z_2, \ldots, z_s)$ is a nonconstant rational function in $X$.
    Since $\overline{k}$ is infinite, there exists $c \in \overline{k}$ such that $G(z_1 + c, z_2, \ldots, z_s) \neq g_i$ for every $1 \leqslant i \leqslant \ell$.
    Since $z_1, \ldots, z_s$ are transcendental over $\cE$, there exists~\cite[Proposition~2.4(a)]{milne2022} an automorphism $\alpha \colon \overline{k(\bx)} \to \overline{k(\bx)}$ such that $\alpha|_{\cE}$ is the identity, $\alpha(z_1) = z_1 + c$ (where $c$ is viewed as an element of $\overline{k} \subset \overline{k(\bx)}$), and $\alpha(z_i) = z_i$ for every $i = 2, \ldots, s$.
    Then $\alpha(g)$ satisfies the same equation over $\cE$ as $g$, so it must belong to the set $\{g_1, \ldots, g_\ell\}$.
    This contradicts the choice of $c$.
    Therefore, $G \in \cE$, so $g\in \cE$.
\end{proof}

Combining \Cref{thm:randomized_membership} and \Cref{lem:transcendence_basis}, we arrive at the following algorithm.

\begin{algorithm}[H]
\caption{Randomized subfield membership check}
\label{alg:membership}
\begin{description}[itemsep=0pt]
\item[Input:] polynomials $p, p_1, \ldots, p_m, q \in k[\bx]$, where $\operatorname{char} k = 0$, and a real number $0< \varepsilon < 1$;
\item[Output:] \texttt{True} if $\frac{p}{q} \in k\left(\frac{p_1}{q}, \ldots, \frac{p_m}{q} \right)$, and \texttt{False} otherwise.
\end{description}

\begin{enumerate}[label = \textbf{(Step~\arabic*)}, leftmargin=*, align=left, labelsep=\favoritelabelsep, itemsep=\favoriteitemsep]
    \item Set $d \coloneqq \max(\deg q + 1, \deg p, \deg p_1, \ldots, \deg p_m)$.
    \item Compute the Jacobian $J(\bx)$ of $\frac{p_1}{q}, \ldots, \frac{p_m}{q}$ (we assume that the gradients are the rows of $J(\bx)$) and the gradient $G(\bx)$ of $\frac{p}{q}$.
    \item Sample each coordinate of $\mathbf{a} \in \ZZ^n$ uniformly at random from $\left[0,\; \frac{2(2\min(m, n) + 1)(d - 1)}{\varepsilon}\right] \cap \ZZ$.
    \item\label{step:membership_Jac_test} If $G(\mathbf{a})$ does not belong to the rowspan of $J(\mathbf{a})$, {\bf return} $\texttt{False}$.
    \item\label{step:choose_pivots} Take any set of columns of $J(\mathbf{a})$ forming a basis of the column space of the matrix.
    Let $S \subset \{1, \ldots, n\}$ be the indices these columns.
    Set $\overline{S} := \{1, \ldots, n\} \setminus S$.
    \item Sample $\mathbf{b} \in \ZZ^n$ with each coordinate drawn uniformly at random from $[0, 12 d^{n + 3} / \varepsilon] \cap \ZZ$.
    \item\label{step:membership_groebner} Using Gr\"obner basis, check if $p(\bx)q(\mathbf{b}) - p(\mathbf{b})q(\bx)$ belongs to 
    \[
    \eOMS_{\mathbf{p} / q}({\mathbf{b}}) + \langle x_j - b_j \mid j \in \overline{S}  \rangle.
    \]
    \item If belongs, {\bf return} \texttt{True}, otherwise {\bf return} \texttt{False}.
\end{enumerate}
\end{algorithm}
\begin{proposition}
\label{prop:membership-proba}
    The output of~\Cref{alg:membership} is correct with the probability at least $1 - \varepsilon$.
\end{proposition}

\begin{proof}
    Denote $g := \frac{p}{q}$, $g_1 := \frac{p_1}{q}, \ldots, g_m := \frac{p_m}{q}$, and $\cE := k(g_1, \ldots, g_m)$.
    Let $\widetilde{J}(\bx)$ be the matrix obtained by adding an extra row $G(\bx)$ to $J(\bx)$.

    First, we consider the case when $g$ is transcendental over $\cE$.
    By~\cite[Proposition 2.4]{ehrenborg1993apolarity}, this implies $\operatorname{rank} J(\bx) < \operatorname{rank} \widetilde{J}(\bx)$ and, in particular, any maximal nonsingular minor of $\widetilde{J}(\bx)$ contains the row correspodning to $G$.
    Let $M(\bx)$ be any maximal nonsingular minor of
    $\widetilde{J}(\bx)$.
    Assume that $q(\mathbf{a}) \neq 0$ (so that $\widetilde{J}(\mathbf{a})$ is well-defined) and $\det M(\mathbf{a}) \neq 0$.
    Then 
    \[\operatorname{rank} J(\mathbf{a}) \leqslant \operatorname{rank} J(\bx) < \operatorname{rank} \widetilde{J}(\bx) = \operatorname{rank} \widetilde{J}(\mathbf{a}),
    \]so the algorithm will return \texttt{False} at~\ref{step:membership_Jac_test} which will be the correct answer. 
    The degree of the numerator of $\det M(\bx)$ does not exceed $\min(m, n) (2d - 2)$. 
    Therefore, by the Schwartz-Zippel lemma (\Cref{lemma:sz}), $\widetilde{J}(\bx)$ will be well-defined under $\bx \to \mathbf{a}$ and the correct result will be returned with the probability at least
    \[
      1 - \frac{d - 1 + \min(m, n) (2d - 2)}{2(2\min(m, n) + 1)(d - 1) / \varepsilon} = 1 - \frac{\varepsilon}{2} > 1 - \varepsilon.
    \]
    Next, we consider the case when $g$ is algebraic over $\cE$.
    In this case, $G(\bx)$ belongs to the rowspan of $J(\bx)$ by~\cite[Proposition 2.4]{ehrenborg1993apolarity}.
    Let $M(\bx)$ be a maximal nonsingular minor of $J(\bx)$.
    Assume that $q(\mathbf{a}) \neq 0$ and $\det M(\mathbf{a}) \neq 0$.
    Then $\operatorname{rank} J(\ba) = \operatorname{rank}J(\bx) = \operatorname{rank}\widetilde{J}(\bx) \geqslant \operatorname{rank}\widetilde{J}(\ba)$.
    Therefore, the algorithm will not return on~\ref{step:membership_Jac_test}.

    Furthermore, since $\operatorname{rank} J(\ba) = \operatorname{rank}J(\bx)$, every set of columns 
    which is
    a basis of the columns space of $J(\ba)$ is a specialization of a column basis for $J(\bx)$.
    Thus, \cite[Proposition 2.4]{ehrenborg1993apolarity} implies that $\{x_{j} \mid j \in \overline{S}\}$ is a transcendence basis of $k(\bx)$ over $\cE$.
    Then, by~\Cref{lem:transcendence_basis}, 
    $g \in \cE$ if and only if
    $g \in \cE(x_j \mid j \in \overline{S})$.
    The latter containment will be assessed by~\ref{step:membership_groebner} with the error probability not exceeding $\frac{\varepsilon}{2}$ by~\Cref{thm:randomized_membership} (where the saturation ideal is replaced by the one given by the Rabinowitsch trick, this does not affect the membership).
    Using the same calculation as in the first part of the proof, the probability of $q(\bx)\det M(\bx)$ vanishing at $\ba$ also does not exceed $\frac{\varepsilon}{2}$.
    Therefore, the returned result will be correct with probability at least~$1 - \varepsilon$.
\end{proof}

\begin{remark}
    In practice, in~\ref{step:choose_pivots}, we cache the set $\overline{S}$ for a field $\cE$ to reuse in  
    further
    membership tests.
\end{remark}

\begin{remark}\label{rem:modular_membership}
    In some parts of our algorithm, we will not need the explicit error bound, so we will perform computation modulo large prime (and sampling from the whole field) instead of $k = \QQ$.
    Then, for every input polynomials and every error probability $\varepsilon$, there will be a large enough prime (e.g., exceeding all the integers arising during the computation over $\QQ$) so that the error bound holds for the modular version of~\Cref{alg:membership} as well.
\end{remark}

\begin{remark}[On saturation]
    When computing \grobner{} bases of specialized $\eOMS$ ideals using our implementation (\Cref{sec:tracing}), we found that factoring the saturation polynomial $Q(\by)$ (as in~\cite[Lemma 1.3.]{MULLERQUADE1999143}) does not provide a measurable speed-up on our examples.
\end{remark}

\begin{remark}\label{rem:field_equality}
    \Cref{alg:membership} yields a randomized algorithm for testing equality of subfields of $k(\bx)$.
    Let $\mathcal{E} = k(f_1, \ldots, f_s)$ and $\mathcal{F} = k(g_1, \ldots, g_r)$ be two such subfields, and let $0 < \varepsilon < 1$ be a real number.
    Then we apply~\Cref{alg:membership} to check whether
    \[
    \bigl(\forall 1 \leqslant i \leqslant s\colon f_i \in \mathcal{F}\bigr)\; \& \; \bigl(\forall 1 \leqslant j \leqslant r\colon g_j \in \mathcal{E}\bigr).
    \]
    If each individual check is performed with the probability of correctness at least $1 - \frac{\varepsilon}{s + r}$, then the result is correct with probability at least $1 - \varepsilon$.
\end{remark}


\section{Building blocks of the algorithm}\label{sec:blocks}

\subsection{Basic minimization}

We start with a method that follows straightforwardly from the field membership algorithm. Let $\bg \subset 
k(\bx)$. The following method computes a subset $\bh \subseteq \bg$ that generates the same field as $\bg$ over $k$ and is minimal with respect to inclusion.

\begin{algorithm}[H]
\caption{Computing a minimal subset}
\label{alg:minimization}
\begin{description}[itemsep=0pt]
\item[Input: ] 
$\bg \subset k(\bx)$, where $\operatorname{char} k = 0$,
and a real number $0 < \varepsilon < 1$;
\item[Output: ] $\bh \subseteq \bg$ such that $k(\bh) = k(\bg)$ and $\bh$ is minimal with respect to inclusion.
\end{description}

\begin{enumerate}[label = \textbf{(Step~\arabic*)}, leftmargin=*, align=left, labelsep=\favoritelabelsep, itemsep=\favoriteitemsep]
    \item Set $\bh \coloneqq (h_1,\ldots,h_m) \coloneqq \bg$.
    \item For $i = 1,\ldots,m$ do
    \begin{enumerate}
        \item If $h_i \in k(\bh \setminus \{ h_i \})$ with probability at least $1-\frac{\varepsilon}{m}$ using~\Cref{alg:membership}, then set $\bh \coloneqq \bh \setminus \{ h_i \}$.
    \end{enumerate}
    \item {\bf Return} $\bh$.
\end{enumerate}
\end{algorithm}

\begin{proposition}\label{prop:minimization_correctness}
    \Cref{alg:minimization} is correct with probability at least $1 - \varepsilon$.
\end{proposition}
\begin{proof}
    To prove correctness, it suffices to show that $k(\bh) = k(\bg)$ and that there does not exist $h \in \bh$ such that $k(\bh \setminus \{h\}) = k(\bh)$. 
    Initially, we have $k(\bh) = k(\bg)$, and the field generated by $\bh$ does not change during the iterations. 
    Thus, it suffices to show that such $h$ does not exist. 
    Assume that it does. Then, $h \in k(\bh \setminus \{h\})$. However, at some iteration $i \in \{ 1,\ldots,m\}$, we had $h = h_i$, so $h$ would have been removed from $\bh$, which is a contradiction.

    There are $m$ independent applications of \Cref{alg:membership}, so we combine \Cref{prop:membership-proba} with the Bernoulli inequality to obtain that the overall probability of correctness is at least
    \[
    \left(1 - \frac{\varepsilon}{m}\right)^m \geqslant 1 - \varepsilon.
    \]\qedhere
\end{proof}\begin{remark}[On simplicity]
\label{remark:onsimplicity}
    We note that result depends on the order of elements in the input $\bg$. We may order $\bg$ so that simple elements come last,
    so that they are used to remove other elements from the list.
    In practice, we use the following total order. Let $a/b, p/q \in k(\bx)$. Recall that for $f \in k[\bx]$, $d_f$ is the degree of $f$ and $s_f$ is the number of terms in $f$.
    We will ensure $d_a \geqslant d_b$ and $d_p \geqslant d_q$ by taking the reciprocal of the functions.
    We say that $a/b$ is simpler than $p/q$:
    \begin{itemize}
        \item if $d_a + d_b < d_p + d_q$, or, in the case of equality,
        \item if $s_a + s_b < s_p + s_q$, or, in the case of equality,
        \item if $d_b < d_q$, or, in the case of equality,
        \item if the largest monomial present in only one of $a$ or $p$ (when such exists) appears in $p$.
    \end{itemize}    
\end{remark}

\subsection{Low-degree coefficients of \grobner{} bases over rational function fields}
\label{sec:gb_sparse_coeffs}
As we have seen, simple generators can be obtained from the low-degree coefficients of \grobner{} bases of OMS ideals.
In this section, we show how to compute those coefficients efficiently using sparse interpolation. 
We note that our algorithm is not restricted to OMS ideals and can be used with arbitrary ideals over the field of rational functions. 

Let $f_1,\ldots,f_m \in k(\bx)[\by]$, where $\bx = (x_1,\ldots,x_n)$ and $\by = (y_1,\ldots,y_t)$. 
In this section, all computations are done with respect to the same fixed monomial ordering.
Let $G$ be the reduced \grobner{} basis of $\langle f_1,\ldots,f_m \rangle$. We have $G = (g_1,\ldots,g_\eta) \subset k(\bx)[\by]$, where for every $i=1,\ldots,\eta$,
\[
g_i = \sum_{j=1}^{s_{g_i}} c_{i,j} 
\Pi(\by^{\balpha_{i,j}}),
\]
where for every $j = 1,\ldots,s_{g_i}$ we have $c_{i,j} = p_{i,j} / q_{i,j} \in k(\bx) \setminus \{0\}$ with coprime $p_{i,j},q_{i,j} \in k[\bx] \setminus \{0\}$, and $\balpha_{i,j} \in \Zplus^t$ are pairwise different.

Let $d$ be an integer and
\begin{equation}
\label{eq:gb-exps}
\begin{aligned}
A_{\leqslant d} &\coloneqq \{(i, j) \mid \deg p_{i,j} + \deg q_{i,j} \leqslant d,~ i=1,\ldots,\eta, ~j = 1,\ldots,s_{g_i}\},\\
C_{\leqslant d} &\coloneqq \{c_{i,j} \mid (i,j) \in A_{\leqslant d}\} \subset k(\bx).
\end{aligned}
\end{equation}
Note $C_{\leqslant d}$ is a subset of the coefficients of $G$. We wish to compute $C_{\leqslant d}$ without computing the whole $G$.

\begin{lemma}\label{lem:GB_spec}
    For every $f_1, \ldots, f_m \in k(\bx)[\by]$ and any monomial ordering, there exists a nonzero polynomial $g(\bx) \in k[\bx]$ with the following property.
    For every $\ba \in k^n$ such that $g(\ba) \neq 0$ the sets of polynomials in $k[\by]$ are well-defined and equal:
    \begin{itemize}
        \item  the reduced Gr\"obner basis of $\langle f_1(\ba)(\by), \ldots, f_m(\ba)(\by) \rangle$;
        \item the image of the reduced Gr\"obner basis of $\langle f_1, \ldots, f_m\rangle \subset k(\bx)[\by]$ under the substitution $\bx \to \ba$.
    \end{itemize}
\end{lemma}
\begin{proof}
    Consider the run of the Buchberger's algorithm for $f_1, \ldots, f_m$ over $k(\bx)$.
    We set $g$ to be the least common multiple of all the denominators occurring in the input and during the computation as well as the numerators of the leading coefficients of all input and intermediate polynomials.
    Consider $\ba \in k^n$ such that $g(\ba) \neq 0$.
    The Buchberger's algorithm, in addition to polynomial operations, queries the coefficients of polynomials in the following cases:
    \begin{enumerate}
        \item it queries the leading terms to form critical pairs and compute S-polynomials;
        \item it queries all the nonzero coefficients of a polynomial when reducing it.
    \end{enumerate}
    We consider the runs of the algorithm for $f_1(\by), \ldots, f_m(\by)$ and for its specialization $f_1(\ba)(\by), \ldots, f_m(\ba)(\by)$.
    Due to the assumption $g(\ba) \neq 0$, the queries of the first type will yield the same leading monomials and the returned coefficients will commute with the specialization.
    For the queries of the second type, the only possible difference is if a non-leading coefficient has vanished under specialization. 
    In this case, we will assume that its coefficient is still returned but is equal to zero, so there will be an extra reduction performed with zero coefficient which will not affect the result.
    Therefore, the polynomial operations performed will be the same regardless if the specialization was performed before the Gr\"obner basis computation or after.
\end{proof}

Let $\operatorname{coeff}(G, i, \balpha)$ denote the coefficient in front of the monomial $\Pi(\by^{\balpha})$
in the $i$-th polynomial in $G$, or zero if the $i$-th polynomial or such monomial are not present in $G$; we shall reuse this notation for different~$G$. 

Let $i \in \{1,\ldots,\eta\}$ and $j \in \{1,\ldots,s_{g_i}\}$. Let $\blackbox_{i,j}$ be a blackbox that, given as input an arbitrary $\ba \in k^n$, computes $G_{\ba} \subset k[\by]$, the reduced \grobner{} basis of $\langle f_1(\ba)(\by),\ldots,f_m(\ba)(\by) \rangle$, and returns $\operatorname{coeff}(G_{\ba}, i, \balpha_{i,j})$.
More precisely, three outcomes are possible:
\begin{itemize}
    \item $\blackbox_{i,j}(\ba)$ returns {\tt FAIL}. 
    This happens when a denominator in ${\bf f}$ vanishes under $\bx \to \ba$. 
    \item $\blackbox_{i,j}(\ba) \neq c_{i,j}(\ba)$. 
    This may happen when the image of $G$ under $\bx \to \ba$ and $G_{\ba}$ do not coincide. 
    \item $\blackbox_{i,j}(\ba) = c_{i,j}(\ba)$. This happens when none of the above happens. 
    By~\Cref{lem:GB_spec}, there exists a nonzero polynomial $g \in k[\bx]$, such that for every $\ba \in k^n$ if
    $g(\ba) \neq 0$ then
    $\blackbox_{i,j}(\ba) = c_{i,j}(\ba)$.
\end{itemize}

The idea of the algorithm is as follows. First,  
consider the case when we have access to the support of $G$, that is, $\eta$ is known and for every $i=1,\ldots,\eta$ and $j=1,\ldots,s_{g_i}$ we have access to $\balpha_{i,j}$. 
Let $S \subseteq k \setminus \{0\}$. 
The output of the blackboxes is normalized consistently across the evaluations, since the \grobner{} bases are reduced. 
So, for every $i=1,\ldots,\eta$ and $j = 1,\ldots,s_i$, we can compute $\deg p_{i,j}$ and $\deg q_{i,j}$ by applying~\Cref{alg:estimate_degrees} to $\blackbox_{i,j}$, $d$, and $S$. 
Consequently, this yields $A_{\leqslant d}$.
Let $\bomega \in (k\setminus\{0\})^{n+1}$ be an admissible ratio as in~\Cref{sec:tiwari}. 
Then, for every $(i,j) \in A_{\leqslant d}$, we may pass $(\blackbox_{i,j}, S, \deg p_{i,j}, \deg q_{i,j}, \bomega)$
to sparse interpolation from~\Cref{alg:interpolation} to recover the sparse representation of $c_{i,j} \in C_{\leqslant d}$. 

Now, for every $i=1,\ldots,\eta$ and $j=1,\ldots,s_{g_i}$, we can obtain $\balpha_{i,j}$ with high probability by computing $G_{\bsigma}$ for $\bsigma \in k^{n}$ whose components are chosen independently and uniformly from $S$. This gives rise to~\Cref{alg:gb}.

\begin{algorithm}[H]
\caption{Computing 
low-degree
coefficients of \grobner{} bases}
\label{alg:gb}
\begin{description}[itemsep=0pt]
\item[Input:] $f_1,\ldots,f_m \in k(\bx)[\by]$; 
an integer $d$; 
a subset $S \subseteq k \setminus \{0\}$.
\item[Output:] $C_{\leqslant d} \subset k(\bx)$ as in~\eqref{eq:gb-exps}, or {\tt FAIL}.
\end{description}

\begin{enumerate}[label = \textbf{(Step~\arabic*)}, leftmargin=*, align=left, labelsep=\favoritelabelsep, itemsep=\favoriteitemsep]
    \item Pick an admissible ratio $\bomega \in (k\setminus \{0\})^{n+1}$, as in~\Cref{sec:tiwari}; choose the coordinates of $\bsigma \in k^n$ independently and uniformly from $S \subseteq k \setminus \{0\}$.
    \item \label{step:discover-shape} For every $i=1,\ldots,\eta$ and every $j=1,\ldots,s_{g_i}$, obtain $\balpha_{i,j} \in \Zplus^t$ by computing $G_{\bsigma}$, the reduced \grobner{} basis of $\langle f_1(\bsigma)(\by),\ldots,f_m(\bsigma)(\by) \rangle \subseteq k[\by]$.
    \item\label{step:blackbox_construction} For every $i=1,\ldots,\eta$ and every $j=1,\ldots,s_{g_i}$, construct the blackbox $\blackbox_{i,j}$ that accepts as input $\ba \in k^n$, computes $G_{\ba}$, the reduced \grobner{} basis of $\langle f_1(\ba)(\by),\ldots,f_m(\ba)(\by)\rangle \subseteq k[\by]$, and returns either $\operatorname{coeff}(G_{\ba}, i, \balpha_{i,j})$ or {\tt FAIL}.
    \item\label{alg:gb:step:degrees} For every $i=1,\ldots,\eta$ and every $j=1,\ldots,s_{g_i}$, compute $\deg p_{i,j}$ and $\deg q_{i,j}$, or {\tt STOPPED} if $\deg p_{i,j} + \deg q_{i,j} > d$, by applying~\Cref{alg:estimate_degrees} to $(\blackbox_{i,j}, d, S)$. 
    If any of the applications failed, {\bf return} {\tt FAIL}. 
    \item Deduce $A_{\leqslant d}$ as in~\eqref{eq:gb-exps}.
    \item\label{alg:gb:step:interpolate} For every $(i,j) \in A_{\leqslant d}$, apply~\Cref{alg:interpolation} to $(\blackbox_{i,j}, S, \deg p_{i,j}, \deg q_{i,j}, \bomega)$ to obtain $p_{i,j}/q_{i,j} \in C_{\leqslant d}$. If any of the applications failed, {\bf return} {\tt FAIL}.
    \item {\bf Return} $C_{\leqslant d}$.
\end{enumerate}
\end{algorithm}

\begin{proposition}
\label{prop:propa-gb-over-rat-funcs}
Let $h\coloneqq \deg g$ as in~\Cref{lem:GB_spec}. Let $\nu = s_{g_1} + \ldots + s_{g_\eta}$. Let $D$ be the maximum of $\deg p_{i,j} + \deg q_{i,j} + 2$ over all $i = 1,\ldots,\eta$ and $j = 1,\ldots,s_{g_i}$. Let $s$ be the maximum of $\max(s_{p_{i,j}}, s_{q_{i,j}})$ over all $(i,j) \in A_{\leqslant d}$. The probability that~\Cref{alg:gb} succeeds and produces a correct result is at least
\[
\underbrace{\left(1 - \frac{h}{|S|}\right)}_{\ref{step:discover-shape}}\underbrace{\left(1 - \frac{\nu D(2 + 3h + 3D)}{|S|}\right)}_{\ref{alg:gb:step:degrees}}\underbrace{\left(1 - \frac{\nu D (\lceil \log_2 s \rceil + 8s + 4sh + 8s D)}{|S|}\right)}_{\ref{alg:gb:step:interpolate}}.
\]
\end{proposition}

\begin{proof}
The blackboxes constructed in~\ref{step:blackbox_construction} might return a wrong result and thus, technically, do not conform to the specifications of \Cref{alg:interpolation,alg:estimate_degrees}. Our goal is to reconcile them.
First we observe that~\Cref{alg:gb} always terminates because \Cref{alg:interpolation,alg:estimate_degrees} always terminate.

Next, for the sake of theoretical argument, we will consider a modified version of~\Cref{alg:gb} which we will call \Cref{alg:gb}${}^*$. 
The difference between these algorithms is that the blackbox constructed in~\ref{step:blackbox_construction} of \Cref{alg:gb}${}^*$ will check if the computed \grobner{} basis is equal to the specialization of the \grobner{} basis over $k(\bx)$ and return {\tt FAIL} if it is not.
This check can be performed, e.g., by explicitly computing the \grobner{} basis over $k(\bx)$ (recall that~\Cref{alg:gb}${}^*$ is a purely theoretical construction).
This makes the constructed blackboxes comply with the requirements of \Cref{alg:estimate_degrees,alg:interpolation}, and the set of inputs yielding failure is still contained in the zero set of $g$ from~\Cref{lem:GB_spec}.
We note that, if \Cref{alg:gb}${}^*$ returns a result, the same result would be returned by~\Cref{alg:gb} with the same input and the same random choices.
Therefore, the probability that~\Cref{alg:gb} will return a correct result is at least the correctness probability for \Cref{alg:gb}${}^*$.
\Cref{alg:gb}${}^*$ succeeds and returns a correct result when~\ref{step:discover-shape} computes the correct exponents and all of the applications of~\Cref{alg:estimate_degrees,alg:interpolation} succeed and return correct results. Note:
\begin{itemize}
    \item The computation in~\ref{step:discover-shape} computes the correct exponents with probability at least $1 - h/|S|$.
    \item The individual applications of \Cref{alg:estimate_degrees,alg:interpolation} are pairwise independent, and all of those are independent from the computation in~\ref{step:discover-shape}.
\end{itemize}
So we apply~\Cref{prop:proba-interpol,prop:estimate_degrees} and the claim follows.
\end{proof}

\begin{remark}
    In~\Cref{alg:gb}, it suffices to compute a single \grobner{} basis to simultaneously obtain the evaluations of all blackboxes, since the admissible ratio is shared among the blackboxes. 
    In particular, in the notation of~\Cref{prop:propa-gb-over-rat-funcs}, \Cref{alg:gb} computes 
    $\Theta(s d)$
    \grobner{} bases in $k[\by]$ to interpolate the result.
\end{remark}


\subsection{Finding polynomial generators}

One specific way to find simple generators is to search for polynomials of low degree in the field $\mathcal{E}$.
Polynomials in a subfield of $k(\bx)$ can be found using the following simple lemma.

\begin{lemma}\label{lem:polynomial}
    Let $g_1, \ldots, g_m \in k(\bx)$ and $p \in k[\bx]$. 
    Then $p \in k(\bg)$ is and only if $p(\by)$ is equivalent to an element of $k(\bx)$ modulo $\eOMS_{\bg}$.
\end{lemma}

\begin{proof}
    If $p \in k(\bg)$, then $p(\by) - p(\bx) \in \eOMS_{\bg}$ by~\Cref{lem:OMS}.
    In the opposite direction, assume that there is $q(\bx) \in k(\bx)$ such that $p(\by) - q(\bx) \in \eOMS_{\bg}$.
    By the definition of $\eOMS_{\bg}$, every element of $\eOMS_{\bg}$ vanishes under substitution $\by \to \bx, t \to Q(\bx)$.
    Then $q(\bx) = p(\bx)$, so $p \in k(\bg)$ by~\Cref{lem:OMS}.
\end{proof}

\begin{notation}\label{not:nfp}
    Let $I \subset k[\bx]$ be an ideal and $p \in k[\bx]$ be a polynomial.
    We fix the degrevlex monomial ordering on $k[\bx]$; another ordering could be used as well.
    By $\NFp_I(p)$ we will denote the image of the normal form of $p$ with respect to the Gr\"obner basis of $I$ in the quotient space $k[\bx] / k\cdot 1$ (where we factor a $k$-vector space $k[\bx]$ by a subspace spanned by $1 \in k[\bx]$).
    In other words, $\NFp_I(p)$ is the normal form of $p$ with removed constant term.
\end{notation}

Using this notation, \Cref{lem:polynomial} can be restated as $p \in k(\bg) \iff \NFp_{\eOMS_\bg}(p) = 0$.

\begin{algorithm}
\caption{Finding low-degree polynomial generators (cf. \cite[Section~3.3]{JimnezPastor2022})}
\label{alg:polynomials}
\begin{description}[itemsep=0pt]
\item[Input: ] $g_1,\ldots,g_m \in k(\bx)$, with $\operatorname{char} k = 0$; integer $\delta$
and subset $S \subseteq k \setminus \{0\}$;
\item[Output: ] a basis of the vector space $k(\bg) \cap \{p(\bx) \in k[\bx] \mid \deg p \leqslant \delta\}$.
\end{description}

\begin{enumerate}[label = \textbf{(Step~\arabic*)}, leftmargin=*, align=left, labelsep=\favoritelabelsep, itemsep=\favoriteitemsep]
    \item\label{step:initV} Set $V$ to be the $k$-vector space spanned by the monomials in $\by = (y_1, \ldots, y_n)$ of degree at most~$\delta$.
    \item Set $s \coloneqq \dim V$.

    \item\label{step:NF_loop} Repeat
    \begin{enumerate}
        \item Set $\ba$ be a vector in $k^n$ with the coordinates sampled uniformly at random from $S$. 
        \item Set $V$ to be the kernel of the linear map $\NFp_{\eOMS_{\bg}(\ba)}\colon V \to k[\by, t] / k\cdot 1$.\\
        \emph{The range of the map belongs to the (finite-dimensional) subspace of polynomials of $\deg \leqslant \delta$}
        \item If $\dim V = s$, then \textbf{return}  a basis of $V$ substituting $\by \to \bx$.
        \item Set $s \coloneqq \dim V$.
    \end{enumerate}
\end{enumerate}
\end{algorithm}

\begin{proposition}\label{prop:polynomials_correctness}
    \Cref{alg:polynomials} terminates.
    For every $g_1, \ldots, g_m$, $\delta$, and $\varepsilon > 0$, there exists $D_0$ such that, if $|S| > D_0$, then the result is correct with probability at least $1 - \varepsilon$.
\end{proposition}

\begin{lemma}\label{lem:extended_matrix}
    Let $M(\bx)$ be a matrix over $k(\bx)$, and let $r$ be the dimension of the $k$-vector space spanned by the rows of $M$.
    Consider a matrix $M(\bx_1, \ldots, \bx_r) := \begin{pmatrix} M(\bx_1) & \ldots & M(\bx_r)\end{pmatrix}$ obtained by concatenating matrices $M(\bx_1), \ldots, M(\bx_r)$, where $\bx_1, \ldots, \bx_r$ are independent tuples of indeterminates.
    Then $M(\bx_1, \ldots, \bx_r)$ has rank $r$.

    Furthermore, for every $s < r$ 
    \[
    \operatorname{rank}M(\bx_1, \ldots, \bx_s) = \operatorname{rank}M(\bx_1, \ldots, \bx_{s + 1}) \implies \operatorname{rank}M(\bx_1, \ldots, \bx_s) = r.
    \]
\end{lemma}

\begin{proof}
    Since every $k$-linear dependence between the rows of $M(\bx)$ yields a $k$-linear dependence between the rows of $M(\bx_1, \ldots, \bx_r)$, the rank of the latter is at most $r$.
    Let $c$ be the number of columns in $M(\bx)$. 
    We introduce a new variable $y$ and consider a vector $v(\bx, y) := M(\bx) \cdot \begin{pmatrix}
        1 & y & \ldots & y^{c - 1}
    \end{pmatrix}^T$.
    The $k$-linear span of the entries of the vector has also dimension $r$.
    Therefore, by~\cite[Lemma 19]{mukhina2025projectingdynamicalsystemssupport}, the matrix $\begin{pmatrix}
        v(\bx_1, y_1) & \ldots  & v(\bx_r, y_r)
    \end{pmatrix}$ has rank $r$.
    Since this matrix is a product of $M(\bx_1, \ldots, \bx_r)$ and a block matrix with blocks of the form $\begin{pmatrix}
        1 & y_i & \ldots & y_i^{c - 1}
    \end{pmatrix}^T$, the rank of $M(\bx_1, \ldots, \bx_r)$ is at least $r$.
    This proves $\operatorname{rank} M(\bx_1, \ldots, \bx_r) = r$.

    To prove the second claim, we observe that the equality $\operatorname{rank}M(\bx_1, \ldots, \bx_s) = \operatorname{rank}M(\bx_1, \ldots, \bx_{s + 1})$ implies that the columns of $M(\bx_{s + 1})$ belong to the column space of $M(\bx_1, \ldots, \bx_s)$.
    Then the same will be true for $M(\bx_{s + 2}), \ldots, M(\bx_r)$, so $\operatorname{rank}M(\bx_1, \ldots, \bx_s) = \operatorname{rank}M(\bx_1, \ldots, \bx_r) = r$.
\end{proof}

\begin{proof}[Proof of~\Cref{prop:polynomials_correctness}]
    For proving termination, we observe that, at each iteration of the loop on~\ref{step:NF_loop}, either the dimension of $V$ becomes smaller or the algorithm returns.
    Thus, the algorithm will return after a finite number of iterations of the loop.

    We apply~\Cref{lem:GB_spec} to the generators of the ideal $\eOMS_\bg$ and denote the resulting polynomial by $b(\bx) \in k[\bx]$.
    Let $m_1, \ldots, m_N$ be the monomials of degree at most $\delta$ in $\by$.
    Then $V$ at~\ref{step:initV} is spanned by $m_1, \ldots, m_N$.
    Let $N_1$ be the dimension of the $k(\bx)$-space of polynomials of degree at most $\delta$ in $\by$ and~$t$ without constant term.
    We choose a $k(\bx)$-basis in this space, then we can write $\NFp_{\eOMS_\bg}(m_1), \ldots, \NFp_{\eOMS_\bg}(m_N)$ in the basis as vectors in $k(\bx)^{N_1}$.
    We consider an $N\times N_1$-matrix $M(\bx)$ whose rows are these vectors.
    By \Cref{lem:polynomial}, polynomials of degree at most $\delta$ in $k(\bg)$ are in one-to-one correspondence with the $k$-linear dependencies between the rows of $M(\bx)$.
    We denote by $r$ the dimension of the $k$-vector space of these dependencies.
    By \Cref{lem:extended_matrix}, matrix $M(\bx_1, \ldots, \bx_r)$ has rank $r$.
    Let $\ell$ be the smallest integer such that $M(\bx_1, \ldots, \bx_{\ell})$ has rank $r$.
    Let the ranks of $M(\bx_1)$, $M(\bx_1, \bx_2), \ldots, M(\bx_1, \ldots, \bx_{\ell})$ be $r_1, r_2, \ldots, r_\ell = r$, respectively.
    By the second part of \Cref{lem:extended_matrix}, $r_1 < r_2 < \ldots < r_\ell$.
    For every $1 \leqslant i \leqslant \ell$, let $q_i(\bx_1, \ldots, \bx_i)$ be the determinant of any nonsingular $r_i\times r_i$ minor of $M(\bx_1, \ldots, \bx_{i})$.
    We denote by $q(\bx_1, \ldots, \bx_\ell)$ the product $q_1\cdot\ldots\cdot q_\ell$.

    Let $\ba_1, \ba_2, \ldots$ be the integer vectors to be generated in the subsequent iterations of the loop from~\ref{step:NF_loop} (we consider an infinite sequence regardless of the actual number of iterations).
    We will prove that if
    \begin{equation}\label{eq:normalform_correctness}
    b(\ba_1)\ldots b(\ba_{\ell}) b(\ba_{\ell + 1}) q(\ba_1, \ldots, \ba_{\ell}) \neq 0,
    \end{equation}
    then the algorithm will return correct result at the $(\ell + 1)$-st iteration of the loop at~\ref{step:NF_loop}.
    Consider the $i$-th iteration of the loop.
    We will denote the values of $V$ and $s$ at the beginning of the iteration by $V_i$ and~$s_i$, respectively.
    Since $b(\ba_i) \neq 0$, computation of the normal form with respect to $\eOMS_\bg$ and evaluation at $\ba_i$ commute.
    Hence, the kernel of the map $\NFp_{\eOMS_\bg(\ba_i)}$ restricted to the polynomials of degree $\leqslant \delta$ will be in bijective correspondence with $k$-linear relations between the rows of $M(\ba_i)$.
    Therefore, $V_{i + 1}$ will be isomorphic to the space of all $k$-linear relations between the rows of $M(\ba_1, \ldots, \ba_i)$.
    Since $q_i(\ba_1, \ldots, \ba_i) \neq 0$, we will have $s_{i + 1} = r_i$.
    Thus, the algorithm will not terminate at the first $\ell$ iterations but will terminate at the $(\ell + 1)$-st.
    Furthermore, since the elements of the resulting $V$ will be in correspondence with the $k$-relations on the rows of $M(\ba_1, \ldots, \ba_{\ell + 1})$, the final $V$ will contain the space of polynomials of degree $\leqslant \delta$ in $k(\bg)$.
    Since $\dim V = r$, this containment will be in fact an equality.

    Thus, if the condition~\eqref{eq:normalform_correctness} is fulfilled, the algorithm will return a correct result.
    By the Schwarz-Zippel lemma (\Cref{lemma:sz}), the probability of this can be made arbitrary close to 1 by increasing the size of $S$.
\end{proof}

\begin{remark}[Modular version (cf.~\Cref{rem:modular_membership})]
\label{rem:modular_low_degree}
    In practice, the computation will be performed modulo a prime number $p$ (and the coordinates for $\ba$'s will be sampled from the whole $\mathbb{F}_p$).
    The proof of \Cref{prop:polynomials_correctness} implies that, for any given $\varepsilon > 0$, the probability of correctness will be at least $1 - \varepsilon$ for large enough $p$.
\end{remark}

\begin{remark}[{On computation of $k(\bg) \cap k[\bx]$}]
    Polynomials in $k(\bg)$ form not only a linear subspace but also a subalgebra. 
    It is therefore tempting not to restrict ourselves to fixed $\delta$ but compute generators of this subalgebra.
    However, as the negative solution to the Hilbert 14-th problem~\cite{Nagata1959} shows, this subalgebra may be not finitely generated (even in the three-variable case~\cite{Kuroda2005}).
\end{remark}


\section{Main algorithm}\label{sec:main_algo}

\subsection{Algorithm}
In this section, we put the subroutines described above together to present our main algorithm.

\begin{algorithm}[H]
\caption{Finding simple generators of a rational function field}
\label{alg:main}
\begin{description}[itemsep=0pt]
\item[Input: ] $\bg = (g_1,\ldots,g_m) \subset k(\bx)$, where $\operatorname{char} k = 0$ and we denote $\mathcal{E} \coloneqq k(\bg)$; a real number $0 < \varepsilon < 1$.
\item[Output: ] $\bh = (h_1,\ldots, h_\ell)\subset k(\bx)$, such that $k(\mathbf{h}) = k(\mathbf{g})$ (and $\mathbf{h}$ is potentially simpler than $\mathbf{g}$).
\end{description}

\begin{enumerate}[label = \textbf{(Step~\arabic*)}, leftmargin=*, align=left, labelsep=\favoritelabelsep, itemsep=\favoriteitemsep]
    \item\label{step:increment} Set $d := 1$ and $\tau \coloneqq 1$ and choose a finite subset $S \subset k \setminus \{0\}$.
    \item \label{main:step:gb} Compute $c_1,\ldots,c_M\in k(\bx)$, the coefficients up to degree $d$ of the \grobner{} basis of $\eOMS_{\bg}$, by applying \Cref{alg:gb} to the generators of $\eOMS_{\bg}$ and $d$, $S$.
    \item \label{main:step:check-after-gb} Check if $\mathcal{E} = k(c_1,\ldots,c_M)$ with probability at least $1-\frac{\varepsilon}{\tau}$ using \Cref{alg:membership} and \Cref{rem:field_equality}.
    If the equality does not hold, set $d := 2d$ and \textbf{go to}~\ref{main:step:gb}.
    \item\label{main:step:linear-relations} Compute $m_1,\ldots,m_N \in k(\bx)$, 
    polynomials up to degree $\delta$ that belong to $\cE$, 
    by applying~\Cref{alg:polynomials} to $c_1,\ldots,c_M$ with 
    $\delta$
    set to a small constant, say 2 or 3, and the set $S$.
    \item \label{step:sort} Sort the functions $(g_1,\ldots,g_m, c_1,\ldots,c_M, m_1,\ldots, m_N)$ using a heuristic simplicity measure
    and obtain a list $(h_1, \ldots, h_{m + N + M})$; 
    for example, as in \Cref{remark:onsimplicity}.
    \item \label{main:step:kind-of-enforce-minimality} 
    For each $i = 1, \ldots, m + N + M$, check if $h_i \in k(h_1, \ldots, h_{i - 1})$ with probability at least $1-\frac{\varepsilon}{\tau}$ using \Cref{alg:membership}.
    Denote the list of the functions for which the inclusion did not hold by $(\hat{h}_1, \ldots, \hat{h}_s)$.  Optionally, also minimize
    by applying \Cref{alg:minimization} to $(\hat{h}_s, \ldots, \hat{h}_1)$.
    \item\label{main:step:final_check} Check if $\cE = k(\hat{h}_1, \ldots, \hat{h}_s)$ with probability at least $1 - \frac{\varepsilon}{2^\tau}$ using \Cref{alg:membership} and \Cref{rem:field_equality}.
    \item \label{step:main:last} If the check has passed, \textbf{return} $\bh \coloneq (\hat{h}_1, \ldots, \hat{h}_s)$.
    Otherwise, choose a larger subset $S \subset k \setminus \{0\}$, set $\tau \coloneqq \tau + 1$, and {\bf go to}~\ref{main:step:gb}. 
    \end{enumerate}
\end{algorithm}
\begin{proposition}
\label{prop:mainalg}
    \Cref{alg:main} terminates with probability one. The output of \Cref{alg:main} is correct with probability at least $1 - \varepsilon$.
\end{proposition}
\begin{proof}
    We prove the probability of termination. The algorithm does not terminate if and only if the check in \ref{main:step:final_check} never passes. 
    The probability of that is $p_1 \cdot p_2 \cdot \ldots$, where at the $\tau$-th iteration of the algorithm $p_\tau$ is the probability that \ref{main:step:final_check} does not pass. 
    To show that this product is zero, it suffices to show that there exists a positive integer
    $n$ such that for every $\tau > n$ we have $p_\tau < 0.42$.
    We denote by $e_0$ the smallest positive integer such that the coefficients up to degree $2^{e_0}$ of the \grobner{} basis of $\eOMS_{\bg}$ generate the whole $\mathcal{E}$.
    Then, we fix these coefficients (and denote their number by $M_0$), and let $\rho(|S|)$ be the probability that~\Cref{alg:polynomials} applied to these coefficients returns a correct result.
    Let $N_0$ be the number of polynomial generators in this correct result.
    We fix a positive integer $\tau$ and consider the 
    $\tau$-th iteration of the algorithm (in other words, the pass from~\ref{main:step:gb} to~\ref{main:step:final_check} preceding the $\tau$-th check in~\ref{main:step:final_check}.)
    Consider the following events:
    \begin{itemize}
        \item The applications of~\Cref{alg:gb} and~\Cref{alg:membership} in~\ref{main:step:gb} and~\ref{main:step:check-after-gb} return correct result.
        In this case, there 
        will be 
        $e_0 + 1$ applications of these algorithms.
        The probability of correctness of applications of~\Cref{alg:gb} is of the form $f(|S|)$ where $f(|S|) \to 1$ when $|S|\to \infty$.
        Since $S$ becomes larger at each step, this probability tends to $1$ as $\tau \to \infty$.
        The probability of correctness of the applications of~\Cref{alg:membership} is at least $1 - \frac{e_0\varepsilon}{\tau}$, which also tends to $1$ as $\tau \to \infty$.
        \item The application of~\Cref{alg:polynomials} at~\ref{main:step:linear-relations} is correct. This happens with probability $\rho(|S|) \to 1$ as $\tau \to \infty$ by \Cref{prop:polynomials_correctness}.
        \item All the applications of~\Cref{alg:membership} in~\ref{main:step:kind-of-enforce-minimality} return correct result.
        This happens with probability at least $1 - \frac{(m + M_0 + N_0)\varepsilon}{\tau}$.
        This value tends to $1$ as $\tau \to \infty$.
        \item The check at~\ref{main:step:final_check} returns correct result.
    \end{itemize}
    If all these events occur, then the check at~\ref{main:step:final_check} passes.
    By the probability calculations above, the probability of this tends to 1 as $\tau$ approaches infinity.
    This means that $p_\tau$ tends to 0 as $\tau \to \infty$.
    In particular, there exists $n$ such that $p_\tau < 0.42$ for any $\tau > n$.
    
    We prove the probability of correctness.  
    At the $\tau$-th iteration of the algorithm, let $B_\tau$ denote the event that $\mathcal{E} = k(\hat{h}_1,\ldots,\hat{h}_s)$, and let $E_\tau$ denote the event that the check in \ref{main:step:final_check} passes (that is, it reports that $\mathcal{E} = k(\hat{h}_1,\ldots,\hat{h}_s)$).
    At the $\tau$-th iteration, we have:
    \[
    \begin{aligned}
        P(\text{\ref{main:step:final_check} is correct}) &= P(E_\tau \mid B_\tau)P(B_\tau) + P(\widebar{E_\tau} \mid \widebar{B_\tau}) P(\widebar{B_\tau})  &\geqslant 1 - \frac{\varepsilon}{2^\tau}\\
        &\Rightarrow (1 - P(\widebar{E_\tau} \mid B_\tau))P(B_\tau) + (1 - P(E_\tau \mid \widebar{B_\tau})) P(\widebar{B_\tau})  &\geqslant 1 - \frac{\varepsilon}{2^\tau}\\
         &\Rightarrow P(\widebar{E_\tau} \mid B_\tau) P(B_\tau) + P(E_\tau \mid \widebar{B_\tau}) P(\widebar{B_\tau}) &\leqslant \frac{\varepsilon}{2^\tau}\\
        &\Rightarrow P(E_\tau \mid \widebar{B_\tau}) P(\widebar{B_\tau})  &\leqslant \frac{\varepsilon}{2^\tau}\\
        &\Rightarrow 
        P(E_\tau \cap \widebar{B_\tau}) &\leqslant \frac{\varepsilon}{2^\tau}
    \end{aligned}
    \]
    Note that, at the $\tau$-th iteration, $E_\tau \cap \widebar{B_\tau}$ is the event that an incorrect result is returned from the algorithm. 
    
    Let $T_\tau$ denote the event that the algorithm terminates at the $\tau$-th iteration with an incorrect result and that it did not terminate at every previous iteration. Then,
    \[
    P(T_\tau) \leqslant P(E_\tau \cap \widebar{B_\tau}) \leqslant \frac{\varepsilon}{2^\tau}.
    \]
    Since iterations of the algorithm are independent, the probability of an incorrect result
    is
    \[
    T_1 + T_2 + \ldots \leqslant \frac{\varepsilon}{2} + \frac{\varepsilon}{4} + \ldots = \varepsilon \left(\frac{1}{2} + \frac{1}{4} + \ldots\right) \leqslant \varepsilon.
    \]
    Since the algorithm terminates with probability one, the probability of a correct result and the probability of an incorrect result sum up to one. Therefore, the probability of a correct result is at least $1 - \varepsilon$.
\end{proof}

\begin{remark}[Retaining the original generators]
\label{remark:retain}
    In \ref{step:sort}, we extend the list of candidate functions with
    the original $g_1,\ldots,g_m$. Roughly speaking, this ensures that the output of \Cref{alg:main} is not worse than the input in terms of being simple. For example, consider $\bg = (x^i + y^i + z^i + u^i + v^i,~i=1,\ldots,5)$. The output of our algorithm for such input (using the heuristic from \Cref{remark:onsimplicity}) is
    \[
    \begin{aligned}
    &x + y + z + u + v, x^2 + y^2 + z^2 + u^2 + v^2, x^3 + y^3 + z^3 + u^3 + v^3,\\
    &x y z u + x y z v + x y u v + x z u v + y z u v,
 x y z u v.
    \end{aligned}
    \]
    However, if we did not add the original generators to the list of candidates,
    the output would be 
    less simple (both visually and according to the chosen heuristic):
    \[
    \begin{aligned}
    &x + y + z + u + v,
    x y + x z + x u + x v + y z + y u + y v + z u + z v + u v,\\
    &x y z + x y u + x y v + x z u + x z v + x u v + y z u + y z v + y u v + z u v,\\
    &x y z u + x y z v + x y u v + x z u v + y z u v,
 x y z u v.
    \end{aligned}
    \]
\end{remark}

\begin{remark}[Modular version for $k = \QQ$]
    \label{remark:all_is_modular}
    In practice, we run the steps~\ref{step:increment}--\ref{main:step:kind-of-enforce-minimality} of~\Cref{alg:main} modulo a large prime number. 
    Only the final check in~\ref{main:step:final_check} is run over 
    $k = \QQ$, and its purpose is to guarantee the probability of correctness, as \Cref{prop:mainalg} would still hold in this case. Thus just prior to the final check we reconstruct the intermediate results using rational number reconstruction~\cite[Section 5.10]{MCA}, \cite{ratrec82}.
    The use of \Cref{alg:membership,alg:minimization,alg:polynomials} modulo a prime is justified due to \Cref{rem:modular_membership,rem:modular_low_degree}.
\end{remark}

\begin{remark}[Lazy specialization of $\eOMS_{\bg}$]
\label{remark:lazy-construction}
    In practice, we do not explicitly construct the generators of $\eOMS_{\bg}$, since they may be large.
    Indeed, let $p/q \in \bg$ with coprime $p,q \in k[\bx]$. 
    The corresponding generator of $\eOMS_{\bg}$ is $p(\by) q(\bx) - q(\by) p(\bx)$ and it has $O(s_p s_q)$ terms when considered as a polynomial in $k[\bx,\by]$. In comparison, the corresponding generator of $\eOMS_{\bg}(\ba)$ for $\ba \in k^n$ has only $O(s_p + s_q)$ terms. Therefore, whenever we need to compute a \grobner{} basis of $\eOMS_{\bg}(\ba)$ for $\ba \in k^n$, we construct the generators of $\eOMS_{\bg}(\ba)$ on the fly by plugging $\bx \to \ba$ in~\eqref{eq:lazy_OMS}.
\end{remark}

\begin{remark}[Computing with zero-dimensional ideals]
    Field membership (\Cref{alg:membership}) specializes some indeterminates in the OMS ideal until the ideal becomes
    zero-dimensional. 
    In theory, a similar specialization could be used for computing the coefficients of \grobner{} bases in \Cref{alg:gb}. 
    In this case, in some polynomials in the resulting ideal, the monomials may glue together, and their coefficients (which are multivariate rational functions) combined may become large. In our experiments, we observed that this makes interpolation harder, and overall slows down the algorithm despite using zero-dimensional ideals.
\end{remark}

\begin{remark}
\label{remark:reuse}
    After we have computed the coefficients of the \grobner{} basis that generate the field, we can use them instead of the original generators in the subsequent steps of the algorithm (as in \ref{main:step:linear-relations}). This typically speeds up the computation when these coefficients are simpler than the original generators.
\end{remark}

\begin{remark}[Achieving minimality]\label{rem:minimality}
    The generating set computed by \Cref{alg:main} may be not minimal with respect to inclusion (see~\Cref{ex:compmodel}) despite the filtering performed at~\ref{main:step:kind-of-enforce-minimality}.
    Since the filtering checks membership 
    against
    the subfield generated by the elements with smaller indices only, there could still be a redundant element at a smaller index that can be expressed in terms of all other generators.
    The minimality with respect to inclusion can be achieved using~\Cref{alg:minimization}.
    However, in some cases, the larger set may be preferable, see the discussion in~\Cref{ex:compmodel}.
\end{remark}


\subsection{Example}

In this section, we run \Cref{alg:main} on an example and trace through the steps of the algorithm.
We consider an example from the area of structural identifiability of dynamical models (\Cref{sec:app_ode}) with the code name
\exname{SEIR34},
collected in~\cite{MASSONIS2021441}. The task consists of finding a simple generating set of the subfield $\mathcal{E} \subseteq \mathbb{Q}(k, N, \beta, \epsilon, \gamma, \mu, r)$ generated by
{\footnotesize\[
\begin{aligned}
    &N, \epsilon+\gamma+3 \mu, -\epsilon-\gamma-2 \mu, \epsilon \mu+\gamma \mu+2 \mu^{2}, \epsilon \gamma \mu+\epsilon \mu^{2}+\gamma \mu^{2}+\mu^{3}, \\
    &\dfrac{\beta r}{k N}, \dfrac{-\beta \epsilon r}{N}, \dfrac{\beta \epsilon r+\beta \gamma r+2 \beta \mu r}{k N}, \dfrac{\beta \epsilon \gamma r+\beta \epsilon \mu r+\beta \gamma \mu r+\beta \mu^{2} r}{k N}.
\end{aligned}
\]}These nine functions are the input to our \Cref{alg:main}. 
When executing the steps of the algorithm, we will suppose that all steps return a correct result. When computing \grobner{} bases, we will be using the degrevlex monomial ordering with $t > k > N > \beta > \epsilon > \gamma > \mu > r$ (where $t$ comes from the Rabinowitsch
trick).

In \ref{main:step:gb}, the algorithm computes the coefficients of degree at most $d = 1$ of the \grobner{} basis of the OMS ideal and gets 
\begin{equation}
\label{eq:ex:coeffs-1}
(c_1,\ldots,c_M) \coloneqq (\mu, \epsilon + \gamma, N).    
\end{equation}
These do not generate $\mathcal{E}$, thus the check in \ref{main:step:check-after-gb} does not pass. 
Therefore, we run \ref{main:step:gb} again, now with $d=2$, and obtain the same result as~\eqref{eq:ex:coeffs-1} but additionally with two coefficients of higher degree: $\epsilon \gamma$ and $\frac{k}{\gamma}$.
The check in \ref{main:step:check-after-gb} still fails, so we run again with $d=4$ and discover an additional coefficient: $\frac{\beta r}{\gamma}$. At this point, 
\begin{equation*}
(c_1,\ldots,c_M) \coloneqq \left(\mu, \epsilon + \gamma, N, \epsilon \gamma, \frac{k}{\gamma}, \frac{\beta r}{\gamma}\right). 
\end{equation*}
Finally, the check signals that $\mathcal{E} = \mathbb{Q}(c_1,\ldots,c_M)$.

Then, \ref{main:step:linear-relations} computes some linearly independent polynomials up to degree $\delta = 2$ that belong to $\mathcal{E}$:
\[
(m_1,\ldots,m_N) \coloneqq (\mu,  \epsilon + \gamma, N, 
k \epsilon, \epsilon \gamma).
\]

Then, in \ref{step:sort}, the original generators and $(c_1,\ldots,c_M)$ and $(m_1,\ldots,m_N)$ are combined in one list, and the list is ordered using the heuristic from \Cref{remark:onsimplicity}, which ensures that simpler elements come first. 
In this example, we omit the original generators, as adding them does not change the final result. 
After removing duplicates, we have
\[
(h_1,\ldots,h_{N'+M'}) \coloneqq \left(\mu, N, \epsilon + \gamma, \epsilon \gamma, k \epsilon, \frac{k}{\gamma}, \frac{\beta r}{\gamma}\right).
\]
Note that $\mathcal{E} = \mathbb{Q}(h_1,\ldots,h_{N'+M'})$. The purpose of \ref{main:step:kind-of-enforce-minimality} is to pick the simplest elements from this list that would form the final result. To this end, the algorithm makes a pass over the list and removes $\frac{k}{\gamma}$ since it can be expressed in terms of the previous elements in the list. We obtain the simplified generating set (minimization from \Cref{alg:minimization} has no effect):
\[
(\hat{h}_1,\ldots,\hat{h}_s) \coloneqq \left(\mu, N, \epsilon + \gamma, \epsilon \gamma, k \epsilon, \frac{\beta r}{\gamma}\right).
\]
These functions are the output of our algorithm.

This result requires both the coefficients of the \grobner{} basis and some of the polynomial generators. In particular, the function $k \epsilon$ is not present among the coefficients of the \grobner{} basis. In turn, the subfield cannot be generated by $m_1,\ldots,m_N$ alone, thus computing some coefficients of the \grobner{} basis was necessary.

\begin{remark}[Algorithm parameters]
    We can increase the effort committed to simplification by computing \grobner{} bases of OMS ideals in several monomial orderings, as in \Cref{ex:duralexsedlex}, and by increasing the parameter $\delta$ of \Cref{alg:polynomials}, and then combining the results. In this section's example, if we used $\delta = 3$ in \ref{main:step:linear-relations}, we would additionally find that $\beta \epsilon r \in \mathcal{E}$, and hence obtain a generating set that contains only polynomials:
$
\mu, N, \epsilon + \gamma, \epsilon \gamma, k \epsilon, \beta \epsilon r.
$
\end{remark}

\section{Implementation and performance}
\label{section:experimental}

\subsection{Implementation}

We have implemented~\Cref{alg:main} and the sub-algorithms over the field of rational numbers $k = \QQ$ in Julia language~\cite{bezanson2017julia}. Our implementation has a convenient user interface and is freely available in the package \code{RationalFunctionFields.jl} at
\begin{center}
\url{https://github.com/pogudingleb/RationalFunctionFields.jl}~\footnote{The URL refers to the most up-to-date version of the code. The code as it had been at the time of this publication can be accessed 
at git revision {\tt f7d4943aa7432bd10c206b59f63909aca52c909d} at: 
\url{https://github.com/pogudingleb/RationalFunctionFields.jl/tree/f7d4943aa7432bd10c206b59f63909aca52c909d}. In particular, the directory {\tt paper} contains the sources of our experiments and examples.
}
\end{center}
One essential component is the algorithm for interpolating the coefficients of a \grobner{} basis over the rational function fields, which is implemented in the sub-package ParamPunPam.jl. 
We use \code{Nemo.jl}~\cite{AbstractAlgebra.jl-2017} for polynomial arithmetic and \code{Groebner.jl}~\cite{groebnerjl2023} for computing \grobner{} bases over integers modulo a prime.

The experiments in this paper were run on a Linux machine with 128 GB of RAM and 32 × 13th Gen i9-13900+ CPU. We used Julia 1.12.4, Singular 4.4.1, and Maple 2025.

\subsection{Modular computation}
Input and output  in our implementation
of the main algorithm 
are rational functions over the rationals. Per \Cref{remark:all_is_modular}, for better performance, intermediate computations are carried out over finite fields. 
The final correctness check in \ref{main:step:final_check} computes \grobner{} basis and normal forms over the rationals. 
Typically, the size of the rational numbers in the output of our algorithm is small in our applications (\Cref{sec:app}); this is also observed in practice (see, e.g., \Cref{table:simplify} and \Cref{app:simplify_everyone}).

\begin{remark}[Reconciling with reality]
    In our implementation, the final correctness check in \ref{main:step:final_check} computes \grobner{} basis over the rationals using a Monte-Carlo probabilistic algorithm based on multi-modular computation~\cite{modular}, for which  no practical bound on the probability of correctness is known. For this reason, our implementation may not guarantee the probability of correctness stated in \Cref{prop:mainalg}.
\end{remark}

\subsection{Sparse interpolation over finite fields}
\label{sec:impl:interpol}
For interpolating sparse multivariate rational functions over finite fields, we implemented the algorithm by van der Hoeven and Lecerf~\cite{vdhl} discussed in~\Cref{sec:prelim:ratinterpol}. We also implemented the algorithm by Cuyt and Lee~\cite{cuyt-lee} because, in theory, it may require fewer blackbox evaluations; see \cite[Remark 3.]{vdhl}. In most of our examples, the algorithms used the same number of evaluations, so we prefer van der Hoeven and Lecerf's algorithm because it was easier to implement efficiently.

For interpolating sparse multivariate polynomials over a finite field $\mathbb{F}_q$, we use the Ben-Or and Tiwari algorithm~\cite{ben-or-tiwari} with the admissible ratio $\bomega = (2,3,5,\ldots,p_n)$ and recover exponents via integer factoring (see~\Cref{sec:tiwari}). Therefore, to interpolate polynomials of degree $d$, we must have $p_n^d < q$. 
To facilitate efficient arithmetic in $\mathbb{F}_q$, we choose $q$ to be a prime that fits in $64$ bits. 
For example, with $d=8$, the largest number of variables such that $p_n^d < 2^{64}$ is $n=54$, which has been sufficient in our use cases. 
An alternative would have been Kronecker substitution (\cite[Section 8.4.]{MCA},~\cite{kronecker}), which requires, roughly, $d^n < q$. 
Our choice between the two is motivated by the fact that we usually have $n > d$.

We double the number of evaluations until interpolation succeeds. 
We found this approach easier to implement in an efficient way than the early termination strategies (as, for example, in~\cite{adaptive-step,early-termination-lee}), since we could rely on the existing P{\'a}de approximation and root-finding functionality from Nemo.jl.
Thus, our implementation overestimates the required number of blackbox evaluations by at most a factor of two.

\subsection{Computing small degree coefficients of \grobner{} bases}
\label{subsec:small-degree}

An appealing property of our algorithm is that the degree of the interpolated coefficients of the \grobner{} basis of the OMS ideal is well-controlled.
In \Cref{alg:main}, \ref{main:step:gb} and~\ref{main:step:check-after-gb} ensure that the coefficients of high degree are not interpolated when they are not necessary to generate the field (as, e.g., in~\Cref{OMS:excess}).
Along with being more efficient, this method also helps with simplification by filtering out generators of high degrees before even computing them.

\begin{table}[H]
\centering
    \begin{tabular}{l|rr|rr|rr}
    \multirow{2}{*}{Example} & \multicolumn{2}{c|}{Input} & \multicolumn{2}{c|}{Full \grobner{} basis} & \multicolumn{2}{c}{\ref{main:step:gb} and~\ref{main:step:check-after-gb}}\\
    & \# vars. & degrees & degrees & \# evals. & degrees & \# evals. \\
    \hline
    \exname{MAPK-5}~\cite{MANRAI20085533} & 22 & 8, 6 & 12, 0 & 15,401 & 1, 0 & 16\\
    \exname{Pharm}~\cite{Demignot1987Pharm} & 7 & 26, 5 & 2, 0 & 20 & 2, 0 & 20\\
    \exname{Fujita}~\cite{Fujita2010} & 16 & 9, 5 & 4, 7 & 129 & 1, 1 & 42\\
    \exname{Bilirubin}~\cite{combos} & 7 & 
    4, 0
    & 4, 2 & 278 & 4, 2 & 278\\
    \exname{EAIHRD}~\cite{Fokas2020} & 10 & 
    22, 6
    & 10, 9 & n/a & 4, 4 & 342\\
    \exname{SLIQR}~\cite{Dankwa2021} & 6 & 
    6, 0
    & 7, 6 & 502 & 7, 6 & 502\\
    \exname{LinComp2}~\cite{ahmed2025identifiabilitydirectedcyclecatenarylinear} & 7 & 
    5, 0
    & 381, 373 & n/a & 4, 2 & 4118
    \end{tabular}
    \caption{The number of blackbox evaluations required to interpolate the coefficients of \grobner{} bases of OMS ideals, using two methods: computing the full \grobner{} basis with~\Cref{alg:gb}, and using~\Cref{alg:gb} with the same $d$ as in~\ref{main:step:gb} and~\ref{main:step:check-after-gb} of~\Cref{alg:main}. We write n/a if the computation is not feasible.}
    \label{table:benchmarks:coeffs}
\end{table}

We will illustrate this claim using 
several examples arising from the problem of structural identifiability (see~\Cref{sec:app_ode}). 
Each example is given as a list of rational functions -- the generating set to be simplified.
We will consider two methods.
For each example, we apply~\Cref{alg:gb} to the generators of the OMS ideal of the input functions, but the first method computes the full \grobner{} basis (by specifying some large degree bound $d$), and the second 
method
uses the degree $d$ from~\ref{main:step:gb} and~\ref{main:step:check-after-gb} of~\Cref{alg:main}.

In~\Cref{table:benchmarks:coeffs}, for each example, we report: the number of variables; the largest degree of the numerator and denominator among the input rational functions and the output rational functions, for both methods; and the number of evaluations required to interpolate the output, for both methods.

From~\Cref{table:benchmarks:coeffs}, we may conclude that
\begin{itemize}
    \item The maximal degrees in the output typically are not large (even in the full \grobner{} bases), and often are smaller than those of the input. This supports
    the expectations
    that OMS ideals are useful for simplification.
    \item In some cases, the field can be generated by elements with degrees smaller than that of the full \grobner{} basis, and exploiting this leads to fewer evaluations required for interpolation. For example, this makes the computation feasible for the \exname{EAIHRD} and \exname{LinComp2} examples, and reduces the running time from 16 minutes to 2 seconds for the \exname{MAPK-5} example.
\end{itemize}

\begin{remark}[Lowering the degrees even further]
The strategy of doubling the guess for the degree $d$ until the check in \ref{main:step:check-after-gb} passes 
may be suboptimal. For instance, in the \exname{SLIQR} example in \Cref{table:benchmarks:coeffs}, the maximal total degrees of the numerator and denominator of a coefficient in the full \grobner{} basis are $7$ and $6$, respectively; the degrees obtained
in \ref{main:step:gb} and \ref{main:step:check-after-gb}  
are the same. However, the check in \ref{main:step:check-after-gb} would pass for degrees as low as ${(5,4)}$. 
Thus, with the current strategy, we interpolate coefficients of higher degrees than necessary. 
This may significantly increase the number of blackbox evaluations as the maximum number of terms $\binom{n+d}{n}$ grows rapidly with the degree. 
Since the check in \ref{main:step:check-after-gb} is cheap, we could
search for a sharper guess for the degree, for example, by increasing $d$ in arithmetic progression with a small step. The current strategy is a compromise with the fact that we do not cache 
blackbox evaluations across the calls to \Cref{alg:gb}, thus we would like to minimize the number of iterations.
\end{remark}

\subsection{Speeding up computation of many \grobner{} bases}\label{sec:tracing}

The efficiency of~\Cref{alg:gb} depends on the ability to evaluate the blackboxes many times, where each evaluation involves computing a \grobner{} basis over a finite field. 
As we have seen, one might require many evaluations to interpolate the result.

We use \grobner{} tracing, as described by Traverso~\cite{tracingtrav}, to speed up the computation of \grobner{} bases. The presentation in the original paper is general, however, for completeness, we briefly present the main ideas while adapting them to our case.  

\begin{table}[H]
\centering
    \begin{tabular}{l|r|r|r}
    Example & F4 algorithm & \thead{F4 with tracing\\{\it(Apply)}} & Speedup\\
    \hline
    \exname{MAPK-5}~\cite{MANRAI20085533} & \SI{39}{\milli\second} & \SI{8}{\milli\second} & $4.9$ \\
    \exname{Pharm}~\cite{Demignot1987Pharm} & \SI{1.91}{\second} & \SI{0.36}{\second} & $5.3$ \\
    \exname{Fujita}~\cite{Fujita2010} & \SI{1.89}{\milli\second} & \SI{0.54}{\milli\second} & $3.5$ \\
    \exname{Bilirubin}~\cite{combos} & \SI{334}{\micro\second} & \SI{129}{\micro\second} & $2.6$ \\
    \exname{EAIHRD}~\cite{Fokas2020} & \SI{680}{\milli\second} & \SI{112}{\milli\second} & $6.1$\\
    \exname{SLIQR}~\cite{Dankwa2021} & \SI{116}{\micro\second} & \SI{39}{\micro\second} & $3.1$\\
    \exname{LinComp2}~\cite{ahmed2025identifiabilitydirectedcyclecatenarylinear} & \SI{68}{\milli\second} & \SI{37}{\milli\second} & $1.8$\\
    \end{tabular}
    \caption{The time to compute a single \grobner{} basis over integers modulo a prime using the F4 algorithm with and without tracing. We denote microseconds by \SI{}{\micro\second}.}
    \label{table:benchmarks:trace}
\end{table}

We describe the method for OMS ideals but the algorithm can be applied to any ideal over a rational function field. Let $\bg = (g_1,\ldots,g_m) \subset k(\bx)$ with $\bx = (x_1,\ldots,x_n)$ and $\eOMS_{\bg} \subset k(\bx)[\by]$ be its OMS ideal. Recall that $\eOMS_{\bg}(\ba)$ denotes the specialization of the OMS ideal under $\bx \to \ba$, where $\ba\in k^n$. Fix a monomial ordering on $\by$. 
Let $\ba_1,\ldots,\ba_\ell \in k^n$. 
For $i = 1,\ldots,\ell$, we wish to compute
the reduced \grobner{} basis of $\eOMS_{\bg}(\ba_i)$.
With the use of tracing, the computation is organized in two steps:
\begin{enumerate}
    \item[{\it(Learn)}] Run a modified Buchberger's algorithm to compute the \grobner{} basis of $\eOMS_{\bg}({\ba}_1)$. The modified algorithm additionally records the information about which S-polynomials reduced to zero; this information could be stored as a collection of indices of the corresponding iterations of the algorithm.
    \item[{\it(Apply)}] For $i=2,\ldots,\ell$, run a (differently) modified Buchberger's algorithm to compute the \grobner{} basis of $\eOMS_{\bg}({\ba}_i)$.
    This version runs as usual but does not reduce the S-polynomials that correspond to the ones that reduced to zero during the {\it(Learn)} stage,
    since those are expected to reduce to zero.
\end{enumerate}
Correctness of this method
can be analyzed using \Cref{lem:GB_spec}. 
Indeed, there exists a nonzero polynomial $g \in k(\bx)$ such that for every $\bb_1, \bb_2 \in k^n$ 
if
$g(\bb_1) \neq 0$ and $g(\bb_2) \neq 0$ then
the Buchberger's algorithm applied to the generators of $\eOMS_{\bg}({\bb}_1)$ and $\eOMS_{\bg}({\bb}_2)$ produces the same S-polynomials and the same reductions of S-polynomials (that is, with the same support but possibly different values of coefficients). Thus, when for every $i=1,\ldots,\ell$ we have $g(\ba_i) \neq 0$, the method correctly computes all \grobner{} bases. 

\grobner{} tracing fits naturally in~\Cref{alg:gb}. Of course, we might not be able to ensure that for every $i = 1,\ldots,\ell$ we have $g(\ba_i) \neq 0$, but the probabilistic statement in \Cref{prop:propa-gb-over-rat-funcs} holds.
In our implementation, $k$ is the field of integers modulo a prime, where the prime fits in $64$ bits.
We use the F4 algorithm~\cite{F4} instead of the Buchberger's algorithm for more
efficiency. 
Tracing in the F4 algorithm can be arranged similarly, with a difference that instead of S-polynomials we have rows of a Macaulay matrix.
We use the package \code{Groebner.jl}~\cite{groebnerjl2023}, which implements the F4 algorithm and tracing as described above.

In~\Cref{table:benchmarks:trace}, we illustrate the performance benefit of using \grobner{} tracing to compute a single \grobner{} basis in degrevlex over integers modulo a prime with our implementation. 
We report the running time of a classic F4 algorithm and of the F4 algorithm modified as described in the {\it(Apply)} bullet point (where we assume that the {\it(Learn)} part was already ran and we have access to its output). The running times are short so we compute them as averages of many runs.
We observe a consistent speedup with the factor ranging from $1.8$ to $6.1$ across different problem sizes. This translates into a similar speedup for the whole \Cref{alg:gb}, whose running time is typically dominated by many runs of {\it(Apply)}.

\begin{remark}
    A different modification of the F4 algorithm that also exploits the low-rank structure of the matrices relies on probabilistic linear algebra over finite fields~\cite[Section 2]{splitting}\footnote{We thank Roman Pearce for pointing out this observation to us.}.
    However, unexpectedly, for many problems considered in this paper (in particular, from \Cref{table:benchmarks:trace}) our implementation of this method does not provide a significant speedup over the classic F4 algorithm.
\end{remark}

\begin{remark}
    A variant of the F4 algorithm that uses \grobner{} tracing for computations over rational function fields has also been discussed in~\cite{f4-variant-tracing}. Their application is different, but we may compare the speedups obtained from tracing for a common example. For the \exname{Katsura-11} and \exname{Katsura-12} problems, over integers modulo a 32-bit prime, the quotients of the running times $\frac{\text{F4 algorithm}}{\text{F4 with tracing}}$ are $3.9$ and $5.1$, respectively, for the implementation in~\cite[Fig.3.]{f4-variant-tracing}; for our implementation, they are $10.7$ and $11.7$, respectively. 
    These appear to be of the similar order of magnitude; we also observe that tracing seems to be more beneficial for larger problems.
\end{remark}

\subsection{Comparison with other methods for computing \grobner{} bases}\label{sec:other_gb}

While computing the full \grobner{} basis with rational function coefficients is not the main focus of the present paper, \Cref{alg:gb} allows doing this.
Therefore, it may be interesting to compare this interpolation-based algorithm with other available methods. 

\begin{table}[!htbp]
\centering
    \begin{tabular}{l|rrrrrr}
    Example &  F4 (direct) & F4 (flat ring) & slimgb~\cite{slimgb} & ffmodStd~\cite{ffmodstd} & \Cref{alg:gb}
    \\
    \hline
    \exname{Simson-3}~\cite{symbolicdata} & \SI{1}{\second} & \SI{1}{\second} & \SI{0}{\second} & \SI{0}{\second} & \SI{0}{\second}\\
    \exname{Boku-I4}~\cite{Boku2016} & OOM & OOM & \SI{141}{\second} & \SI{108}{\minute} & \SI{36}{\second}\\
    \exname{Param-1} & OOM & \SI{3}{\second} & \SI{3}{\second} & 
    \SI{20}{\minute}
    & OOM\\
    \exname{Param-2} & OOM & \SI{62}{\second} & \SI{0}{\second} & 
    \SI{774}{\second}
    & 
    \SI{6}{\second}
    \\
    \exname{Moments-3}~\cite{MingKueiHu1962} & OOM & OOM & 
    4 hours 
    & Timeout & \SI{81}{\second}\\
    \exname{MAPK-5}~\cite{MANRAI20085533} & \SI{14}{\second} & OOM & \SI{3}{\second} & Timeout & \SI{16}{\minute}\\
    \exname{Pharm}~\cite{Demignot1987Pharm} & \SI{10}{\minute} & OOM & \SI{15}{\second} & Timeout & \SI{100}{\second}\\
    \exname{Fujita}~\cite{Fujita2010} & \SI{0}{\second} & \SI{0}{\second} & \SI{0}{\second} & \SI{52}{\second} & \SI{0}{\second}\\
    \exname{Bilirubin}~\cite{combos} & \SI{1}{\second} & \SI{5}{\second} & \SI{0}{\second} & 113 s & \SI{0}{\second}\\
    \exname{EAIHRD}~\cite{Fokas2020} & OOM & OOM & Timeout & Timeout & Timeout\\
    \exname{SLIQR}~\cite{Dankwa2021} & \SI{0}{\second} & OOM & \SI{0}{\second} & \SI{127}{\second} & \SI{0}{\second}\\
    \exname{LinComp2}~\cite{ahmed2025identifiabilitydirectedcyclecatenarylinear} & OOM & OOM & OOM & Timeout & Timeout
    \end{tabular}
\caption{The running time of several different methods of computing \grobner{} bases in $\QQ(\bx)[\by]$ in degrevlex. Timeout means the computation was aborted after
24 hours.
OOM means the computation was aborted after using
more than $50$ GB of RAM.}
\label{table:benchmarks1}
\end{table}

We will be computing \grobner{} bases of ideals in $\QQ(\bx)[\by]$ in degrevlex ordering using the following methods:
\begin{enumerate}
    \item The F4 algorithm applied directly in $\mathbb{Q}(\bx)[\by]$. We use the implementation in \code{Groebner.jl}~\cite{groebnerjl2023}, which is generic and has no special optimizations in this case.
    \item The F4 algorithm applied in the flat ring $\mathbb{Q}[\bx,\by]$ with a suitable monomial order $<_{\bx,\by}$. For $i=1,2$ for monomials $\bx^{\balpha_i}\by^{\bbeta_i}$ with 
    $\balpha_i \in \Zplus^{n}$ and $\bbeta_i \in \Zplus^{t}$, we have: $\bx^{\balpha_1}\by^{\bbeta_1} <_{\bx,\by} \bx^{\balpha_2}\by^{\bbeta_2}$ if $\by^{\bbeta_1} <_{\operatorname{drl}} \by^{\bbeta_2}$ or $\by^{\bbeta_1} = \by^{\bbeta_2}$ and $\bx^{\balpha_1} <_{\operatorname{drl}} \bx^{\balpha_2}$, where $<_{\operatorname{drl}}$ denotes the degrevlex monomial ordering. 
    We use the implementation in \code{Groebner.jl}, which homogenizes the input generators in this case.
    \item The slimgb algorithm~\cite{slimgb} applied directly in $\mathbb{Q}(\bx)[\by]$. We use the implementation in Singular.
    \item The ffmodStd algorithm~\cite{Boku2016,ffmodstd} implemented in Singular. This is an interpolation-based algorithm similar to~\Cref{alg:gb}. ffmodStd reconstructs intermediate results to the rationals and performs interpolation over the rationals using the Cuyt-Lee's algorithm~\cite{cuyt-lee}, while we carry out all computations (including interpolation) modulo a prime and reconstruct to the rationals only at the end. ffmodStd uses the built-in modStd command~\cite{modstd} for computing \grobner{} bases in $\QQ[\by]$ and does not use \grobner{} tracing.
    \item \Cref{alg:gb}, which we use to compute the full \grobner{} basis.
\end{enumerate}

In~\Cref{table:benchmarks1}, we report the running time for each method for several examples. Each example is given by a list of generators of the ideal in $\mathbb{Q}(\bx)[\by]$. All implementations 
are run using only one CPU.
Note that the methods 2., 4., and 5. are Monte-Carlo probabilistic. 
We made sure that the outputs of the implementations agree. Examples roughly fall into groups based on their provenance:
\begin{itemize}
    \item \exname{Simson-3}, distributed in the SymbolicData project~\cite{symbolicdata}\footnote{We downloaded the systems' definitions from \url{https://github.com/symbolicdata/data}}. In the original slimgb publication, this system took the most time to compute among benchmark systems in $\QQ(\bx)[\by]$ in degrevlex ordering. Now the computation is near instant with all methods.
    \item \exname{Boku-I4}, from~\cite[Section 5.3.3]{Boku2016}. This system appeared in the benchmarks of ffmodStd, where it took the longest time to run among benchmark systems in degrevlex.
    We modify the example by introducing a fictional variable in the ground ring (the system itself does not change) to disable the optimization that speeds up the computation over univariate rational function fields in ffmodStd\footnote{With the optimization enabled, the running time of ffmodStd is $110$ seconds}. In this example, the quotient of the running times $\frac{\text{ffmodStd}}{\text{\Cref{alg:gb}}}$ is around $180$. 
    \item \exname{Param-1}, \exname{Param-2}. These were constructed from \exname{Katsura-5} by replacing some numerical coefficients with transcendental parameters. For \exname{Param-1}, our implementation of \Cref{alg:gb} has no chance to succeed, because the output contains a coefficient of total degree $46$ in $3$ variables, which is not possible to interpolate using our version of the Ben-Or and Tiwari algorithm modulo a 64-bit prime. For \exname{Param-2}, the coefficients of the \grobner{} basis contain large rationals (the numerators are 550 bits in size and the denominators are small), and our implementation uses 18 primes of size 64-bit to reconstruct the result.
    \item \exname{Moments-3}, \Cref{example:moments} from~\Cref{sec:app_invariants}.
    \item Examples arising from the problem of structural identifiability (see~\Cref{sec:app_ode}). These have the structure of $\eOMS$ ideals (\Cref{notation:eOMS}). \Cref{table:benchmarks:coeffs,table:benchmarks:trace} give an idea about the hardness of these examples, in particular we note that the \grobner{} bases for these examples are typically almost linear and have very sparse coefficients. 
\end{itemize}

Overall, computations tend to either finish quickly or not finish due to time and space constraints. As expected, straightforward
direct methods (1. and 2.) require an unpredictable amount of memory and are thus difficult to use. The slimgb algorithm fares very well. 
\Cref{alg:gb} is slower than slimgb for several examples, but is quicker for some others (e.g., \exname{Moments-3}). 
Compared to ffmodStd, our implementation appears to be a hundred or more times more efficient; however, ffmodStd handles the \exname{Param-1} example, and our implementation does not.
Finally, although none of the methods can compute the full \grobner{} basis for the \exname{EAIHRD} and \exname{LinComp2} examples, our \Cref{alg:gb} readily tackles these examples using the approach from \Cref{subsec:small-degree}. 


\section{Experimental results}
\label{sec:experiment}

\subsection{Assessing the quality of simplification}
\label{sec:experiment-simplify}

In this section, we use our implementation of the main algorithm (\Cref{alg:main}) to simplify generators of subfields of rational function fields that arise in the problem of structural identifiability (see \Cref{sec:app_ode}).

{
\small
\begin{table}[H]
\setlength{\tabcolsep}{8pt}
\renewcommand{\arraystretch}{1.0}
\centering
    \begin{tabular}{lrr}
    Example & Original generators & Result of~\Cref{alg:main}
    \\
    \midrule
    \midrule
    \Cref{ex:simple-OMS-simplify} & $x_1^2 + x_2^2, x_1^3 + x_2^3, x_1^4 + x_2^4$ & $x_1+x_2, x_1x_2$
    \\
    \midrule
    \exname{genLV}~\cite{Remien2021GLVIdentifiability} & 
    {\small
    \setlength{\tabcolsep}{0pt}
\renewcommand{\arraystretch}{1}\setlength{\extrarowheight}{0pt}
\begin{tabular}{r}$(\beta_{12} r_1 r_2 - \beta_{22} r_1^2)/\beta_{12},$\\$(-\beta_{11}^2 \beta_{22} + \beta_{11} \beta_{12} \beta_{21})/\beta_{12},$\\$(\beta_{12} r_2 - 2 \beta_{22} r_1)/\beta_{12}, (\beta_{12} + \beta_{22})/\beta_{12}, $\\$(\beta_{11} \beta_{12} r_2 - 2 \beta_{11} \beta_{22} r_1 + \beta_{12} \beta_{21} r_1)/\beta_{12}, $\\$ (-\beta_{11} \beta_{12} + 2 \beta_{11} \beta_{22} - \beta_{12} \beta_{21})/\beta_{12}$\end{tabular}
     } & $r_1, r_2, \beta_{11}, \beta_{21}, \beta_{12}/\beta_{22}$\\
    \midrule
    \exname{Pharm}~\cite{Demignot1987Pharm}  & 
    {
    \setlength{\tabcolsep}{0pt}
\renewcommand{\arraystretch}{1}\setlength{\extrarowheight}{0pt}\begin{tabular}{r}
        $3533$ functions, max. degree $31$,\\
        $2210$ max. terms,
        size on disk $35$ MB\\
     \end{tabular}
     } & $n, k_a, k_c, b_1, b_2, a_1, a_2$\\
    \midrule
    \exname{CGV1990}~\cite{Chappell1990Global} & {
    \setlength{\tabcolsep}{0pt}
\renewcommand{\arraystretch}{1}\setlength{\extrarowheight}{0pt}\begin{tabular}{r}
        $412$ functions;\\
        indeterminates $R$ and $S$ appear\\
        in functions of degree at least $14$
     \end{tabular}
     } & {
    \setlength{\tabcolsep}{0pt}
\renewcommand{\arraystretch}{1}\setlength{\extrarowheight}{0pt}\begin{tabular}{r}
        $k_3, k_4, k_6, k_7, R V_3 + S V_{36}$,\\ 
        $R V_{36} + S V_3 / 25,$\\
        $ k_5 (R V_{36} + S V_{36}) / 5,$\\$k_5 (V_3 + 5 V_{36}) / V_3$
     \end{tabular}
     }\\
     \midrule
     \exname{Fujita}~\cite{Fujita2010} & 
     {
    \setlength{\tabcolsep}{0pt}
\renewcommand{\arraystretch}{1}\setlength{\extrarowheight}{0pt}\begin{tabular}{r}
        $217$ functions;\\
        indeterminates $a_1,a_2,a_3$ appear\\
        in functions of degrees $5,3,2$, resp.
     \end{tabular}
     }
     & 
     {
    \setlength{\tabcolsep}{0pt}
\renewcommand{\arraystretch}{1}\setlength{\extrarowheight}{0pt}\begin{tabular}{r}
$r_{22}, r_{31}, r_{41}, r_{52}, r_{61}, r_{71},$\\{$r_{81}, r_{11} - r_{12} - r_{91}, a_1/r_{51}, $}\\{$a_3/r_{51}, a_2/r_{51}, r_{21}/r_{51}$}
\end{tabular}
     }\\
     \midrule
    \exname{JAK-STAT}~\cite{Raia2011} & 
    {\setlength{\tabcolsep}{0pt}
\renewcommand{\arraystretch}{1}\setlength{\extrarowheight}{0pt}\begin{tabular}{r}
        $1300$ functions;\\
        indeterminates $t_{15}$ and $t_{21}$ appear\\
        in functions of degree at least $4$
     \end{tabular}
     }
     & 
    {
    \setlength{\tabcolsep}{0pt}
\renewcommand{\arraystretch}{1}\setlength{\extrarowheight}{0pt}\begin{tabular}{r}$t_1, t_2, t_3, t_4, t_5, t_6, t_7, t_8, t_9, $\\$t_{10}, t_{12}, t_{13}, t_{14}, t_{16}, t_{18}, t_{19}, $\\$t_{20}, t_{17} t_{22}, t_{15} t_{21}, t_{11} t_{21}$
\end{tabular}
     }\\
    \end{tabular}
    \caption{Results of our simplification procedure: field generators before and after applying~\Cref{alg:main}}
  \label{table:simplify}
\end{table}
}

Our benchmark suite consists of 53 examples, which correspond to problems found in literature. Many of these examples have been collected by Barreiro and Villaverde in~\cite{on-the-origins,ReyBarreiro2023} and in the works referenced therein. For each example, we construct a list of rational functions, the generating set to be simplified, and pass it to our main algorithm. We run the algorithm with the same parameters for all examples: \grobner{} bases are computed in degrevlex ordering, \Cref{alg:polynomials} uses the degree bound $\delta \coloneqq 3$, and \ref{main:step:kind-of-enforce-minimality} 
orders the generators using the strategy from \Cref{remark:onsimplicity} and does not enforce minimality.

In \Cref{table:simplify}, we list results for a selection of
examples (the rest are given in \Cref{app:simplify_everyone}) and also include \Cref{ex:simple-OMS-simplify} to clarify the format. 
In this section, for brevity, the degree of a rational function means the sum of the degrees of its numerator and denominator.
Since we did not state the simplification problem in a fully formal way, we have no single qualitative measure to assess the quality of simplification.
Therefore, we use some natural related measures (such as the number of generators, their degrees, and further criteria from \Cref{remark:onsimplicity}) in combination with manual inspection.
Analyzing the overall results, we make the following observations, which we will illustrate using the examples from the table.
\begin{itemize}
    \item \Cref{alg:main} drastically reduces the size 
    of the generating sets.
    Note that this does not amount to simply removing excessive generators: we observe that even the \emph{minimal} degrees of generators for some of the variables may be reduced significantly.

    \item In particular, we see instances of nontrivial rewriting. In the \exname{genLV} example, the original generating set has a common denominator $\beta_{12}$, and the variables $r_1, r_2$ always appear together and in fractions with the numerator of degree at least two. 
    In the simplified generating set there is only one fraction, and $r_1$ and $r_2$ are separated. 
    In the \exname{Pharm} example, every variable appears in the original generating set in an element of degree at least two, but the result of simplification is of degree one. 
    \item Simplification reveals that some subfields can be generated by polynomials. In the \exname{JAK-STAT} example, the original generating set contains rational functions with denominators of degree up to six, but the subfield can be generated by simple polynomials only. The sorting criteria from \Cref{remark:onsimplicity} facilitate such simplifications by favoring elements with small denominators. Overall, in 26 out of 53 examples, the original generating set contained rational functions with non-constant denominators, but the simplification revealed that the subfield 
    can
    be generated by polynomials only.
    \item In every but one example (\exname{Bilirubin}), the resulting generating set is algebraically independent over~$\mathbb{C}$. Independence implies that the generating set contains the fewest elements possible~\cite[Ch.~VIII, Theorem 1.1.]{Lang}.
\end{itemize}

Each example was run with a 1-hour timeout and a 50 GB RAM limit. The median running time was $20$ seconds. The longest running example was \exname{Ovarian\_follicle} with the running time of $23$ minutes. One example, \exname{Covid-2}, did not finish getting stuck in \ref{main:step:gb} and \ref{main:step:check-after-gb} of \Cref{alg:main}. 
This example stands apart because the coefficients of the OMS ideal are considerably larger than the input, making interpolation infeasible.
Even in the case of succesfull computation of these coefficients, it is not clear whether they would be useful for simplification.
More precisely, the largest degree in the input is 29, while the coefficients of the \grobner{} basis of the OMS ideal that generate the subfield are of degree 229. This is not typical for examples coming from this domain, and we do not what makes this example special.

From this experiment, we conclude that our method effectively deals with most of the real-world examples in the benchmark suite. 
In \Cref{sec:app}, we will discuss the usefulness of such simplifications in applications, including but not limited to structural identifiability for ODEs.
Additionally, given that we handle most of the benchmark problems, a natural direction would be to apply our algorithm to harder examples from the same area. One current limitation is that producing generating sets for such examples using the approach from~\cite{stident} (see \Cref{sec:app_ode}) is itself computationally demanding.

\subsection{Comparison with other methods for simplification}\label{sec:comparison_simplification}

We will compare our method with the simplification algorithm from~\cite{allident} implemented in Maple\footnote{We used the implementation from \url{https://github.com/pogudingleb/AllIdentifiableFunctions}}. This method also uses the coefficients of OMS ideals, but computes them directly instead of using interpolation. To pick the simplest elements from a generating set, it orders the functions by the length of their string representation. It does not implement \Cref{alg:polynomials}. This method has been developed with the applications in mind that are similar to ours.

{\small\begin{table}[H]
\setlength{\tabcolsep}{4pt}
\renewcommand{\arraystretch}{1.0}
\setlength\extrarowheight{0pt}
\centering
\begin{adjustbox}{center}
    \begin{tabular}{p{3.2cm}rr}
    Example & Result of~\cite{allident} & Result of~\Cref{alg:main}
    \\
    \midrule
    \midrule
    {\exname{Biohydrogenation}}~\cite{MOATE2008731}
    & 
    {\setlength{\tabcolsep}{0pt}
\renewcommand{\arraystretch}{1}\setlength{\extrarowheight}{0pt}\begin{tabular}{r}
$k_5, k_6, k_7, k_9^2, k_{10}/k_9$,\\$(k_{10} k_9 + 2 k_8 k_9) / k_{10}$
\end{tabular}}
& 
{\setlength{\tabcolsep}{0pt}
\renewcommand{\arraystretch}{1}\setlength{\extrarowheight}{0pt}\begin{tabular}{r}$k_5, k_6, k_7, k_9^2, k_9 k_{10},$\\$2 k_8 + k_{10}$
\end{tabular}}
\\
\midrule
\exname{SEIR1}~\cite{Zha2020} & 
{\setlength{\tabcolsep}{0pt}
\renewcommand{\arraystretch}{1}\setlength{\extrarowheight}{0pt}\begin{tabular}{r}
$\gamma, \beta/\psi, \gamma \psi - \psi - v,$\\
$\beta^2(\gamma- v)/(\gamma \psi - \psi)$
\end{tabular}}
&
{\setlength{\tabcolsep}{0pt}
\renewcommand{\arraystretch}{1}\setlength{\extrarowheight}{0pt}\begin{tabular}{r}
$\gamma, \beta/\psi, \psi v - \psi - v,$\\
$\gamma \psi - \psi v$
\end{tabular}}
\\
\midrule
\exname{SLIQR}~\cite{Dankwa2021} &
{\setlength{\tabcolsep}{0pt}
\renewcommand{\arraystretch}{1}\setlength{\extrarowheight}{0pt}\begin{tabular}{r}
$\sigma, \beta, N, \alpha + \gamma, \frac{(\eta \sigma + \alpha - \sigma)}{\alpha \eta},$\\  {\footnotesize $\frac{\eta^2 \gamma \sigma^2+\alpha^2 \eta \sigma+3 \alpha \eta \gamma \sigma - \alpha \eta \sigma^2 - 2 \eta \gamma \sigma^2 + \alpha^2 \gamma - 2 \alpha \gamma \sigma + \gamma \sigma^2}{\eta \sigma + \alpha - \sigma}$}
\end{tabular}}
&
{\setlength{\tabcolsep}{0pt}
\renewcommand{\arraystretch}{1}\setlength{\extrarowheight}{0pt}\begin{tabular}{r}
$\sigma, \beta, N, \alpha + \gamma, \alpha \eta \gamma,$\\$\alpha \gamma + \eta \gamma \sigma - \gamma \sigma$
\end{tabular}}
\\
\midrule
\exname{Bilirubin}~\cite{combos}
& 
{\footnotesize\setlength{\tabcolsep}{0pt}
\renewcommand{\arraystretch}{1}\setlength{\extrarowheight}{8pt}\begin{tabular}{r}
$k_{01}, k_{12}+k_{13}+k_{14}, k_{21}+k_{31}+k_{41}$,\\
$\frac{k_{12}^2 k_{13}-k_{12}^2 k_{14}-k_{12} k_{13}^2+k_{12} k_{14}^2+k_{13}^2 k_{14}-k_{13} k_{14}^2}{k_{12} k_{31}-k_{12} k_{41}-k_{13} k_{21}+k_{13} k_{41}+k_{14} k_{21}-k_{14} k_{31}}$,\\$\frac{k_{21}^2 k_{31}-k_{21}^2 k_{41}-k_{21} k_{31}^2+k_{21} k_{41}^2+k_{31}^2 k_{41}-k_{31} k_{41}^2}{k_{12} k_{31}-k_{12} k_{41}-k_{13} k_{21}+k_{13} k_{41}+k_{14} k_{21}-k_{14} k_{31}},$\\$\frac{-k_{12} k_{31}^2+k_{12} k_{41}^2+k_{13} k_{21}^2-k_{13} k_{41}^2-k_{14} k_{21}^2+k_{14} k_{31}^2}{k_{12} k_{31}-k_{12} k_{41}-k_{13} k_{21}+k_{13} k_{41}+k_{14} k_{21}-k_{14} k_{31}}, $\\$\frac{-k_{12}^2 k_{31}+k_{12}^2 k_{41}+k_{13}^2 k_{21}-k_{13}^2 k_{41}-k_{14}^2 k_{21}+k_{14}^2 k_{31}}{k_{12} k_{31}-k_{12} k_{41}-k_{13} k_{21}+k_{13} k_{41}+k_{14} k_{21}-k_{14} k_{31}}$
\end{tabular}}
&
{\footnotesize\setlength{\tabcolsep}{0pt}
\renewcommand{\arraystretch}{1}\setlength{\extrarowheight}{3pt}\begin{tabular}{r}
$k_{01}, k_{12} k_{13} k_{14}, k_{21} k_{31} k_{41}$,\\$k_{12} + k_{13} + k_{14},k_{21} + k_{31} + k_{41},$\\$k_{21} k_{31} + k_{21} k_{41} + k_{31} k_{41}, $\\$k_{12} k_{13} + k_{12} k_{14} + k_{13} k_{14}, $\\$k_{12} k_{31} + k_{12} k_{41} + k_{13} k_{21} + k_{13} k_{41} $\\$+ k_{14} k_{21} + k_{14} k_{31}, k_{21} + k_{31} + k_{41}$
\end{tabular}}
\\
    \end{tabular}
    \end{adjustbox}
    \caption{The differences of results obtained by the method from \cite{allident} and our implementation of \Cref{alg:main} on some cherry-picked examples}
  \label{table:simplify-compare-maple}
\end{table}}

We run the method from~\cite{allident} on our benchmark suite, with the same time and memory constraints as in \Cref{sec:experiment-simplify}. Overall, this method does not finish for 8 examples\footnote{The examples where the method does not finish also do not finish with the time limit increased ten-fold}, obtains a result similar to ours for 36 examples, and in the remaining 9 examples our method's results seem better. To corroborate the latter claim, we present the differences for some of the results in \Cref{table:simplify-compare-maple}. For \exname{Biohydrogenation}, \exname{SEIR1}, and \exname{SLIQR}, the methods produce generating sets with the same number of elements; however, our method produces elements of smaller degrees. In the \exname{Bilirubin} example, our output is visually simpler but contains more elements (8 vs. 7); this also means that the result of~\cite{allident} is algebraically independent, while ours is not. Of course, in some use cases, the result of \cite{allident} may be preferable.

We conclude that our method appears to be an improvement over \cite{allident} both in terms of better performance and for producing simpler generating sets.

\begin{remark}[Computing the algebraic closure]
  In some applications, such as structural parameter identifiability, it may be of interest to consider not only a subfield 
  $\mathcal{E}$ of $\QQ(\bx)$
  but also its relative algebraic closure $\mathcal{E}^{\operatorname{rel}}$ in $\QQ(\bx)$.
  An approach for finding a simple transcendence basis for $\mathcal{E}^{\operatorname{rel}}$ based on a construction reminiscent to OMS ideals was proposed in a series of papers~\cite{combinations_combos_algo, Meshkat2011, combos} and implemented in the COMBOS web-application~\cite{combos}.
  A~completely different approach 
  for finding simple elements in $\mathcal{E}^{\operatorname{rel}}$,
  based on linear algebra and avoiding Gr\"obner basis computations, was proposed in~\cite[Section~5]{MRS2018}.
  A direct comparison of these methods to ours is complicated: for example, all the fields in~\Cref{table:simplify-compare-maple} have the full $\mathbb{Q}(\bx)$ as the algebraic closure, so $\bx$ itself could be taken as a simple transcendence basis of $\mathcal{E}^{\operatorname{rel}}$. 
\end{remark}

\begin{example}
    Some simplified generators appearing in the literature can be sometimes simplified even further using our algorithm. Consider, for example, the following field generators produced by COMBOS in~\cite[Test-Model 2, $n=3$]{combos}:
    \[
    \begin{aligned}
    &p_{11}, p_{21}, p_{22} + p_{33}, p_{23}/p_{13}, p_{22} - p_{12} p_{23} / p_{13},\\
    &p_{12} p_{23} p_{33} / p_{13} - p_{23} p_{32} + p_{12} p_{22} p_{23} / p_{13} - p_{22}^2,\\
    &p_{11} p_{23} p_{33} / p_{13} - p_{23} p_{31} + p_{11} p_{22} p_{23} / p_{13} - p_{21} p_{22}.
    \end{aligned}
    \]
    We can either feed these generators to our simplification algorithm, or use the approach from \Cref{sec:app_ode}. Both produce the same result, which is visually simpler and contains elements of lower degrees:
    \[
     p_{11}, p_{21}, p_{22} + p_{33}, p_{13}/p_{23}, p_{22} p_{33} - p_{23} p_{32}, p_{21} p_{33} - p_{23} p_{31}, p_{12} p_{21} + p_{13} p_{31}.
    \]
\end{example}

\section{Applications}
\label{sec:app}

In this section, we present three sets of examples that showcase the simplifications achieved by our algorithm. 
We start with examples arising in the study of structural identifiability of ODE models in~\Cref{sec:app_ode}.
This application domain has been our original motivation for this work.
Then we show that our algorithm 
produces
interesting results applied in other contexts such as identifiability and observability of discrete-time models (\Cref{sec:app_dde}) and rational invariant theory (\Cref{sec:app_invariants}). 

In addition to the situations where one is interested in simpler expressions for generators (such as the case studies below), we expect that having more concise generators can improve the performance of other algorithms on rational function fields (such as membership testing, see also \Cref{remark:reuse}).


\subsection{Structural identifiability of ODE models}\label{sec:app_ode}

The main motivation for the algorithmic developments presented in this paper comes from the problem of computing \emph{identifiable functions} of dynamical models.
We start the section with a high-level explanation of the way generator simplification question appears in modeling and its importance.
For further technical details, we refer to~\cite{combinations_combos_algo,OPT,allident}.
Then we will showcase simplifications achieved by our algorithm using several systems from modeling literature and demonstrate that the simplified generating set is better suited to derive insights about the structure of the model.
However, since mathematical modeling is not a central topic in this paper, we will not get deep into the domain-specific details.

We consider a parametric system of differential equations
\begin{equation}\label{eq:ode}
\bx' = \bg(\bmu, \bx),
\end{equation}
where $\bx$ is a vector of unknown functions of time, $\bmu$ is a vector of scalar parameters, and $\bg$ is a vector of rational functions from $\CC(\bx, \bmu)$ defining the dynamics.
The system may also involve external inputs, we will omit them to simplify the presentation.
In a typical experimental setup, it is not feasible to measure the values of all the coordinates of $\bx$.
Thus, data is available only for some functions of $\bx$ and $\bmu$ typically referred to as \emph{outputs} of the model.
For the sake of simplicity, we will consider the case of a single output $y = f(\bmu, \bx)$.

A parameter $\mu_i$ or, more generally, a function $h(\bmu)$ of parameters is \emph{(structurally) identifiable} if its value can be uniquely reconstructed from the output data assuming continuous and noise-free measurements (see~\cite[Section~2]{HOPY} for more details).
Since identifiability is a prerequisite for meaningful parameter estimation, it is an important and desirable property.
The set of identifiable functions can be shown to be closed under addition, multiplication, and division, so they form \emph{the field of identifiable functions}.
The elements of this field (also called identifiable combinations) are exactly the quantities whose values can be in principle inferred from data, so computing a convenient representation of this field is a natural question.

\begin{example}\label{ex:toy_ident}
    Consider a model
    \[
    x_1' = \mu_1 x_2,\quad x_2' = -\mu_2x_1,
    \]
    and assume that the observation is $y = x_1$, the first coordinate.
    First, we observe that the scaling $x_2 \to \lambda x_2, \mu_2 \to \lambda \mu_2, \mu_1 \to \frac{1}{\lambda}\mu_1$ preserves the output function for any nonzero constant $\lambda$.
    Thus, neither of $\mu_1$ and $\mu_2$ can be identified even from perfect knowledge of $y$.
    On the other hand, $y'' = -\mu_1\mu_2 y$, so the product $\mu_1\mu_2$ can be found as long as $y$ is not identically zero.
\end{example}

The classical approach for finding a generating set of the field of identifiable functions going back to~\cite{OllivierPhD} is the following.
First, we compute the minimal polynomial differential equation (first with respect to the order, then degree) with coefficients in $\CC(\bmu)$ satisfied by $y$ (called \emph{the input-output} equation).
By normalizing one of the coefficients to one, we can write it as
\begin{equation}\label{eq:ioeq}
m_0(y, \ldots, y^{(h)}) + \sum\limits_{i = 1}^M c_i(\bmu) m_i(y, \ldots, y^{(h)}) = 0,
\end{equation}
where $m_0, \ldots, m_M$ are distinct monomials in $y, y', \ldots, y^{(h)}$ and $c_1, \ldots, c_M \in \CC(\bmu)$.
Then, under certain mild assumptions, the functions $c_1, \ldots, c_M$ generate the whole field of identifiable functions~\cite[Corollary~1]{OPT}.

\begin{example}[\Cref{ex:toy_ident} continued]
    The minimal differential equation satisfied by $y$ in this example is $y'' + \mu_1\mu_2y = 0$.
    Therefore, \cite[Corollary~1]{OPT} implies that identifiable functions for this model are precisely rational functions in $\mu_1\mu_2$.
\end{example}

Unlike the toy example above, in more realistic models, the functions $c_1, \ldots, c_M$ are often large and complicated (see~\Cref{app:simplify_everyone}); this makes such generating set practically useless.
Thus, the challenge is, given a generating set $c_1, \ldots, c_M$, to produce a \emph{simpler} one which can later be understood and interpreted by a modeler.
Below we will demonstrate our simplification algorithm on several fields coming from the structural identifiability  problem.

\begin{example}[Compartmental model]
\label{ex:compmodel}
Consider the following linear model (which was employed, for example, for modeling bilirubin kinetics, see~\cite[Figure~1]{Berk1969} and~\cite[Example~4]{combos}):
\begin{equation}
    \vcenter{\hbox{\begin{tikzpicture}[
    node distance=3cm,
    state/.style={rectangle, draw, minimum width=1cm, minimum height=1cm, thick},
    arrow/.style={->,>=Stealth,thick}
]
\node[state] (X1) at (0,0) {$x_1$};
\node[state] (X2) at (3,0) {$x_2$};
\node[state] (X3) at (0,-2.5) {$x_3$};

\draw[arrow, transform canvas={yshift=0.5em}] (X1.east) -- node[above] {$\mu_{21}$} (X2.west);
\draw[arrow, transform canvas={yshift=-0.5em}] (X2.west) -- node[below] {$\mu_{12}$} (X1.east);
\draw[arrow, transform canvas={xshift=-0.5em}] (X1.south) -- node[left] {$\mu_{31}$} (X3.north);
\draw[arrow, transform canvas={xshift=0.5em}] (X3.north) -- node[right] {$\mu_{13}$} (X1.south);
\draw[arrow] (X1.west) -- node[above] {$\mu_{01}$} ++(-1.2,0);
\draw[arrow] (0,1.5) -- node[right] {$u$} (X1.north);

\end{tikzpicture}}}\qquad
    \begin{cases}
        x_1'(t) = -(\mu_{21} + \mu_{31} + \mu_{01}) x_1(t) + \\
        \hspace{15mm}\mu_{12} x_2(t) + \mu_{13} x_3(t)  + u(t),\\
       x_2'(t) = \mu_{21} x_1(t) - \mu_{12} x_2(t),\\
       x_3'(t) = \mu_{31} x_1(t) - \mu_{13} x_3(t).
    \end{cases}\quad 
\end{equation}
In this model, we have
\[
 \bx  = [x_1,\; x_2,\; x_3] \quad \text{ and } \quad \bmu  = [\mu_{01}, \mu_{12}, \mu_{13}, \mu_{21}, \mu_{31}],
\]
observations are taken to be $y(t) = x_1(t)$, and $u(t)$ is an external input.
Due to the linearity of the model, there will be not too many coefficients in the input-output equation~\eqref{eq:ioeq}, and taking them directly would yield the following
generators of the field of identifiable functions:
\[
\mu_{12} + \mu_{13},\; \mu_{12} \mu_{13},\; \mu_{01} \mu_{12} \mu_{13},\; \mu_{01} + \mu_{12} + \mu_{13} + \mu_{21} + \mu_{31}, \;   \mu_{01} \mu_{12} + \mu_{01} \mu_{13} + \mu_{12} \mu_{13} + \mu_{12} \mu_{31} + \mu_{13} \mu_{21}.
\]
The simplified generating set produced by our~\Cref{alg:main} is
\begin{equation}\label{eq:lincomp_simplified}
 \mu_{01},\; \mu_{21} + \mu_{31},\; \mu_{21} \mu_{31}, \; \mu_{12} + \mu_{13}, \; \mu_{12} \mu_{13}, \; \mu_{12}\mu_{31} + \mu_{13} \mu_{21}.
\end{equation}
This generating set reveals several important properties of the model:
\begin{itemize}
    \item identifiability of $\mu_{01}$ is made explicit (in the original generating set, $\mu_{01}$ could be obtained as the quotient of the second and third generators);
    \item the symmetry with respect to the exchange of the second and third compartments ($x_2$ and $x_3$) also becomes explicit; it is, furthermore, clear that the field is precisely the field of invariants with respect to the simultaneous exchange $(\mu_{12},\mu_{21}) \leftrightarrow (\mu_{13}, \mu_{31})$.
\end{itemize}
We note, however, that the simplified
generating set contains one more element than the original one (6 vs. 5), so the generators are not independent.
The minimality and independence can be restored by applying~\Cref{alg:minimization}, the resulting set will be~\eqref{eq:lincomp_simplified} without $\mu_{12} \mu_{13}$.
While this set is smaller, it does not reveal the symmetry $x_2 \leftrightarrow x_3$ as explicitly, so, depending on the goals, one of these two outputs may be preferable.
\end{example}

\begin{example}[\exname{SIRT} epidemiological model]
We consider an epidemiological model describing a spread of infection which takes into account that the individuals receiving treatment are less infectious~\cite[Eq. (2.3)]{Tuncer2018}:
\begin{equation}\label{eq:SIRT_ode}
\begin{cases}
    S'(t) = -\frac{\beta S(t) I(t)}{N} - \frac{\delta \beta S(t) T(t)}{N},\\
    I'(t) = \frac{\beta S(t) I(t)}{N} + \frac{\delta \beta S(t) T(t)}{N} - (\alpha + \gamma) I(t),\\
    T'(t) = \gamma I(t) - \nu T(t).
\end{cases}\quad \text{where} \quad \begin{matrix}
            \bx & = & [S,\; I,\; T]\\
            \bmu & = & [\alpha, \beta, \gamma, \delta, \nu, N]
        \end{matrix}
\end{equation}
In this model, $S$ and $I$ denote the number of susceptible and infectious individuals, respectively, while $T$ stands for the number of individuals receiving treatment.
The observation is taken to be $y(t) = T(t)$.

The generating set of the field of identifiable functions extracted directly from the input-output equation~\eqref{eq:ioeq} consists of twelve rational elements: polynomials up to degree six and rational functions with linear denominator and numerator up to degree five.
\Cref{alg:main} produces the following generating set for the same field:
\begin{equation}\label{eq:SIRT}
 \alpha + \gamma + \nu, \quad \nu (\alpha + \gamma),\quad \nu + \delta \gamma,\quad \frac{N \gamma}{\beta}.
\end{equation}
The first two generators are the elementary symmetric functions in $\alpha + \gamma$ and $\nu$ hinting at a symmetry between the outflow rates of the $I$ and $T$ compartments of the model.
Indeed, such a symmetry is a recurring motif (see, e.g., \cite[A.3, Eq. (9)]{Hong2019} and~\cite[Example 4.1]{glebhdr}).

To facilitate the interpretation of the third generator in~\eqref{eq:SIRT}, we observe that the nonlinear terms in the first two equations~\eqref{eq:SIRT_ode} can be factored as $\frac{\beta S(t)}{N} (I(t) + \delta T(t))$ suggesting $\widetilde{I}(t) := I(t) + \delta T(t)$ (``total infectiousness'') as a natural quantity to consider.
Rewriting the equation for $T$ in terms of $T$ and $\widetilde{I}(t)$, we get
\[
T'(t) = \gamma \widetilde{I}(t) - (\dashuline{\nu + \gamma\delta}) T(t), 
\]
where the third generator from~\eqref{eq:SIRT} appears as an outflow rate.
\end{example}

\begin{example}[\exname{SLIQR} epidemiological model]\label{ex:SLIQR}
    Consider an epidemiological model describing a spread of infection which takes into account the latent period and possible remissions (\cite[Equation~(4)]{Dankwa2022}):
    \begin{equation}
        \begin{cases}
            S'(t) = -\frac{\beta S(t)I(t)}{N} - v(t) \frac{S(t)}{N},\\
            L'(t) = \frac{\beta S(t)I(t)}{N} - \alpha L(t),\\
            I'(t) = \alpha L(t) + \sigma Q(t) -\gamma I(t),\\
            Q'(t) = \gamma (1 - \eta) I(t) - \sigma Q(t)
        \end{cases}\quad \text{where} \quad \begin{matrix}
            \bx & = & [S,\; L,\; I,\; Q]\\
            \bmu & = & [\alpha,\beta, \gamma, \eta,\sigma, N]
        \end{matrix}
    \end{equation}
    In this model, $S$ and $I$ denote the number of susceptible and infectious individuals, respectively.
    $L$ and $Q$ stand for the number of individuals in the latent and remission stages, respectively.
    Furthermore, $v(t)$ is an external 
    input
    (the proportion of vaccinated individuals) and the observation is taken to be the prevalence (the proportion of infectious individuals), that is, $y(t) = \frac{I(t)}{N}$.

    The generating set of the field of identifiable functions extracted directly from the input-output equation~\eqref{eq:ioeq} consists of 40 polynomials
    with the degree reaching eight.
    \Cref{alg:main} produces the following set of generators for the same field: 
    \[
      \beta,\quad \sigma,\quad N, \quad \alpha + \gamma, \quad \alpha\eta\gamma,\quad \alpha \gamma + (\eta - 1) \gamma \sigma.
    \]
    First, this generating set explicitly features identifiable parameters $\beta, \sigma$, and $N$.
    Two of the remaining generators appear naturally if we introduce a new quantity, $C(t) := I'(t)$, the number of new infected cases, a quantity frequently used in epidemiology (also called incidence).
    A direct computation shows that the differential
    equation for
    $C(t)$ will be
    \[
    C'(t) = \frac{\alpha \beta S(t) I(t)}{N} - \dashuline{(\alpha + \gamma)} C(t) - \dashuline{(\alpha \gamma + (\eta - 1)\gamma \sigma )} I(t) + \sigma (\alpha - \sigma)Q(t).
    \]
    This equation explicitly involves the forth and sixth generators as the flow rates.
\end{example}


\subsection{Identifiability and observability for discrete-time systems}\label{sec:app_dde}

In this section we will apply our algorithm for the problems of similar flavor as in the previous subsection but for models in discrete time.
The general setup is very similar although the computation will be quite different.
Specifically, we consider a parametric system of difference equations (cf.~\eqref{eq:ode}):
\begin{equation}\label{eq:discr_ode}
    \bx(t + 1) = \bg(\bmu, \bx(t)),
\end{equation}
where $\bx(t)$ is a vector of time-series (for $t = 0, 1, 2, \ldots$), $\bmu$ is a vector of unknown parameters, and $\bg$ is a vector of rational functions from $\CC(\bx, \bmu)$.
We assume that we observe an output time series $y(t) = f(\bmu, \bx(t))$, and are interested in identifiable functions in this context.
In other words, we are given a sequence $y(0), y(1), y(2), \ldots$ or a finite part of it, and we want to find out which functions of parameters $\bmu$ and initial conditions $\bx(0)$ we can reconstruct.
For generic values of $\bmu$ and $\bx(0)$ (so that none of the relevant denominators vanish), these functions form a field which we will call the field of identifiable functions (in the literature, for initial conditions $\bx(0)$, term \emph{observable functions} is also often employed).

Despite of the similarity of the general question to the one in Section~\ref{sec:app_ode}, the approach via input-output equations does not work in the difference case as well as it does in the differential (see~\cite{LGEL2011} for a discussion).
Therefore, we will use the following more direct method (cf.~\cite{Kawano2011}).
We define rational functions $\bg^{\circ j}(\bmu, \bx(0))$ by
\[
  \bg^{\circ 0}(\bmu, \bx(0)) := \bx(0)\quad \text{and} \quad \bg^{\circ (j + 1)}(\bmu, \bx(0)) := \bg(\bmu, \bg^{\circ j}(\bmu, \bx(0))) \text{ for }j \geqslant 0.
\]
Comparing with~\eqref{eq:discr_ode}, we observe that $\bx(j) = \bg^{(j)}(\bmu, \bx(0))$.
Thus, $y(j)$ can be written as a rational function $f(\bmu, \bg^{(j)}(\bmu, \bx(0)))$, and the field of functions identifiable from $y(0), \ldots, y(\ell)$ is generated by these rational functions for $j = 0, \ldots, \ell$.
However, as we will see below, these rational functions become complicated very quickly, so simplification is a key step in finding an interpretable set of generators.

\begin{example}[General fractional-linear model]
    We will consider a one-dimensional model with dynamics defined by a general fractional-linear function:
    \[
    x(t + 1) = \frac{\mu_1 x(t) + \mu_2}{\mu_3 x(t) + \mu_4}, \quad y(t) = x(t).
    \]
    We start with computing simple generators for the field generated by the values of the output at $t = 0, 1, 2, 3$.
    These values can be written as rational functions in $\bmu, x(0)$ as follows:
    \begin{align*}
    y(0) &= x(0), \;\; y(1) = \frac{\mu_1 x(0) + \mu_2}{\mu_3 x(0) + \mu_4},\;\; y(2) = \frac{(\mu_1^2 + \mu_2 \mu_3) x(0) + \mu_1 \mu_2 + \mu_2 \mu_4}{(\mu_1 \mu_3 + \mu_3 \mu_4)x(0) + \mu_2 \mu_3 + \mu_4^2} \\
    y(3) &= \frac{(\mu_1^3 + 2  \mu_1 \mu_2 \mu_3 + \mu_2 \mu_3 \mu_4) x(0) + \mu_1^2 \mu_2 + \mu_1 \mu_2 \mu_4 + \mu_2^2 \mu_3 + \mu_2 \mu_4^2}{ (\mu_1^2 \mu_3 + \mu_1 \mu_3 \mu_4 + \mu_2 \mu_3^2 +  \mu_3 \mu_4^2) x(0) + \mu_1 \mu_2 \mu_3 + 2 \mu_2 \mu_3 \mu_4 + \mu_4^3}
    \end{align*}
    Our simplification algorithm finds the following much simpler generators for the same field: 
    \[
    x(0),\;\; \frac{\mu_1}{\mu_4}, \;\;\frac{\mu_2}{\mu_4}, \;\;\frac{\mu_3}{\mu_4}.
    \]
    One can show that taking further time points will not add any new information since the fraction in the right-hand side is invariant under multiplying all the $\mu$'s by the same number.
    On the other hand, 
    the output of our algorithm suggests
    that the fields obtained by taking fewer time points are strictly smaller.
    Therefore, the field of identifiable functions is generated by $x(0), \frac{\mu_1}{\mu_4}, \frac{\mu_2}{\mu_4}, \frac{\mu_3}{\mu_4}$ and four time points are enough to find these values.
\end{example}

\begin{example}[Logistic growth]
    The following logistic growth model is used, for example, in population modeling~\cite[p. 12]{Allman2003}:
    \[
    x(t + 1) = x(t) \left(1 + \mu_1\left(1 - \frac{x(t)}{\mu_2} \right) \right).
    \]
    We will further assume that what is observed is an unknown but fixed fraction of the population, that is, $y(t) = \mu_3 x(t)$.
    In order to find identifiable functions, we will consider the field generated by $y(0), y(1)$, and $y(2)$, where the latter is of the form $\frac{P(x(0), \bmu)}{\mu_2^3}$ with $\deg P = 8$.
    In contrast to this, our algorithm finds the following generators for the same field
    \[
    \mu_3 x(0), \;\; \mu_1, \;\; \mu_2 \mu_3.
    \]
    By repeating this computation with $y(4)$, one can show that adding new time points does not change the generators. 
    Since the field is smaller if only one or two time points are taken, we conclude that the three functions above generate the field of identifiable functions and knowledge of only three data points is generically enough to reconstruct them.
\end{example}

\begin{example}[SIS/SIR epidemiological model]
    The following model combines the dynamics of discrete-time SIS and SIR models~\cite[Sections~3 and 4]{Allen1994}:
        \begin{equation}
        \begin{cases}
            S(t + 1) = S(t) (1 - \beta I(t)) + \alpha I(t),\\
            I(t + 1) = I(t) (1 - \alpha - \gamma + \beta S(t)),\\
            R(t + 1) = R(t) + \gamma I(t),\\
            y(t) = \eta I(t),
        \end{cases}\quad \text{where} \quad \begin{matrix}
            \bx & = & [S,\; I,\; R]\\
            \bmu & = & [\alpha, \beta, \gamma,   \eta]
        \end{matrix}
    \end{equation}
    As in the previous examples, we compute the field generated by the values $y(0), y(1), \ldots$ and observe that after $y(0), \ldots, y(3)$ the field stabilizes.
    Simplifying $y(0), \ldots, y(3)$, where the latter is a polynomial of degree $16$ with $102$ terms, we obtain a much simpler generating set of identifiable functions:
    \[
    \gamma,\;\; \eta I(0),\;\; \beta I(0),\;\; \beta S(0) - \alpha.
    \]
    These simpler expressions also admit biological interpretations: $\eta I(0)$ is the observed number of infections at $t = 0$, $\beta I(0)$ is essentially the total infection rate at $t = 0$, and $\alpha - \beta S(0)$ is the total rate of the flow 
     from the $I$-compartment to the $S$-compartment.
\end{example}


\subsection{Fields of invariants}\label{sec:app_invariants}

Let $G$ be a group and consider its action on an $n$-dimensional $k$-vector space $V$ (in other words, a homomorphism $G \to \operatorname{GL}(V)$).
This action naturally defines the action of $G$ on the ring of polynomials $k[V]$ and the field of rational fractions $k(V)$.
\emph{Invariant theory} studies the invariants of this action: the ring of polynomial invariants $k[V]^G$ and the field of rational invariants $k(V)^G$.
Historically, the invariant theory was focused on the polynomial invariants (see, e.g.~\cite{Weyl}), but in the recent years rational invariants have attracted attention (see, e.g, ~\cite{Bandeira2023,Hubert2025,Kogan2023,blumsmith2025genericorbitsnormalbases,Hubert26Jalard,Hubert25JalardPz,breloer2025rationalinvariantsdegreepolynomials}).
From the applications point of view, invariants are useful to characterize data (signals or images) which is defined up to a group action~\cite{Bandeira2023}.
The degrees of the generators are related to the sampling complexity of the reconstruction problem for such data~\cite{Bandeira2023}, and the generators themselves can be used to distinguish different orbits with respect to the action~\cite{Flusser2000}. 

The first algorithm for computing a generating set of $k(V)^G$ was proposed by M\"uller-Quade and Beth in~\cite{Muller-Quade99b}.
It was further extended by Hubert and Kogan in~\cite{Hubert2007}, in particular, by allowing to provide
a cross-section~\cite[Section~3.1]{Hubert2007} of the group action as an extra input (for some other generalizations, see~\cite{Kemper2007}).
Providing a cross-section typically reduces the computation time and results in a simpler generating set.
In the following examples we show that the methods and software described in this paper may simplify further the generating sets produced by the state of the art algorithms.

For the computations below, we used an implementation of the Hubert-Kogan algorithm from the AIDA~\cite{Aida} Maple package.

\begin{example}[$S_3$-action]
Consider the action of the symmetric group $S_3$ on $6$-dimensional space (inspired by~\cite[Section~4.1]{Gorlach19}) such that a permutation $\sigma \in S_3$ acts on a vector $(a_1, a_2, a_3, b_1, b_2, b_3) \in \mathbb{R}^6$ as follows:
\[
\sigma(a_1, a_2, a_3, b_1, b_2, b_3) = (a_{\sigma(1)},\; a_{\sigma(2)},\; a_{\sigma(3)},\; (-1)^\sigma b_{\sigma(1)},\;(-1)^\sigma b_{\sigma(2)},\; (-1)^\sigma b_{\sigma(3)} ),
\]
where $(-1)^\sigma$ denotes the sign of the permutation.
Finite groups do not admit cross-sections, so we run the Hubert-Kogan algorithm without providing any additional information.
The resulting generating set contains a natural symmetric polynomial $a_1 + a_2 + a_3$ and $33$ rational functions with total degrees ranging from five to seven.
We apply our simplification algorithm to this set and obtain a
simpler generating set consisting of seven polynomials only:
\begin{align*}
&a_1 + a_2 + a_3,\quad a_1 a_2 + a_1 a_3 + a_2 a_3,\quad a_1 a_2 a_3,\quad b_1 b_2 + b_1 b_3 + b_2 b_3,\\
&b_1^2 + b_2^2 + b_3^2,\quad a_1 b_2 - a_1 b_3 - a_2 b_1 + a_2 b_3 + a_3 b_1 - a_3 b_2,\quad a_1 b_2 b_3 + a_2 b_1 b_3 + a_3 b_1 b_2.
\end{align*}
\end{example}

\begin{example}[Triangular action with swap]\label{ex:triang}
    Consider a space $V$ of at most quadratic bivariate polynomials without constant term $a_{01}x + a_{10}y + a_{02} x^2 + 2 a_{11} x y + a_{20} y^2$.
    Group $T_2$ of nonsingular upper-triangular $2 \times 2$ matrices acts on $V$ via its action on the $(x, y)$-space.
    We define the action of $T_2 \times (\mathbb{Z} / 2\mathbb{Z})$ on $V \oplus V \cong \mathbb{R}^{10}$ by $T_2$ acting on the summands and $(\mathbb{Z} / 2\mathbb{Z})$ swapping the two copies of $V$.
    The coordinates on the second copy of $V$ will be denoted by $b_{01}, b_{10}, b_{02}, b_{11}, b_{20}$.

    We apply the Hubert-Kogan algorithm with the cross-section $a_{11} = 0, a_{02} = 1$.
    This cross-section is chosen as the one with the simplest output among all cross-sections consisting of the equations of the form $x = 0$ or $x = 1$, where $x$ is a variable\footnote{We decided to use this class of cross-sections based on the existing literature and private communication with Evelyne Hubert.}, for which the computation finishes within 10 minutes.
    The resulting generating set contains $\frac{a_{11} b_{02} - a_{02} b_{11}}{a_{02} - b_{02}}$ and $\frac{a_{02}^2 + b_{02}^2}{a_{02} b_{02}}$ as well as 18 rational function of total degrees ranging from 5 to 11.
    Our simplification algorithm produces the following generating set consisting of five polynomials and four rational functions (including the two simple elements of the original generating set) with degrees ranging from two to four:
    \begin{align*}
        &a_{20} b_{02} - 2 a_{11} b_{11} + a_{02} b_{20}, \quad a_{20} a_{02} - a_{11}^2 + b_{20} b_{02} - b_{11}^2, \\
        &a_{10}^2 a_{02} - 2 a_{10} a_{01} a_{11} + a_{01}^2 a_{20} + b_{10}^2 b_{02} - 2 b_{10} b_{01} b_{11} + b_{01}^2 b_{20},\\
        &a_{10}^2 b_{02} - 2 a_{10} a_{01} b_{11} + a_{01}^2 b_{20} + a_{20} b_{01}^2 - 2 a_{11} b_{10} b_{01} + a_{02} b_{10}^2,\\
        &\frac{a_{01} b_{01}}{a_{02} + b_{02}}, \quad \frac{a_{11} b_{02} - a_{02} b_{11}}{a_{02} - b_{02}},\quad \frac{a_{10} b_{01} - a_{01} b_{10}}{a_{20} a_{02} - a_{11}^2 - b_{20} b_{02} + b_{11}^2}, \quad \frac{a_{02}^2 + b_{02}^2}{a_{02} b_{02}},\\
        &a_{10} a_{11} b_{01} - a_{10} a_{02} b_{10} - a_{10} b_{10} b_{02} + a_{10} b_{01} b_{11} - a_{01} a_{20} b_{01} + a_{01} a_{11} b_{10} + a_{01} b_{10} b_{11} - a_{01} b_{01} b_{20}.
    \end{align*}
\end{example}

\begin{example}[Moment invariants for 2D images]
\label{example:moments}
    In this example, we will focus on \emph{moment invariants} which is a popular tool in computer vision and patter recognition~\cite{MingKueiHu1962, Flusser2000, Flusser2009, Hickman2011}.
    For our purposes, for any fixed integer $d$, the space of moments $M_d$ will simply be the space of nonconstant monomials of degree at most $d$ in variables $x$ and $y$.
    For example, $M_2$ is spanned by $x, y, x^2, xy, y^2$.
    The monomial $x^i y^j$ will be denoted by $m_{i, j}$.
    We will consider the action of the planar rotation group $\operatorname{SO}(2)$ induced by its action on the $(x, y)$-plane. 
    More precisely:
    \[
    \text{for } R = \begin{pmatrix}
        a & b \\
        c & d
    \end{pmatrix} \in \operatorname{SO}(2) \text{ we define } R \cdot m(x, y) := m(ax + by, cx + dy).
    \]
    For example, $m_{2,0} + m_{0, 2} = x^2 + y^2$ is an invariant of this action, since rotations preserve the distance to the origin.
    The same is true for $m_{1, 0}^2 + m_{0, 1}^2$ which represents the same polynomial in $x$ and $y$ but is different as a function of the moments and, thus, is considered a different invariant.

    The problem of computing invariants of this action goes back to the work of Hu~\cite{MingKueiHu1962} who proposed a set of seven invariants in the $d = 3$ case.
    This set turned out to be redundant~\cite{Flusser2000} and was not complete (which is a desirable property, see~\cite{Flusser2021}).
    A new irredundant set of invariants was proposed by Flusser~\cite{Flusser2000}, this set was not generating the whole field $k(M_3)^{\operatorname{SO}(2)}$ but formed its transcendence basis.
    Another set was proposed in~\cite{Hickman2011} which was still not a generating set of $k(M_3)^{\operatorname{SO}(2)}$.
    In the same paper, the Hubert-Kogan algorithm was applied to the problem~\cite[p.~231]{Hickman2011} but the resulting invariants were deemed too complicated.
    We computed generators of $k(M_3)^{\operatorname{SO}(2)}$ using the Hubert-Kogan algorithm with the cross-section $m_{1, 1} = 0$ (the cross-section chosen using the same strategy as in~\Cref{ex:triang}).
    The generating set contained 
    polynomials $m_{2, 0} + m_{0, 2}$ and
    $m_{2, 0} m_{0, 2} - m_{1, 1}^2$, and 13 rational functions with degrees ranging from 8 to 10.
    After applying our simplification algorithm (with the result minimized using~\Cref{alg:minimization}) to this generating set, we arrive at the following set of seven polynomial generators:
\begin{align*}
f_1 =& ~m_{2, 0} + m_{0, 2},\quad f_2 = m_{2, 0} m_{0, 2} - m_{1, 1}^2,\\
f_3 =& ~m_{3, 0} m_{1, 2} - m_{2, 1}^2 + m_{2, 1} m_{0, 3} - m_{1, 2}^2, \quad f_4 = m_{3, 0}^2 + 3 m_{3, 0} m_{1, 2} + 3 m_{2,1} m_{0, 3} + m_{0, 3}^2,\\
f_5 =& ~m_{2, 0} m_{2, 1} m_{0, 3} - m_{2, 0} m_{1, 2}^2 - m_{1, 1} m_{3, 0} m_{0, 3} + m_{1, 1} m_{2, 1} m_{1, 2} + m_{0, 2} m_{3, 0} m_{1, 2} - m_{0, 2} m_{2, 1}^2,\\
f_6 =& ~m_{2, 0} m_{3, 0} m_{0, 3} - m_{2, 0} m_{2, 1} m_{1, 2} - 2 m_{1, 1} m_{3, 0} m_{1, 2} + 2 m_{1, 1} m_{2, 1}^2 + 2 m_{1, 1} m_{2, 1}m_{0, 3} - 2 m_{1, 1} m_{1, 2}^2 \\
&- m_{0, 2} m_{3, 0} m_{0, 3} + m_{0, 2} m_{2, 1} m_{1, 2},\\
f_7 =& ~m_{2, 0} m_{3, 0} m_{1, 2} + 2 m_{2, 0} m_{2, 1} m_{0, 3} + m_{2, 0} m_{0, 3}^2 - 2 m_{1, 1} m_{3, 0} m_{2, 1} - m_{1, 1} m_{3, 0} m_{0, 3} - 3 m_{1, 1} m_{2, 1} m_{1, 2}\\
&- 2 m_{1, 1} m_{1, 2} m_{0, 3} + m_{0, 2} m_{3, 0}^2 + 2 m_{0, 2} m_{3, 0} m_{1, 2} + m_{0, 2} m_{2, 1} m_{0,3}.    
\end{align*}
    This generating set has lower degrees than the list from~\cite{Flusser2000} (3 instead of 4) and smaller maximal monomial count than~\cite[p. 231]{Hickman2011} (10 instead of 15).
    On the other hand, our list recovers several classical invariants: $f_1$ and $f_2$ are equivalent to $\phi_1$ and $\phi_2$ in the Hu's list~\cite{MingKueiHu1962} following the notation from~\cite[p.~1405]{Flusser2000}; the same holds for $f_3$ and $f_4$ above and $\phi_3$ and $\phi_4$ ($\phi_3 = f_4 - 9f_3$ and $\phi_4 = f_4 - f_3$).
    This indicates that our automatic simplification procedure can produce simple generators coherent with the results in the pattern recognition literature.
    However, we would like to point out that, in this case, since $\operatorname{SO}_2(\mathbb{C}) \cong \mathbb{C}^*$,
    more refined methods like~\cite{Hubert2013} could be used instead for invariant computation (in the spirit of~\cite{Flusser2000})\footnote{We thank Evelyne Hubert for this remark}.
\end{example}

\bibliographystyle{abbrvnat}
\bibliography{bib}

@inproceedings{MRS2018,
  title = {Algebraic tools for the analysis of state space models},
  ISSN = {0920-1971},
  url = {http://dx.doi.org/10.2969/aspm/07710171},
  booktitle = {The 50th Anniversary of Gr\"{o}bner Bases},
  publisher = {Mathematical Society of Japan},
  author = {Meshkat,  Nicolette and Rosen,  Zvi and Sullivant,  Seth},
  pages = {171–205}
}

@incollection{Muller-Quade99b,
    author = {M{\"u}ller-Quade, J. and Beth, T.},
    booktitle = {Applied algebra, algebraic algorithms and error-correcting codes (Honolulu, HI, 1999)},
    publisher = {Springer},
    series = {Lecture Notes in Computer Science},
    title = {Calculating generators for invariant fields of linear algebraic groups},
    volume = {1719},
    year = {1999},
    url = {https://doi.org/10.1007/3-540-46796-3_37}}

@article{Hubert26Jalard,
    author = {Hubert, Evelyne and Jalard, Martin},
    journal = {to appear in SIAM Journal on Applied Algebra and Geometry},
    url = {https://inria.hal.science/hal-04604969},
    number = {1},
    title = {Algebraically independent generators for the invariant field of $SO_3$ and $O_3$ representations $R^3 \oplus H$},
    volume = {},
    year = {2026}}

@misc{breloer2025rationalinvariantsdegreepolynomials,
      title={Rational invariants of even degree polynomials under the orthogonal group}, 
      author={Henri Breloer},
      year={2025},
      eprint={2501.17504},
      archivePrefix={arXiv},
      primaryClass={math.AC},
      url={https://arxiv.org/abs/2501.17504}, 
}

@article{Hubert25JalardPz,
    author = {Hubert, Evelyne and Jalard, Martin},
    journal = {Journal of Pure and Applied Algebra},
    url = {https://doi.org/10.1016/j.jpaa.2025.108034},
    number = {9},
    pages = {108034},
    title = {Orbit separation and stratification by isotropy classes of piezoelectricity tensors},
    volume = {229},
    year = {2025}}

@article{Gorlach19,
    author = {G{\"o}rlach, P. and Hubert, E. and Papadopoulo, T.},
    journal = {Foundations of Computational Mathematics},
    note = {https://doi.org/10.1007/s10208-018-9404-1},
    pages = {1315-1361},
    title = {{Rational invariants of even ternary forms under the orthogonal group}},
    volume = {19},
    year = {2019}}

@article{HOPY,
  url = {https://doi.org/10.1002/cpa.21921},
  year = {2020},
  volume = {73},
  number = {9},
  pages = {1831--1879},
  author = {Hong, Hoon and Ovchinnikov, Alexey and Pogudin, Gleb and Yap, Chee},
  title = {Global Identifiability of Differential Models},
  journal = {Communications on Pure and Applied Mathematics}
}

@phdthesis{dissertation,
	author = {Binder, Anna Katharina},
	title = {Algorithms for Fields and an Application to a Problem in Computer Vision},
	year = {2009},
	school = {Technische Universität München},
	pages = {179},
	language = {en},
	note = {},
	url = {https://mediatum.ub.tum.de/685465},
}

@article{Prony,
    author = {Prony, R.},
    title = {Essai experimental et analytique sur les lois de la dilatabilite de fluides elastiques et sur celles da la force expansion de la vapeur de l'alcool, a differentes temperatures},
    journal = {J. de l'École Polytechnique Floréal et Plairial, an III},
    volume = 22,
    pages = {24--76},
    year = {1795}
}

@article{demillo-lipton,
title = {A probabilistic remark on algebraic program testing},
journal = {Information Processing Letters},
volume = {7},
number = {4},
pages = {193-195},
year = {1978},
url = {https://doi.org/10.1016/0020-0190(78)90067-4},
author = {Richard A. Demillo and Richard J. Lipton}
}

@article{Schwartz,
author = {Schwartz, Jacob Theodore},
title = {Fast Probabilistic Algorithms for Verification of Polynomial Identities},
year = {1980},
issue_date = {Oct. 1980},
publisher = {Association for Computing Machinery},
address = {New York, NY, USA},
volume = {27},
number = {4},
url = {https://doi.org/10.1145/322217.322225},
journal = {J. ACM},
month = oct,
pages = {701–717},
numpages = {17}
}

@article{demin-vdhl-factoring,
title = {Factoring sparse polynomials fast},
journal = {Journal of Complexity},
volume = {88},
pages = {101934},
year = {2025},
url = {https://doi.org/10.1016/j.jco.2025.101934},
author = {Alexander Demin and Joris {van der Hoeven}}
}

@inproceedings{diversification,
author = {Giesbrecht, Mark and Roche, Daniel S.},
title = {Diversification improves interpolation},
year = {2011},
isbn = {9781450306751},
url = {https://doi.org/10.1145/1993886.1993909},
booktitle = {Proceedings of the 36th International Symposium on Symbolic and Algebraic Computation},
pages = {123–130},
numpages = {8},
location = {San Jose, California, USA},
series = {ISSAC '11}
}

@inproceedings{mqrfr,
author = {Khodadad, Sara and Monagan, Michael},
title = {Fast rational function reconstruction},
year = {2006},
url = {https://doi.org/10.1145/1145768.1145801},
booktitle = {Proceedings of the 2006 International Symposium on Symbolic and Algebraic Computation},
pages = {184–190},
numpages = {7},
location = {Genoa, Italy},
series = {ISSAC '06}
}

@inproceedings{gcd-interpolation,
author = {Hu, Jiaxiong and Monagan, Michael},
title = {A Fast Parallel Sparse Polynomial GCD Algorithm},
year = {2016},
isbn = {9781450343800},
url = {https://doi.org/10.1145/2930889.2930903},
booktitle = {Proceedings of the 2016 ACM International Symposium on Symbolic and Algebraic Computation},
pages = {271–278},
numpages = {8},
location = {Waterloo, ON, Canada},
series = {ISSAC '16}
}

@phdthesis{Boku2016,
  author      = {Dereje Kifle Boku},
  title       = {Gr{\"o}bner Bases over Extention Fields of \(\mathbb{Q}\)},
  pages       = {xi, 150},
  school      = {Technische Universit{\"a}t Kaiserslautern},
  url  = {https://nbn-resolving.de/urn:nbn:de:hbz:386-kluedo-44289},
  year        = {2016},
}

@phdthesis{OllivierPhD,
     author = {Ollivier, François},
     title = {Le probl{\`e}me de l’identifiabilit{\'e} structurelle globale: approche th {\'e}orique, m{\'e}thodes effectives et bornes de complexit{\'e}},
     school = {{\'E}cole polytechnique},
     year = {1990},
     url="https://www.theses.fr/1990EPXX0009"
}

@article{Dankwa2021,
  title = {Estimating vaccination threshold and impact in the 2017–2019 hepatitis {A} virus outbreak among persons experiencing homelessness or who use drugs in {L}ouisville,  {K}entucky,  {U}nited {S}tates},
  volume = {39},
  url = {http://dx.doi.org/10.1016/j.vaccine.2021.10.001},
  number = {49},
  journal = {Vaccine},
  author = {Dankwa,  Emmanuelle A. and Donnelly,  Christl A. and Brouwer,  Andrew F. and Zhao,  Rui and Montgomery,  Martha P. and Weng,  Mark K. and Martin,  Natasha K.},
  year = {2021},
  month = dec,
  pages = {7182–7190}
}

@article{Dankwa2022,
  title = {Structural identifiability of compartmental models for infectious disease transmission is influenced by data type},
  volume = {41},
  url = {http://dx.doi.org/10.1016/j.epidem.2022.100643},
  journal = {Epidemics},  
  author = {Dankwa,  Emmanuelle A. and Brouwer,  Andrew F. and Donnelly,  Christl A.},
  year = {2022},
  month = dec,
  pages = {100643}
}

@article{Berk1969,
  title = {Studies of bilirubin kinetics in normal adults},
  volume = {48},
  url = {http://dx.doi.org/10.1172/JCI106184},
  number = {11},
  journal = {Journal of Clinical Investigation},
  author = {Berk,  Paul D. and Howe,  Robert B. and Bloomer,  Joseph R. and Berlin,  Nathaniel I.},
  year = {1969},
  month = nov,
  pages = {2176–2190}
}

@article{OPT,
  url = {https://doi.org/10.1007/s00200-021-00486-8},
  year = {2021},
  volume = {34},
  number = {2},
  pages = {165--182},
  author = {Ovchinnikov, Alexey and Pogudin, Gleb and Thompson, Peter},
  title = {Parameter identifiability and input-output equations},
  journal = {Applicable Algebra in Engineering,  Communication and Computing}
}

@unpublished{corniquel:hal-05487354,
  TITLE = {{Solving parametric polynomial systems using Generic Rational Univariate Representation}},
  AUTHOR = {Corniquel, Florent},
  URL = {https://hal.science/hal-05487354},
  NOTE = {working paper or preprint},
  YEAR = {2026},
  MONTH = Jan,
  HAL_ID = {hal-05487354},
  HAL_VERSION = {v1},
}

@article{allident,
  url = {https://doi.org/10.1016/j.sysconle.2021.105030},
  year = {2021},
  volume = {157},
  pages = {105030},
  author = {Ovchinnikov, Alexey and Pillay, Anand and Pogudin, Gleb and Scanlon, Thomas},
  title = {Computing all identifiable functions of parameters for {ODE} models},
  journal = {Systems {$\&$} Control Letters}
}

@article{LGEL2011,
  title = "Difference algebra and system identification",
  journal = "Automatica",
  volume = "47",
  number = "9",
  pages = "1896 - 1904",
  year = "2011",
  url = "https://doi.org/10.1016/j.automatica.2011.06.013",
  author = "Lyzell, Christian and Glad, Torkel and Enqvist, Martin and Ljung, Lennart"
}

@article{Kawano2011,
  title = {An Algebraic Approach to Local Observability at an Initial State for Discrete-Time Polynomial Systems},
  volume = {44},
  url = {http://dx.doi.org/10.3182/20110828-6-IT-1002.00336},
  number = {1},
  journal = {IFAC Proceedings Volumes},
  author = {Kawano,  Yu and Ohtsuka,  Toshiyuki},
  year = {2011},
  month = jan,
  pages = {6449–6453}
}

@book{Allman2003,
  title = {Mathematical Models in Biology: An Introduction},
  ISBN = {9780511790911},
  url = {http://dx.doi.org/10.1017/CBO9780511790911},
  publisher = {Cambridge University Press},
  author = {Allman,  Elizabeth S. and Rhodes,  John A.},
  year = {2003},
  month = oct 
}

@article{Allen1994,
  title = {Some discrete-time {SI}, {SIR}, and {SIS} epidemic models},
  volume = {124},
  url = {http://dx.doi.org/10.1016/0025-5564(94)90025-6},
  number = {1},
  journal = {Mathematical Biosciences},
  author = {Allen,  Linda J.S.},
  year = {1994},
  month = nov,
  pages = {83–105}
}

@article{modstd,
title = {Parallelization of Modular Algorithms},
journal = {Journal of Symbolic Computation},
volume = {46},
number = {6},
pages = {672-684},
year = {2011},
url = {https://doi.org/10.1016/j.jsc.2011.01.003},
author = {Nazeran Idrees and Gerhard Pfister and Stefan Steidel},
}

@article{ReyBarreiro2023,
  title = {Benchmarking tools for a priori identifiability analysis},
  volume = {39},
  url = {http://dx.doi.org/10.1093/bioinformatics/btad065},
  number = {2},
  journal = {Bioinformatics},
  author = {Rey Barreiro,  Xabier and Villaverde,  Alejandro F},
  editor = {Wren,  Jonathan},
  year = {2023},
  month = jan 
}

@misc{ffmodstd,
    author={Boku, D. K. and Decker, W. and Fieker, C.},
    year={2016},
    title = {ffmodstd.lib: A {S}ingular version 4-0-3 library for computing {G}roebner bases of ideals in polynomial rings over algebraic function fields}
}

@inproceedings{symbolicdata,
    author = {Gr\"abe, Hans-Gert},
    title = {{T}he {S}ymbolic{D}ata {B}enchmark {P}roblems {C}ollection of {P}olynomial {S}ystems},
    booktitle = {Proceedings of the Workshop on Under- and Overdetermined Systems of Algebraic or Differential Equations},
    pages = {57-76},
    publisher = {IAS, Univ. Karlsruhe},
    url = {https://symbolicdata.github.io/Papers/karlsruhe-02.pdf},
    year = {2002}
}

@article{Hubert2013,
  title = {Scaling Invariants and Symmetry Reduction of Dynamical Systems},
  volume = {13},
  url = {http://dx.doi.org/10.1007/s10208-013-9165-9},
  number = {4},
  journal = {Foundations of Computational Mathematics},
  author = {Hubert,  Evelyne and Labahn,  George},
  year = {2013},
  month = aug,
  pages = {479–516}
}

@article{ratrec82,
author = {Wang, Paul S. and Guy, M. J. T. and Davenport, J. H.},
title = {P-adic reconstruction of rational numbers},
year = {1982},
issue_date = {May 1982},
publisher = {Association for Computing Machinery},
address = {New York, NY, USA},
volume = {16},
number = {2},
url = {https://doi.org/10.1145/1089292.1089293},
journal = {SIGSAM Bull.},
month = may,
pages = {2–3},
numpages = {2}
}

@misc{ilmer2022efficientidentifiabilityverificationode,
      title={More Efficient Identifiability Verification in {ODE} Models by Reducing Non-Identifiability}, 
      author={Ilmer, Ilia and Ovchinnikov, Alexey and Pogudin, Gleb and Soto, Pedro},
      year={2022},
      url={https://arxiv.org/abs/2204.01623}, 
}

@inbook{Ilmer2021,
  title = {Web-Based Structural Identifiability Analyzer},
  ISBN = {9783030856335},
  ISSN = {1611-3349},
  url = {http://dx.doi.org/10.1007/978-3-030-85633-5_17},
  booktitle = {Computational Methods in Systems Biology},
  publisher = {Springer International Publishing},
  author = {Ilmer,  Ilia and Ovchinnikov,  Alexey and Pogudin,  Gleb},
  year = {2021},
  pages = {254–265}
}

@article{ehrenborg1993apolarity,
  title={Apolarity and canonical forms for homogeneous polynomials},
  author={Ehrenborg, Richard and Rota, Gian-Carlo},
  journal={European Journal of Combinatorics},
  volume={14},
  number={3},
  pages={157--181},
  year={1993},
  url={https://doi.org/10.1006/eujc.1993.1022}
}

@article{stident,
  title = {Differential Elimination for Dynamical Models via Projections with Applications to Structural Identifiability},
  volume = {7},
  url = {http://dx.doi.org/10.1137/22M1469067},
  number = {1},
  journal = {SIAM Journal on Applied Algebra and Geometry},
  publisher = {Society for Industrial & Applied Mathematics (SIAM)},
  author = {Dong,  Ruiwen and Goodbrake,  Christian and Harrington,  Heather A. and Pogudin,  Gleb},
  year = {2023},
  month = mar,
  pages = {194–235}
}

@misc{smith1997rationalnonrationalalgebraicvarieties,
      title={Rational and Non-Rational Algebraic Varieties: Lectures of {J\'a}nos {K}oll{\'a}r}, 
      author={Smith, Karen E.  and Rosenberg, Joel},
      year={1997},
      eprint={alg-geom/9707013},
      archivePrefix={arXiv},
      primaryClass={alg-geom},
      url={https://arxiv.org/abs/alg-geom/9707013}, 
}

@book{MCA,
	Author = {von~zur~Garthen, Joachim and Gerhard, Jürgen},
	Publisher = {Cambridge University Press},
	Title = {Modern Computer Algebra},
	Year = {2013},
    edition={3rd}
}

@article{F4,
title = {{A} new efficient algorithm for computing {G}r{\"o}bner bases ({F}4)},
journal = {Journal of Pure and Applied Algebra},
volume = {139},
number = {1},
pages = {61-88},
year = {1999},
url = {https://doi.org/10.1016/S0022-4049(99)00005-5},
author = {Faugére, Jean-Charles}
}

@article{modular,
title = {Modular algorithms for computing {G}röbner bases},
journal = {Journal of Symbolic Computation},
volume = {35},
number = {4},
pages = {403-419},
year = {2003},
url = {https://doi.org/10.1016/S0747-7171(02)00140-2},
author = {Arnold, Elizabeth A.}
}

@inproceedings{zippel-lemma,
author = {Zippel, Richard},
title = {Probabilistic algorithms for sparse polynomials},
year = {1979},
isbn = {3540095195},
publisher = {Springer-Verlag},
address = {Berlin, Heidelberg},
booktitle = {Proceedings of the International Symposiumon on Symbolic and Algebraic Computation},
pages = {216–226},
numpages = {11},
series = {EUROSAM '79}
}

@ARTICLE{on-the-origins,
  author={Barreiro, Xabier Rey and Villaverde, Alejandro F.},
  journal={IEEE Access}, 
  title={On the Origins and Rarity of Locally but Not Globally Identifiable Parameters in Biological Modeling}, 
  year={2023},
  volume={11},
  pages={65457-65467},
  url = {https://doi.org/10.1109/ACCESS.2023.3288998}
}

@inproceedings{monagan-parametric-linsys,
author = {Jinadu, Ayoola and Monagan, Michael},
title = {Solving Parametric Linear Systems Using Sparse Rational Function Interpolation},
year = {2023},
isbn = {978-3-031-41723-8},
publisher = {Springer-Verlag},
address = {Berlin, Heidelberg},
url = {https://doi.org/10.1007/978-3-031-41724-5_13},
booktitle = {Computer Algebra in Scientific Computing, CASC 2023},
pages = {233–254},
numpages = {22},
location = {Havana, Cuba}
}

@article{Hong2019,
  title = {{SIAN}: software for structural identifiability analysis of {ODE} models},
  volume = {35},
  url = {http://dx.doi.org/10.1093/bioinformatics/bty1069},
  number = {16},
  journal = {Bioinformatics},
  author = {Hong,  Hoon and Ovchinnikov,  Alexey and Pogudin,  Gleb and Yap,  Chee},
  year = {2019},
  pages = {2873–2874}
}

@phdthesis{glebhdr,
  TITLE = {Symbolic Transformations of Dynamical Models},
  AUTHOR = {Pogudin, Gleb},
  URL = {https://hal.science/tel-05041736},
  SCHOOL = {{Institute Polytechnique de Paris}},
  YEAR = {2024},
  TYPE = {Habilitation {\`a} diriger des recherches},
  HAL_ID = {tel-05041736},
  HAL_VERSION = {v1},
}

@article{vdhl-interpol,
  title = {Sparse polynomial interpolation: faster strategies over finite fields},
  author = {van der Hoeven, Joris and Lecerf, Gr{\'e}goire},
  url = {https://doi.org/10.1007/s00200-024-00655-5},
  journal = {Applicable Algebra in Engineering, Communication and Computing},
  publisher = {Springer Verlag},
  year = {2024},
  month = {apr}
}

@article{gvdhl-det-roots,
author = {Grenet, Bruno and {van der Hoeven}, Joris and Lecerf, Grégoire},
year = {2016},
month = {06},
pages = {237-257},
title = {Deterministic root finding over finite fields using {G}raeffe transforms},
volume = {27},
journal = {Applicable Algebra in Engineering, Communication and Computing},
url = {https://doi.org/10.1007/s00200-015-0280-5}
}

@article{cantor-zassenhaus-81,
  title={A new algorithm for factoring polynomials over finite fields},
  author={David Geoffrey Cantor and Hans Zassenhaus},
  journal={Mathematics of Computation},
  year={1981},
  volume={36},
  pages={587-592},
  url = {https://doi.org/10.2307/2007663}
}

@article{vdhl,
author = {van der Hoeven, Joris and Lecerf, Gr\'{e}goire},
title = {On Sparse Interpolation of Rational Functions and Gcds},
year = {2021},
publisher = {Association for Computing Machinery},
address = {New York, NY, USA},
volume = {55},
number = {1},
url = {https://doi.org/10.1145/3466895.3466896},
journal = {ACM Commun. Comput. Algebra},
month = may,
pages = {1–12},
numpages = {12}
}

@article{JimnezPastor2022,
  title = {Computing exact nonlinear reductions of dynamical models},
  volume = {56},
  url = {http://dx.doi.org/10.1145/3572867.3572869},
  number = {2},
  journal = {ACM Communications in Computer Algebra},
  author = {Jiménez-Pastor,  Antonio and Pogudin,  Gleb},
  year = {2022},
  month = jun,
  pages = {25–31}
}

@inproceedings{ben-or-tiwari,
author = {Ben-Or, Michael and Tiwari, Prasoon},
title = {A deterministic algorithm for sparse multivariate polynomial interpolation},
year = {1988},
url = {https://doi.org/10.1145/62212.62241},
booktitle = {Proceedings of the Twentieth Annual ACM Symposium on Theory of Computing},
pages = {301–309},
numpages = {9},
location = {Chicago, Illinois, USA},
series = {STOC '88}
}

@article{slimgb,
author = {Brickenstein, Michael},
year = {2010},
month = {07},
pages = {453-466},
title = {{S}limgb: {G}r{\"o}bner {B}ases with {S}lim {P}olynomials},
volume = {23},
journal = {Revista Matem{\'a}tica Complutense},
url = {https://doi.org/10.1007/s13163-009-0020-0}
}

@inproceedings{splitting,
author = {Monagan, Michael and Pearce, Roman},
title = {An Algorithm For Splitting Polynomial Systems Based On {F}4},
year = {2017},
url = {https://doi.org/10.1145/3115936.3115948},
booktitle = {Proceedings of the International Workshop on Parallel Symbolic Computation},
articleno = {12},
numpages = {5},
location = {Kaiserslautern, Germany},
series = {PASCO 2017}
}

@InProceedings{f4-variant-tracing,
author="Joux, Antoine
and Vitse, Vanessa",
editor="Kiayias, Aggelos",
title="A Variant of the {F}4 Algorithm",
booktitle="Topics in Cryptology -- CT-RSA 2011",
year="2011",
publisher="Springer Berlin Heidelberg",
address="Berlin, Heidelberg",
pages="356--375",
isbn="978-3-642-19074-2",
url = {https://doi.org/10.1007/978-3-642-19074-2_23}
}

@InProceedings{transposed_vandermonde,
author="Kaltofen, Erich
and Yagati, Lakshman",
editor="Gianni, P.",
title="Improved sparse multivariate polynomial interpolation algorithms",
booktitle="Symbolic and Algebraic Computation",
year="1989",
publisher="Springer Berlin Heidelberg",
address="Berlin, Heidelberg",
pages="467--474",
url = {https://doi.org/10.1007/3-540-51084-2_44}
}

@article{early-termination-lee,
title = {Early termination in sparse interpolation algorithms},
journal = {Journal of Symbolic Computation},
volume = {36},
number = {3},
pages = {365-400},
year = {2003},
url = {https://doi.org/10.1016/S0747-7171(03)00088-9},
author = {Kaltofen, Erich and Lee, Wen-shin},
}

@inproceedings{adaptive-step,
author = {Kaltofen, Erich and Lee, Wen-shin and Lobo, Austin A.},
title = {Early termination in {B}en-{O}r/{T}iwari sparse interpolation and a hybrid of {Z}ippel's algorithm},
year = {2000},
url = {https://doi.org/10.1145/345542.345629},
booktitle = {Proceedings of the 2000 International Symposium on Symbolic and Algebraic Computation},
pages = {192–201},
numpages = {10},
location = {St. Andrews, Scotland}
}

@article{cuyt-lee,
title = {Sparse interpolation of multivariate rational functions},
journal = {Theoretical Computer Science},
volume = {412},
number = {16},
pages = {1445-1456},
year = {2011},
url = {https://doi.org/10.1016/j.tcs.2010.11.050},
author = {Cuyt, Annie and Lee, Wen-shin}
}

@InProceedings{tracingtrav,
author="Traverso, Carlo",
editor="Gianni, P.",
title="Gr{\"o}bner trace algorithms",
booktitle="Symbolic and Algebraic Computation",
year="1989",
publisher="Springer Berlin Heidelberg",
address="Berlin, Heidelberg",
pages="125--138",
url="https://doi.org/10.1007/3-540-51084-2\_12"
}

@misc{ahmed2025identifiabilitydirectedcyclecatenarylinear,
  title = {Identifiability of Directed-Cycle and Catenary Linear Compartmental Models},
  volume = {25},
  url = {http://dx.doi.org/10.1137/25M1728636},
  number = {1},
  journal = {SIAM Journal on Applied Dynamical Systems},
  author = {Ahmed,  Saber and Crepeau,  Natasha and Dessauer,  Paul R. and Edozie,  Alexis and Garcia-Lopez,  Odalys and Grimsley,  Tanisha and Garcia,  Jordy Lopez and Neri,  Viridiana and Shiu,  Anne},
  year = {2026},
  pages = {304–350}
}

@article{bezanson2017julia,
  title={Julia: A fresh approach to numerical computing},
  author={Bezanson, Jeff and Edelman, Alan and Karpinski, Stefan and Shah, Viral B},
  journal={SIAM review},
  volume={59},
  number={1},
  pages={65--98},
  year={2017},
  publisher={SIAM},
  url={https://doi.org/10.1137/141000671}
}

@inproceedings{AbstractAlgebra.jl-2017,
  author = {Fieker, Claus and Hart, William and Hofmann, Tommy and Johansson, Fredrik},
  title = {Nemo/{H}ecke: Computer Algebra and Number Theory Packages for the {J}ulia Programming Language},
  booktitle = {Proceedings of the 2017 ACM on International Symposium on Symbolic and Algebraic Computation},
  series = {ISSAC '17},
  year = {2017},
  pages = {157--164},
  numpages = {8},
  url = {https://doi.acm.org/10.1145/3087604.3087611}
}

@misc{groebnerjl2023,
  title = {Groebner.jl: A package for {G}r\"obner bases computations in {J}ulia}, 
  author = {Alexander Demin and Shashi Gowda},
  year = {2023},
  eprint = {2304.06935},
  url = {https://arxiv.org/abs/2304.06935}
}

@inbook{Beauville2016,
  title = {The L\"{u}roth Problem},
  ISBN = {9783319462097},
  url = {http://dx.doi.org/10.1007/978-3-319-46209-7_1},
  booktitle = {Rationality Problems in Algebraic Geometry},
  publisher = {Springer International Publishing},
  author = {Beauville,  Arnaud},
  year = {2016},
  pages = {1–27}
}

@article{combinations_combos_algo,
title = {An algorithm for finding globally identifiable parameter combinations of nonlinear ODE models using {G}röbner Bases},
journal = {Mathematical Biosciences},
volume = {222},
number = {2},
pages = {61-72},
year = {2009},
url = {https://doi.org/10.1016/j.mbs.2009.08.010},
author = {Nicolette Meshkat and Marisa Eisenberg and Joseph J. {DiStefano III}}
}

@article{Meshkat2011,
  title = {Finding identifiable parameter combinations in nonlinear {ODE} models and the rational reparameterization of their input–output equations},
  volume = {233},
  url = {http://dx.doi.org/10.1016/j.mbs.2011.06.001},
  number = {1},
  journal = {Mathematical Biosciences},
  author = {Meshkat,  Nicolette and Anderson,  Chris and Joseph J. {DiStefano III}},
  year = {2011},
  pages = {19–31}
}

@article{combos,
author = {Meshkat, Nicolette and Kuo, Christine and Joseph J. {DiStefano III}},
year = {2014},
month = {10},
pages = {e110261},
title = {On Finding and Using Identifiable Parameter Combinations in Nonlinear Dynamic Systems Biology Models and {COMBOS}: A Novel Web Implementation},
volume = {9},
journal = {PloS one},
url = {https://doi.org/10.1371/journal.pone.0110261}
}

@misc{mukhina2025projectingdynamicalsystemssupport,
      title={Projecting dynamical systems via a support bound}, 
      author={Mukhina, Yulia and Pogudin, Gleb},
      year={2025},
      url={https://arxiv.org/abs/2501.13680}, 
}

@article{MULLERQUADE1999143,
title = {Basic Algorithms for Rational Function Fields},
journal = {Journal of Symbolic Computation},
volume = {27},
number = {2},
pages = {143-170},
year = {1999},
url = {https://doi.org/10.1006/jsco.1998.0246},
author = {J. Müller-Quade and R. Steinwandt}
}

@article{MllerQuade2000,
  title = {Gr\"{o}bner Bases Applied to Finitely Generated Field Extensions},
  volume = {30},
  url = {http://dx.doi.org/10.1006/jsco.1999.0417},
  number = {4},
  journal = {Journal of Symbolic Computation},
  author = {M\"{u}ller-Quade,  J\"{o}rn and Steinwandt,  Rainer},
  year = {2000},
  month = oct,
  pages = {469–490}
}

@inproceedings{Gutierrez2001,
  series = {ISSAC01},
  title = {Unirational fields of transcendence degree one and functional decomposition},
  url = {http://dx.doi.org/10.1145/384101.384124},
  booktitle = {Proceedings of the 2001 international symposium on Symbolic and algebraic computation},
  publisher = {ACM},
  author = {Gutierrez,  Jamie and Rubio,  Rosario and Sevilla,  David},
  year = {2001},
  month = jul,
  pages = {167–174},
  collection = {ISSAC'01}
}

@book{Lang,
  author = {Lang, Serge},
title = {Algebra},
edition = {3rd},
year = {2002},
publisher = {Springer},
url = {https://doi.org/10.1007/978-1-4613-0041-0}
}

@inbook{Robbiano1990,
  title = {Subalgebra bases},
  ISBN = {9783540471363},
  url = {http://dx.doi.org/10.1007/BFb0085537},
  booktitle = {Commutative Algebra},
  publisher = {Springer Berlin Heidelberg},
  author = {Robbiano,  Lorenzo and Sweedler,  Moss},
  year = {1990},
  pages = {61–87}
}

@inbook{Sweedler1993,
  title = {Using {G}r{\"o}bner bases to determine the algebraic and transcendental nature of field extensions: Return of the killer tag variables},
  ISBN = {9783540476306},
  url = {http://dx.doi.org/10.1007/3-540-56686-4_34},
  booktitle = {Applied Algebra,  Algebraic Algorithms and Error-Correcting Codes},
  publisher = {Springer Berlin Heidelberg},
  author = {Sweedler,  Moss},
  year = {1993},
  pages = {66–75}
}

@article{Kemper1996,
  title = {A constructive approach to {N}oether’s problem},
  volume = {90},
  url = {http://dx.doi.org/10.1007/BF02568311},
  number = {1},
  journal = {Manuscripta Mathematica},
  author = {Kemper,  Gregor},
  year = {1996},
  month = dec,
  pages = {343–363}
}

@inproceedings{Binder1996,
  series = {ISSAC ’96},
  title = {Fast computations in the lattice of polynomial rational function fields},
  url = {http://dx.doi.org/10.1145/236869.236895},
  booktitle = {Proceedings of the 1996 international symposium on Symbolic and algebraic computation},
  publisher = {ACM Press},
  author = {Binder,  Franz},
  year = {1996},
  pages = {43–48},
  collection = {ISSAC ’96}
}

@article{Kuroda2005,
  title = {A counterexample to the fourteenth problem of {H}ilbert in dimension three},
  volume = {53},
  url = {http://dx.doi.org/10.1307/MMJ/1114021089},
  number = {1},
  journal = {Michigan Mathematical Journal},
  author = {Kuroda,  Shigeru},
  year = {2005}, 
}

@article{Nagata1959,
  title = {On the 14-th Problem of {H}ilbert},
  volume = {81},
  url = {http://dx.doi.org/10.2307/2372927},
  number = {3},
  journal = {American Journal of Mathematics},
  author = {Nagata,  Masayoshi},
  year = {1959},
  pages = {766}
}

@book{milne2022,
  author={Milne, J. S.},
  title={Fields and Galois Theory},
  year={2022},
  publisher={Kea Books},
  address={Ann Arbor, MI}
}

@misc{chen2024practicalidentifiabilityparameterestimation,
      title={Practical identifiability and parameter estimation of compartmental epidemiological models}, 
      author={Chen, Q. Y. and Rapti, Z. and Drossinos, Y.  and Cuevas-Maraver, J.  and Kevrekidis, G. A.  and Kevrekidis, P. G.},
      year={2024},
      url={https://arxiv.org/abs/2406.17827} 
}

@article{MASSONIS2021441,
title = {Structural identifiability and observability of compartmental models of the {COVID}-19 pandemic},
journal = {Annual Reviews in Control},
volume = {51},
pages = {441-459},
year = {2021},
url = {https://doi.org/10.1016/j.arcontrol.2020.12.001},
author = {Gemma Massonis and Julio R. Banga and Alejandro F. Villaverde}
}

@article{Hubert2025,
  title = {Rationality of the invariant field for a class of representations of the real orthogonal groups},
  volume = {682},
  url = {http://dx.doi.org/10.1016/j.jalgebra.2025.05.033},
  journal = {Journal of Algebra},
  author = {Hubert,  Evelyne and Jalard,  Martin},
  year = {2025},
  pages = {109–130}
}

@article{Bandeira2023,
  title = {Estimation under group actions: Recovering orbits from invariants},
  volume = {66},
  url = {http://dx.doi.org/10.1016/j.acha.2023.06.001},
  journal = {Applied and Computational Harmonic Analysis},
  author = {Bandeira,  Afonso S. and Blum-Smith,  Ben and Kileel,  Joe and Niles-Weed,  Jonathan and Perry,  Amelia and Wein,  Alexander S.},
  year = {2023},
  month = sep,
  pages = {236–319}
}

@book{Weyl,
    author = {Weyl, Hermann},
    title = {The Classical Groups: Their Invariants and Representations},
    publisher = {Princeton University Press},
    year = {1997},
    edition = {2nd}
}

@inproceedings{Kogan2023,
  series = {ISSAC 2023},
  title = {Invariants: Computation and Applications},
  url = {http://dx.doi.org/10.1145/3597066.3597149},
  booktitle = {Proceedings of the 2023 International Symposium on Symbolic and Algebraic Computation},
  publisher = {ACM},
  author = {Kogan,  Irina A.},
  year = {2023},
  pages = {31–40},
  collection = {ISSAC 2023}
}

@article{Flusser2000,
  title = {On the independence of rotation moment invariants},
  volume = {33},
  url = {http://dx.doi.org/10.1016/S0031-3203(99)00127-2},
  number = {9},
  journal = {Pattern Recognition},
  author = {Flusser,  Jan},
  year = {2000},
  month = sep,
  pages = {1405–1410}
}

@article{Hubert2007,
  title = {Rational invariants of a group action. {C}onstruction and rewriting},
  volume = {42},
  url = {http://dx.doi.org/10.1016/j.jsc.2006.03.005},
  number = {1–2},
  journal = {Journal of Symbolic Computation},
  author = {Hubert,  Evelyne and Kogan,  Irina A.},
  year = {2007},
  month = jan,
  pages = {203–217}
}

@article{Hickman2011,
  title = {Geometric Moments and Their Invariants},
  volume = {44},
  url = {http://dx.doi.org/10.1007/s10851-011-0323-x},
  number = {3},
  journal = {Journal of Mathematical Imaging and Vision},
  author = {Hickman,  Mark S.},
  year = {2011},
  month = dec,
  pages = {223–235}
}

@article{MingKueiHu1962,
  title = {Visual pattern recognition by moment invariants},
  volume = {8},
  ISSN = {0018-9448},
  url = {http://dx.doi.org/10.1109/TIT.1962.1057692},
  number = {2},
  journal = {IEEE Transactions on Information Theory},
  author = {Ming-Kuei Hu},
  year = {1962},
  month = feb,
  pages = {179–187}
}

@misc{Aida,
    url = {https://www-sop.inria.fr/members/Evelyne.Hubert/aida/},
    author = {Hubert, Evelyne},
    note = {accessed on November 3, 2025},
    title = {The {AIDA} {M}aple package }
}

@book{Flusser2009,
  title = {Moments and Moment Invariants in Pattern Recognition},
  ISBN = {9780470684757},
  url = {http://dx.doi.org/10.1002/9780470684757},
  publisher = {Wiley},
  author = {Flusser,  Jan and Suk,  Tomáš and Zitová,  Barbara},
  year = {2009},
  month = oct 
}

@article{Flusser2021,
  title = {Complete and Incomplete Sets of Invariants},
  volume = {63},
  url = {http://dx.doi.org/10.1007/s10851-021-01039-x},
  number = {7},
  journal = {Journal of Mathematical Imaging and Vision},
  author = {Flusser,  Jan and Suk,  Tomáš and Zitová,  Barbara},
  year = {2021},
  month = may,
  pages = {917–922}
}

@article{Kemper2007,
  title = {The Computation of Invariant Fields and a Constructive Version of a Theorem by {R}osenlicht},
  volume = {12},
  url = {http://dx.doi.org/10.1007/s00031-007-0056-5},
  number = {4},
  journal = {Transformation Groups},
  author = {Kemper,  Gregor},
  year = {2007},
  month = nov,
  pages = {657–670}
}

@misc{blumsmith2025genericorbitsnormalbases,
      title={Generic orbits, normal bases, and generation degree for fields of rational invariants}, 
      author={Blum-Smith, Ben and Derksen, Harm},
      year={2025},
      url={https://arxiv.org/abs/2506.05650}, 
}

@article{Tuncer2018,
  title = {Structural and practical identifiability analysis of outbreak models},
  volume = {299},
  url = {http://dx.doi.org/10.1016/j.mbs.2018.02.004},
  journal = {Mathematical Biosciences},
  author = {Tuncer,  Necibe and Le,  Trang T.},
  year = {2018},
  pages = {1–18}
}

@article{Remien2021GLVIdentifiability,
  title = {Structural identifiability of the generalized {L}otka--{V}olterra model for microbiome studies},
  author = {Remien, Christopher H. and Eckwright, Mariah J. and Ridenhour, Benjamin J.},
  journal = {Royal Society Open Science},
  volume = {8},
  number = {7},
  pages = {20378},
  year = {2021},
  url = {https://doi.org/10.1098/rsos.201378}
}

@article{Demignot1987Pharm,
	title   = {Effect of Prosthetic Sugar Groups on the Pharmacokinetics of Glucose-Oxidase},
	author  = {Demignot, S. and Domurado, D.},
	journal = {Drug Design and Delivery},
	year    = {1987},
	volume  = {1},
	number  = {4},
	pages   = {333--348},
	month   = {May},
	url     = {https://pubmed.ncbi.nlm.nih.gov/2855567/}
}

@article{Chappell1990Global,
  title = {Global identifiability of the parameters of nonlinear systems with specified inputs: A comparison of methods},
  author = {Chappell, Michael J. and Godfrey, Keith R. and Vajda, Sandor},
  journal = {Mathematical Biosciences},
  volume = {102},
  number = {1},
  pages = {41--73},
  year = {1990},
  url = {https://doi.org/10.1016/0025-5564(90)90055-4}
}

@article{Fujita2010,
    title   = {Decoupling of Receptor and Downstream Signals in the {A}kt Pathway by Its Low-Pass Filter Characteristics},
    author  = {Fujita, K. A. and Toyoshima, Y. and Uda, S. and Ozaki, Y. I. and Kubota, H. and Kuroda, S.},
    journal = {Science Signaling},
    year    = {2010},
    volume  = {3},
    number  = {132},
    pages   = {ra56--ra56},
    url     = {https://doi.org/10.1126/scisignal.2000810}
}

@article{Raia2011,
    title   = {Dynamic Mathematical Modeling of {IL13}-Induced Signaling in Hodgkin and Primary Mediastinal {B}-Cell Lymphoma Allows Prediction of Therapeutic Targets},
    author  = {Raia, Valentina and Schilling, Marcel and B{\"o}hm, Martin and Hahn, Bettina and Kowarsch, Andreas and Raue, Andreas and Sticht, Carsten and Bohl, Sebastian and Saile, Maria and M{\"o}ller, Peter and Gretz, Norbert and Timmer, Jens and Theis, Fabian and Lehmann, Wolf-Dieter and Lichter, Peter and Klingm{\"u}ller, Ursula},
    journal = {Cancer Research},
    year    = {2011},
    volume  = {71},
    number  = {3},
    pages   = {693--704},
    url     = {https://doi.org/10.1158/0008-5472.CAN-10-2987}
}

@article{MANRAI20085533,
title = {The Geometry of Multisite Phosphorylation},
journal = {Biophysical Journal},
volume = {95},
number = {12},
pages = {5533-5543},
year = {2008},
url = {https://doi.org/10.1529/biophysj.108.140632},
author = {Arjun Kumar Manrai and Jeremy Gunawardena},
}

@article{MOATE2008731,
  title = {Kinetics of Ruminal Lipolysis of Triacylglycerol and Biohydrogenation of Long-Chain Fatty Acids: New Insights from Old Data},
  journal = {Journal of Dairy Science},
  volume = {91},
  number = {2},
  pages = {731-742},
  year = {2008},
  url = {https://doi.org/10.3168/jds.2007-0398},
  author = {P.J. Moate and R.C. Boston and T.C. Jenkins and I.J. Lean},
}

@article{Zha2020,
    title   = {Research about the optimal strategies for prevention and control of varicella outbreak in a school in a central city of {C}hina: based on an {SEIR} dynamic model},
    author  = {Zha, Wen-ting and Pang, Fen-rui and Zhou, Nan and Wu, Bin and Liu, Ying and Du, Yan-bing and Hong, Xiu-qin and Lv, Yuan},
    journal = {Epidemiology and Infection},
    year    = {2020},
    volume  = {148},
    pages   = {e56},
    url     = {https://doi.org/10.1017/S0950268819002188}
}

@inproceedings{kronecker,
author = {Arnold, Andrew and Roche, Daniel S.},
title = {Multivariate sparse interpolation using randomized {K}ronecker substitutions},
year = {2014},
url = {https://doi.org/10.1145/2608628.2608674},
booktitle = {Proceedings of the 39th International Symposium on Symbolic and Algebraic Computation},
pages = {35–42},
numpages = {8},
location = {Kobe, Japan},
series = {ISSAC '14}
}

@article{Bruno2016,
    author = {Bruno, Mark and Koschmieder, Julian and Wuest, Florian and Schaub, Patrick and Fehling-Kaschek, Mirjam and Timmer, Jens and Beyer, Peter and Al-Babili, Salim},
    title = {Enzymatic study on {AtCCD4} and {AtCCD7} and their potential to form acyclic regulatory metabolites},
    journal = {Journal of Experimental Botany},
    volume = {67},
    number = {21},
    pages = {5993-6005},
    year = {2016},
    month = {10},
    url = {https://doi.org/10.1093/jxb/erw356},
    eprint = {https://academic.oup.com/jxb/article-pdf/67/21/5993/20831700/erw356.pdf},
}

@article{SedoglavicLocal,
title = {A Probabilistic Algorithm to Test Local Algebraic Observability in Polynomial Time},
journal = {Journal of Symbolic Computation},
volume = {33},
number = {5},
pages = {735-755},
year = {2002},
url = {https://doi.org/10.1006/jsco.2002.0532},
author = {Alexandre Sedoglavic},
}

@article{CRN,
title = {Dynamics of Posttranslational Modification Systems: Recent Progress and Future Directions},
journal = {Biophysical Journal},
volume = {114},
number = {3},
pages = {507-515},
year = {2018},
url = {https://doi.org/10.1016/j.bpj.2017.11.3787},
author = {Carsten Conradi and Anne Shiu}
}

@article{Covid1,
title = {A novel {COVID}-19 epidemiological model with explicit susceptible and asymptomatic isolation compartments reveals unexpected consequences of timing social distancing},
journal = {Journal of Theoretical Biology},
volume = {510},
pages = {110539},
year = {2021},
url = {https://doi.org/10.1016/j.jtbi.2020.110539},
author = {Jana L. Gevertz and James M. Greene and Cynthia H. Sanchez-Tapia and Eduardo D. Sontag}
}

@misc{IntroToSeirSlides,
      title={Introduction to {SEIR} models}, 
      author={Chitnis Nakul},
      year={2017},
      url={https://indico.ictp.it/event/7960/session/3/contribution/19/material/slides/0.pdf}, 
}

@article{Crauste,
title = {Identification of Nascent Memory {CD8} {T} Cells and Modeling of Their Ontogeny},
journal = {Cell Systems},
volume = {4},
issue = {3},
year = {2017},
issn = {0022-5193},
pages = {306-317},
url = {https://doi.org/10.1016/j.cels.2017.01.014},
author = {Crauste, Fabien and Mafille, Julien and Boucinha, Lilia and Djebali, Sophia and Gandrillon, Olivier and Marvel, Jacqueline and Arpin, Christophe}
}

@article{Fokas2020,
  title = {A quantitative framework for exploring exit strategies from the {COVID}-19 lockdown},
  volume = {140},
  url = {http://dx.doi.org/10.1016/j.chaos.2020.110244},
  journal = {Chaos,  Solitons \& Fractals},
  author = {Fokas,  A.S. and Cuevas-Maraver,  J. and Kevrekidis,  P.G.},
  year = {2020},
  pages = {110244}
}

@article{GOODWIN1965425,
title = {Oscillatory behavior in enzymatic control processes},
journal = {Advances in Enzyme Regulation},
volume = {3},
pages = {425-437},
year = {1965},
url = {https://doi.org/10.1016/0065-2571(65)90067-1},
author = {Brian C. Goodwin}
}

@article{Saccomani,
title = {Examples of testing global identifiability of biological and biomedical models with the {DAISY} software},
journal = {Computers in Biology and Medicine},
volume = {40},
number = {4},
pages = {402-407},
year = {2010},
url = {https://doi.org/10.1016/j.compbiomed.2010.02.004},
author = {Maria Pia Saccomani and Stefania Audoly and Giuseppina Bellu and Leontina D’Angiò}
}

@article{
Wodarz-HIV,
author = {Dominik Wodarz  and Martin A. Nowak},
title = {Specific therapy regimes could lead to long-term immunological control of {HIV}},
journal = {Proceedings of the National Academy of Sciences},
volume = {96},
number = {25},
pages = {14464-14469},
year = {1999},
URL = {https://www.pnas.org/doi/abs/10.1073/pnas.96.25.14464},
eprint = {https://www.pnas.org/doi/pdf/10.1073/pnas.96.25.14464}
}

@article{Influenza, 
title = {Modeling the {CD}8+ {T} cell immune response to influenza infection in adult and aged mice}, 
journal = {Journal of Theoretical Biology}, 
volume = {593}, 
pages = {111898}, 
year = {2024}, 
url = {https://doi.org/10.1016/j.jtbi.2024.111898}, 
author = {Benjamin Whipple and Tanya A. Miura and Esteban A. Hernandez-Vargas} }

@inproceedings{LLW,
  author={Lecourtier, Yves and Lamnabhi-Lagarrigue, Francoise and Walter, Eric},
  booktitle={26th IEEE Conference on Decision and Control}, 
  title={A method to prove that nonlinear models can be unidentifiable}, 
  year={1987},
  volume={26},
  number={},
  pages={2144-2145},
  url={https://doi.org/10.1109/CDC.1987.272467}
}

@article{ovarian,
  title = {Nonlinear compartmental modeling to monitor ovarian follicle population dynamics on the whole lifespan},
  author = {Ballif, Guillaume and Cl{\'e}ment, Fr{\'e}d{\'e}rique and Yvinec, Romain},
  journal = {{Journal of Mathematical Biology}},
  publisher = {{Springer}},
  volume = {89},
  number = {9},
  pages = {43},
  year = {2024},
  url = {https://doi.org/10.1007/s00285-024-02108-6}
}

@article{pitavastatin,
title = {Structural identifiability analyses of candidate models for in vitro {P}itavastatin hepatic uptake},
journal = {Computer Methods and Programs in Biomedicine},
volume = {114},
number = {3},
pages = {e60-e69},
year = {2014},
url = {https://doi.org/10.1016/j.cmpb.2013.06.013},
author = {Grandjean, Thomas R.B. and Chappell, Michael J. and Yates, James W.T. and Evans, Neil D.}
}

@article{seir2t,
author = {Roosa, Kimberlyn and Chowell, Gerardo},
year = {2019},
month = {01},
title = {Assessing parameter identifiability in compartmental dynamic models using a computational approach: application to infectious disease transmission models},
volume = {16},
journal = {Theoretical Biology and Medical Modelling},
url = {https://doi.org/10.1186/s12976-018-0097-6}
}

@article{Sauer2021,
  title = {Identifiability of Infection Model Parameters Early in an Epidemic},
  volume = {60},
  url = {http://dx.doi.org/10.1137/20M1353289},
  number = {2},
  journal = {SIAM Journal on Control and Optimization},
  author = {Sauer,  Timothy and Berry,  Tyrus and Ebeigbe,  Donald and Norton,  Michael M. and Whalen,  Andrew J. and Schiff,  Steven J.},
year = {2021},
  month = nov,
  pages = {S27–S48}
}

@article{SIR24,
title = {A mathematical model of epidemics with screening and variable infectivity},
journal = {Mathematical and Computer Modelling},
volume = {21},
number = {7},
pages = {29-42},
year = {1995},
url = {https://doi.org/10.1016/0895-7177(95)00029-2},
author = {M.Y. Kim and F.A. Milner}
}

@misc{zheng2020totalvariationregularizationcompartmental,
      title={Total Variation Regularization for Compartmental Epidemic Models with Time-Varying Dynamics}, 
      author={Wenjie Zheng},
      year={2020},
      eprint={2004.00412},
      archivePrefix={arXiv},
      primaryClass={stat.ML},
      url={https://arxiv.org/abs/2004.00412}, 
}

@article{sirc-forced,
author = {Capistran, Marcos and Moreles, Miguel and Lara, Bruno},
year = {2009},
month = {08},
pages = {1890-901},
title = {Parameter Estimation of Some Epidemic Models. {T}he Case of Recurrent Epidemics Caused by Respiratory Syncytial Virus},
volume = {71},
journal = {Bulletin of mathematical biology},
url = {https://doi.org/10.1007/s11538-009-9429-3}
}

@article{siwr,
title = {Model distinguishability and inference robustness in mechanisms of cholera transmission and loss of immunity},
journal = {Journal of Theoretical Biology},
volume = {420},
pages = {68-81},
year = {2017},
url = {https://doi.org/10.1016/j.jtbi.2017.01.032},
author = {Lee, Elizabeth C. and Kelly, Michael R. and Ochocki, Brad M. and Akinwumi, Segun M. and Hamre, Karen E.S. and Tien, Joseph H. and Eisenberg, Marisa C.}
}

@thesis{Transfection-masters,
     author = {Fink, Laura},
    type    = {mastersthesis},
     title = {Model selection in deterministic models of m{RNA} transfection, {M}aster's thesis},
     school = {Ludwig-Maximilians-Universitaet},
     year = {2015}
}

@article{genssi-2-0,
    author = {Ligon, Thomas S and Fröhlich, Fabian and Chiş, Oana T and Banga, Julio R and Balsa-Canto, Eva and Hasenauer, Jan},
    title = {{GenSSI} 2.0: multi-experiment structural identifiability analysis of {SBML} models},
    journal = {Bioinformatics},
    volume = {34},
    number = {8},
    pages = {1421-1423},
    year = {2017},
    month = {11},
    url = {https://doi.org/10.1093/bioinformatics/btx735}
}

@article{cLV,
    author = {Joseph, Tyler A. AND Shenhav, Liat AND Xavier, Joao B. AND Halperin, Eran AND Pe’er, Itsik},
    journal = {PLOS Computational Biology},
    title = {Compositional {L}otka-{V}olterra describes microbial dynamics in the simplex},
    year = {2020},
    month = {05},
    volume = {16},
    url = {https://doi.org/10.1371/journal.pcbi.1007917},
    pages = {1-22},
    number = {5}
}

@INPROCEEDINGS{daisy,
  author={Saccomani, M. P. and Bellu, G.},
  booktitle={2008 16th Mediterranean Conference on Control and Automation}, 
  title={{DAISY}: An efficient tool to test global identifiability. Some case studies}, 
  year={2008},
  pages={1723-1728},
  url={https://doi.org/10.1109/MED.2008.4602152}
}

@article{kd1999,
author =	 {Kumar, Aditya and Daoutidis, Prodomos},
title =	 {Control of Nonlinear Differential Algebraic Equation Systems : An Overview},
year =	 {1998},
journal = {Nonlinear Model Based Process Control},
number =	 {397},
volume = {353},
doi={10.1007/978-94-011-5094-1_11},
url ={https://doi.org/10.1007/978-94-011-5094-1_11}
}

@book{DiStefano2015DynamicSB,
  title={Dynamic Systems Biology Modeling and Simulation},
  author={Joseph J. {DiStefano III}},
  year={2015},
publisher = {Academic Press Inc},
  isbn = {0124104118}
}

@BOOK{Kumar1999-dl,
  title     = "Control of nonlinear differential algebraic equation systems
               with applications to chemical processes",
  author    = "Kumar, Aditya and Daoutidis, Prodromos",
  publisher = "CRC Press",
  series    = "Chapman \& Hall/CRC Research Notes in Mathematics Series",
  month     =  feb,
  year      =  1999,
  address   = "Boca Raton, FL",
  language  = "en"
}

@misc{Sin,
    author = {Sin, Celine},
    title = {Algorithmic Parameter Space Reduction of a Systems Biology Model: A Case Study},
    year = {2012},
    note = {Master thesis},
    url = {https://escholarship.org/uc/item/4r25f875}
}

@article{Castro2020,
  title = {The turning point and end of an expanding epidemic cannot be precisely forecast},
  volume = {117},
  url = {http://dx.doi.org/10.1073/pnas.2007868117},
  number = {42},
  journal = {Proceedings of the National Academy of Sciences},
  publisher = {Proceedings of the National Academy of Sciences},
  author = {Castro,  Mario and Ares,  Saúl and Cuesta,  José A. and Manrubia,  Susanna},
  year = {2020},
  month = oct,
  pages = {26190–26196}
}

\appendix

\newpage
\section{Appendix}
\label{app:simplify_everyone}

{
This section lists our benchmark examples that arise from the problem of structural identifiability of ODE models\footnote{The sources of the ODE models are available at \url{https://github.com/SciML/StructuralIdentifiability.jl/blob/v0.5.18/benchmarking/benchmarks.jl}} (see \Cref{sec:app_ode}). Many of these examples have been collected in \cite{on-the-origins,Hong2019,stident,genssi-2-0,SedoglavicLocal}.
For each example, we report:
\begin{itemize}
    \item information about the original generating set;
    \item the names of the indeterminates;
    \item the original generating set (when it is not too large);
    \item the result of our algorithm, that is, a simplified generating set.
\end{itemize}
We also provide the generating sets in plain text format, together with the code to produce them\footnote{The data in plain text format is available at \url{https://github.com/pogudingleb/RationalFunctionFields.jl/tree/f7d4943aa7432bd10c206b59f63909aca52c909d/paper/data}}. The files that exceed 1 MB in size are not provided, but code in Julia is provided to reproduce the files.

\begin{enumerate}
    
\item \exname{Bilirubin} \cite[Eq. (14)]{combos}.

Original generating set information: 7 indeterminates; 7 non-constant functions; maximal total degrees of numerator and denominator are~$(4, 0)$; 390~bytes total in string representation.

Indeterminates: $k_{01}$, $k_{12}$, $k_{13}$, $k_{14}$, $k_{21}$, $k_{31}$, $k_{41}$.

Original generating set: 

{\footnotesize$\setlength{\extrarowheight}{0pt}\begin{tabular}{l}$-k_{12} k_{13} k_{14}$, $k_{01} k_{12} k_{13} k_{14}$, $-k_{12}-k_{13}-k_{14}$, \\$-k_{12} k_{13}-k_{12} k_{14}-k_{13} k_{14}$, $k_{01}+k_{12}+k_{13}+k_{14}+k_{21}+k_{31}+k_{41}$, \\$k_{01} k_{12} k_{13}+k_{01} k_{12} k_{14}+k_{01} k_{13} k_{14}+k_{12} k_{13} k_{14}+k_{12} k_{13} k_{41}+k_{12} k_{14} k_{31}+k_{13} k_{14} k_{21}$, \\$k_{01} k_{12}+k_{01} k_{13}+k_{01} k_{14}+k_{12} k_{13}+k_{12} k_{14}+k_{12} k_{31}+k_{12} k_{41}+k_{13} k_{14}+k_{13} k_{21}+k_{13} k_{41}+k_{14} k_{21}+k_{14} k_{31}$.\end{tabular}$}

Result of our algorithm: 

{\footnotesize$\setlength{\extrarowheight}{0pt}\begin{tabular}{l}$k_{01}$, $k_{12} k_{13} k_{14}$, $k_{21} k_{31} k_{41}$, $k_{12}+k_{13}+k_{14}$, $k_{21}+k_{31}+k_{41}$, $k_{12} k_{13}+k_{12} k_{14}+k_{13} k_{14}$, \\$k_{21} k_{31}+k_{21} k_{41}+k_{31} k_{41}$, $k_{12} k_{31}+k_{12} k_{41}+k_{13} k_{21}+k_{13} k_{41}+k_{14} k_{21}+k_{14} k_{31}$.\end{tabular}$}

Result is algebraically independent over $\mathbb{C}$: no.

\item \exname{Biohydrogenation} \cite[Eq. (6)-(9)]{MOATE2008731}.

Original generating set information: 6 indeterminates; 37 non-constant functions; maximal total degrees of numerator and denominator are~$(6, 1)$; 1.2~KB total in string representation.

Indeterminates: $k_{5}$, $k_{6}$, $k_{7}$, $k_{8}$, $k_{9}$, $k_{10}$.

Original generating set: too large to be listed.

Result of our algorithm: 

{\footnotesize$\setlength{\extrarowheight}{0pt}\begin{tabular}{l}$k_{5}$, $k_{6}$, $k_{7}$, $k_{9}^{2}$, $k_{9}k_{10}$, $2 k_{8} + k_{10}$.\end{tabular}$}

Result is algebraically independent over $\mathbb{C}$: yes.

\item \exname{Bruno2016} \cite{Bruno2016}.

Original generating set information: 6 indeterminates; 4 non-constant functions; maximal total degrees of numerator and denominator are~$(2, 0)$; 166~bytes total in string representation.

Indeterminates: $k_{\operatorname{zea}}$, $k_{\operatorname{\beta}}$, $k_{\operatorname{cry-OH}}$, $k_{\operatorname{\beta-{10}}}$, $k_{\operatorname{cry-\beta}}$, $k_{\operatorname{OH-\beta-{10}}}$.

Original generating set: 

{\footnotesize$\setlength{\extrarowheight}{0pt}\begin{tabular}{l}$k_{\beta}$, $k_{\operatorname{\beta-{10}}}+k_{\operatorname{cry-OH}} + k_{\operatorname{cry-\beta}}$, $k_{\operatorname{\beta-{10}}}\mspace{2mu} k_{\operatorname{cry-OH}} + k_{\operatorname{\beta-10}}\mspace{2mu} k_{\operatorname{cry-\beta}}$, $k_{\beta}^{2} - k_{\beta}\mspace{2mu} k_{\operatorname{cry-OH}} - k_{\beta}\mspace{2mu} k_{\operatorname{cry-\beta}}$.\end{tabular}$}

Result of our algorithm: 

{\footnotesize$\setlength{\extrarowheight}{0pt}\begin{tabular}{l}$k_{\beta}$, $k_{\operatorname{\beta-{10}}}$, $k_{\operatorname{cry-OH}}+k_{\operatorname{cry-\beta}}$.\end{tabular}$}

Result is algebraically independent over $\mathbb{C}$: yes.

%







\item \exname{CGV1990} \cite{Chappell1990Global,SedoglavicLocal}.

Original generating set information: 9 indeterminates; 412 non-constant functions; maximal total degrees of numerator and denominator are~$(18, 6)$; 90.0~KB total in string representation.

Indeterminates: $R$, $S$, $V_{3}$, $k_{3}$, $k_{4}$, $k_{5}$, $k_{6}$, $k_{7}$, $V_{36}$.

Original generating set: too large to be listed.

Result of our algorithm: 

{\footnotesize$\setlength{\extrarowheight}{8pt}\begin{tabular}{l}$k_{3}$, $k_{4}$, $k_{6}$, $k_{7}$, $R V_{3}+S V_{36}$, $R V_{36}+\dfrac{1}{25} S V_{3}$, $\dfrac{V_{3} k_{5}+5 V_{36} k_{5}}{V_{3}}$, $R V_{36} k_{5}+\dfrac{1}{5} S V_{36} k_{5}$.\end{tabular}$}

Result is algebraically independent over $\mathbb{C}$: yes.

\item \exname{CRN} \cite[Eq. (2)]{CRN}.

Original generating set information: 6 indeterminates; 69 non-constant functions; maximal total degrees of numerator and denominator are~$(5, 2)$; 2.0~KB total in string representation.

Indeterminates: $k_{1}$, $k_{2}$, $k_{3}$, $k_{4}$, $k_{5}$, $k_{6}$.

Original generating set: too large to be listed.

Result of our algorithm: 

{\footnotesize$\setlength{\extrarowheight}{0pt}\begin{tabular}{l}$k_{1}$, $k_{2}$, $k_{3}$, $k_{4}$, $k_{5}$, $k_{6}$.\end{tabular}$}

Result is algebraically independent over $\mathbb{C}$: yes.

\item \exname{Covid1} \cite[Eq. (8)-(13)]{Covid1}.

Original generating set information: 10 indeterminates; 302 non-constant functions; maximal total degrees of numerator and denominator are~$(16, 4)$; 214.5~KB total in string representation.

Indeterminates: $f$, $\beta_A$, $\beta_I$, $\epsilon_A$, $\epsilon_S$, $h_{1}$, $h_{2}$, $\gamma_{\operatorname{AI}}$, $\gamma_{\operatorname{IR}}$, $\delta$.

Original generating set: too large to be listed.

Result of our algorithm: 

{\footnotesize$\setlength{\extrarowheight}{0pt}\begin{tabular}{l}$f$, $\beta_A$, $\beta_I$, $\epsilon_A$, $\epsilon_S$, $h_{1}$, $h_{2}$, $\gamma_{\operatorname{AI}}$, $\delta + \gamma_{\operatorname{IR}}$.\end{tabular}$}

Result is algebraically independent over $\mathbb{C}$: yes.

\item \exname{Covid3} \cite[Slide 35]{IntroToSeirSlides}.

Original generating set information: 8 indeterminates; 7 non-constant functions; maximal total degrees of numerator and denominator are~$(4, 2)$; 247~bytes total in string representation.

Indeterminates: $k$, $\Lambda$, $N$, $b$, $e$, $g$, $m$, $r$.

Original generating set: 

{\footnotesize$\setlength{\extrarowheight}{8pt}\begin{tabular}{l}
$\dfrac{b\mspace{2mu} r}{k\mspace{2mu} N}$, $e+g+3\mspace{2mu} m$, $-e-g-2\mspace{2mu} m$, $e\mspace{2mu} m+g\mspace{2mu} m+2\mspace{2mu} m^{2}$, \\$\dfrac{b\mspace{2mu} e\mspace{2mu} r+b\mspace{2mu} g\mspace{2mu} r+2\mspace{2mu} b\mspace{2mu} m\mspace{2mu} r}{k\mspace{2mu} N}$, $\dfrac{b\mspace{2mu} e\mspace{2mu} g\mspace{2mu} r+b\mspace{2mu} e\mspace{2mu} m\mspace{2mu} r+b\mspace{2mu} g\mspace{2mu} m\mspace{2mu} r+b\mspace{2mu} m^{2}\mspace{2mu} r}{k\mspace{2mu} N}$, $\dfrac{-\Lambda\mspace{2mu} b\mspace{2mu} e+N\mspace{2mu} e\mspace{2mu} g\mspace{2mu} m+N\mspace{2mu} e\mspace{2mu} m^{2}+N\mspace{2mu} g\mspace{2mu} m^{2}+N\mspace{2mu} m^{3}}{N}$.\end{tabular}$}

Result of our algorithm: 

{\footnotesize$\setlength{\extrarowheight}{8pt}\begin{tabular}{l}$m$, $e\mspace{2mu} g$, $e+g$, $\dfrac{k\mspace{2mu} \Lambda}{g\mspace{2mu} r}$, $\dfrac{k\mspace{2mu} N}{b\mspace{2mu} r}$.\end{tabular}$}

Result is algebraically independent over $\mathbb{C}$: yes.

\item \exname{Crauste} \cite[Figure 3C]{Crauste}.

Original generating set information: 13 indeterminates; 30 non-constant functions; maximal total degrees of numerator and denominator are~$(4, 2)$; 1.2~KB total in string representation.

Indeterminates: $\mu_{\operatorname{M}}$, $\mu_{\operatorname{N}}$, $\mu_{\operatorname{P}}$, $\mu_{\operatorname{EE}}$, $\mu_{\operatorname{LE}}$, $\mu_{\operatorname{LL}}$, $\mu_{\operatorname{PE}}$, $\mu_{\operatorname{PL}}$, $\rho_{\operatorname{E}}$, $\rho_{\operatorname{P}}$, $\delta_{\operatorname{EL}}$, $\delta_{\operatorname{LM}}$, $\delta_{\operatorname{NE}}$.

Original generating set: too large to be listed.

Result of our algorithm: 

{\footnotesize$\setlength{\extrarowheight}{8pt}\begin{tabular}{l}$\mu_{\operatorname{M}}$, $\mu_{\operatorname{N}}$, $\mu_{\operatorname{P}}$, $\mu_{\operatorname{EE}}$, $\mu_{\operatorname{LE}}$, $\mu_{\operatorname{LL}}$, $\mu_{\operatorname{PE}}$, $\mu_{\operatorname{PL}}$, $\delta_{\operatorname{EL}}$, $\delta_{\operatorname{LM}}$, $\dfrac{\rho_{\operatorname{E}}}{\rho_{\operatorname{P}}}$, $\dfrac{\delta_{\operatorname{NE}}}{\rho_{\operatorname{P}}}$.\end{tabular}$}

Result is algebraically independent over $\mathbb{C}$: yes.

\item \exname{EAIHRD} \cite[Eq. (1)-(6)]{Fokas2020}, equations copied from \cite[Eq. (6a)-(6f)]{chen2024practicalidentifiabilityparameterestimation}.

Original generating set information: 10 indeterminates; 316 non-constant functions; maximal total degrees of numerator and denominator are~$(22, 6)$; 11.0~MB total in string representation.

Indeterminates: $N$, $a$, $d$, $h$, $s$, $c_{1}$, $c_{2}$, $r_{1}$, $r_{2}$, $r_{3}$.

Original generating set: too large to be listed.

Result of our algorithm: 

{\footnotesize$\setlength{\extrarowheight}{8pt}\begin{tabular}{l}$r_{1}$, $d+r_{3}$, $a+h+r_{2}+s$, 
$a\mspace{2mu} h+a\mspace{2mu} r_{2}+h\mspace{2mu} s+r_{2}\mspace{2mu} s$, 
$\dfrac{d\mspace{2mu} h\mspace{2mu} s}{a\mspace{2mu} c_{1}+c_{2}\mspace{2mu} s}$,  $\dfrac{a\mspace{2mu} c_{1}\mspace{2mu} h+a\mspace{2mu} c_{1}\mspace{2mu} r_{2}+c_{2}\mspace{2mu} r_{1}\mspace{2mu} s}{a\mspace{2mu} c_{1}+c_{2}\mspace{2mu} s}$.\end{tabular}$}

Result is algebraically independent over $\mathbb{C}$: yes.

\item \exname{Fujita} \cite{Fujita2010}, \cite[Example Akt pathway]{stident}.

Original generating set information: 16 indeterminates; 217 non-constant functions; maximal total degrees of numerator and denominator are~$(9, 5)$; 24.9~KB total in string representation.

Indeterminates: $a_{1}$, $a_{2}$, $a_{3}$, 
$r_{11}$, $r_{12}$, $r_{21}$, $r_{22}$, $r_{31}$, $r_{41}$, $r_{51}$, $r_{52}$, $r_{61}$, $r_{71}$, $r_{81}$, $r_{91}$.

Original generating set: too large to be listed.

Result of our algorithm: 

{\footnotesize$\setlength{\extrarowheight}{8pt}\begin{tabular}{l}$r_{22}$, $r_{31}$, $r_{41}$, $r_{52}$, $r_{61}$, $r_{71}$, $r_{81}$, $r_{11}-r_{12}-r_{91}$, $\dfrac{a_{1}}{r_{51}}$, $\dfrac{a_{2}}{r_{51}}$, $\dfrac{a_{3}}{r_{51}}$, $\dfrac{r_{21}}{r_{51}}$.\end{tabular}$}

Result is algebraically independent over $\mathbb{C}$: yes.

\item \exname{Goodwin} \cite{GOODWIN1965425}, \cite[Example Goodwin oscillator in Supplementary materials]{Hong2019}.

Original generating set information: 7 indeterminates; 89 non-constant functions; maximal total degrees of numerator and denominator are~$(12, 0)$; 4.6~KB total in string representation.

Indeterminates: $b$, $c$, $\beta$, $\gamma$, $\alpha$, $\delta$, $\sigma$.

Original generating set: too large to be listed.

Result of our algorithm: 

{\footnotesize$\setlength{\extrarowheight}{0pt}\begin{tabular}{l}$b$, $c$, $\sigma$, $\beta \delta$, $\beta+\delta$.\end{tabular}$}

Result is algebraically independent over $\mathbb{C}$: yes.

\item \exname{HighDimNonLin} \cite[Section 6]{Saccomani}.

Original generating set information: 22 indeterminates; 42 non-constant functions; maximal total degrees of numerator and denominator are~$(2, 0)$; 271~bytes total in string representation.

Indeterminates: {\footnotesize$k_m$, $p_{1}$, $p_{2}$, $p_{3}$, $p_{4}$, $p_{5}$, $p_{6}$, $p_{7}$, $p_{8}$, $p_{9}$, $v_m$, $p_{10}$, $p_{11}$, $p_{12}$, $p_{13}$, $p_{14}$, $p_{15}$, $p_{16}$, $p_{17}$, $p_{18}$, $p_{19}$, $p_{20}$}.

Original generating set: 

{\footnotesize$\setlength{\extrarowheight}{0pt}\begin{tabular}{l} $k_m$, $p_{1}$, $p_{2}$, $p_{3}$, $p_{4}$, $p_{5}$, $p_{6}$, $p_{7}$, $p_{8}$, $p_{9}$, $-k_m$, $-p_{1}$, $-p_{2}$, $-p_{3}$, $-p_{4}$, $-p_{5}$, $-p_{6}$, \\$-p_{7}$, $-p_{8}$, $-p_{9}$, $p_{10}$, $p_{11}$, $p_{12}$, $p_{13}$, $p_{14}$, $p_{15}$, $p_{16}$, $p_{17}$, $p_{18}$, $p_{19}$, $p_{20}$, $-p_{10}$, \\$-p_{11}$, $-p_{12}$, $-p_{13}$, $-p_{14}$, $-p_{15}$, $-p_{16}$, $-p_{17}$, $-p_{18}$, $-p_{19}$, $k_m\mspace{2mu} p_{1}+v_m$.\end{tabular}$}

Result of our algorithm: 

{\footnotesize$\setlength{\extrarowheight}{0pt}\begin{tabular}{l}$k_m$, $p_{1}$, $p_{2}$, $p_{3}$, $p_{4}$, $p_{5}$, $p_{6}$, $p_{7}$, $p_{8}$, $p_{9}$, $v_m$, $p_{10}$, $p_{11}$, $p_{12}$, $p_{13}$, $p_{14}$, $p_{15}$, $p_{16}$, $p_{17}$, $p_{18}$, $p_{19}$, $p_{20}$.\end{tabular}$}

Result is algebraically independent over $\mathbb{C}$: yes.

\item \exname{HIV} \cite[Page 1]{Wodarz-HIV}.

Original generating set information: 10 indeterminates; 77 non-constant functions; maximal total degrees of numerator and denominator are~$(8, 0)$; 2.3~KB total in string representation.

Indeterminates: $a$, $b$, $c$, $d$, $h$, $k$, $q$, $u$, $\lambda$, $\beta$.

Original generating set: too large to be listed.

Result of our algorithm: 

{\footnotesize$\setlength{\extrarowheight}{8pt}\begin{tabular}{l}$a$, $b$, $d$, $h$, $u$, $c\mspace{2mu} q^{2}$, $\dfrac{\lambda}{q}$, $\beta\mspace{2mu} k\mspace{2mu} q$.\end{tabular}$}

Result is algebraically independent over $\mathbb{C}$: yes.

\item \exname{HIV2} \cite[Eq. (9)]{daisy}.

Original generating set information: 10 indeterminates; 9 non-constant functions; maximal total degrees of numerator and denominator are~$(4, 0)$; 221~bytes total in string representation.

Indeterminates: $b$, $c$, $d$, $s$, $k_{1}$, $k_{2}$, $q_{1}$, $q_{2}$, $w_{1}$, $w_{2}$.

Original generating set: 

{\footnotesize$\setlength{\extrarowheight}{0pt}\begin{tabular}{l}$b$, $d$, $-s$, $-b\mspace{2mu} k_{2}\mspace{2mu} q_{2}$, $b^{2}\mspace{2mu} k_{2}\mspace{2mu} q_{2}$, $c+k_{1}+w_{1}+w_{2}$, $-b\mspace{2mu} k_{2}\mspace{2mu} q_{2}\mspace{2mu} s+c\mspace{2mu} k_{1}\mspace{2mu} w_{2}+c\mspace{2mu} w_{1}\mspace{2mu} w_{2}$,\\ $c\mspace{2mu} k_{1}+c\mspace{2mu} w_{1}+c\mspace{2mu} w_{2}+k_{1}\mspace{2mu} w_{2}+w_{1}\mspace{2mu} w_{2}$, $b\mspace{2mu} d\mspace{2mu} k_{2}\mspace{2mu} q_{2}-b\mspace{2mu} k_{1}\mspace{2mu} k_{2}\mspace{2mu} q_{1}-b\mspace{2mu} k_{1}\mspace{2mu} k_{2}\mspace{2mu} q_{2}-b\mspace{2mu} k_{2}\mspace{2mu} q_{2}\mspace{2mu} w_{1}$.\end{tabular}$}

Result of our algorithm: 

{\footnotesize$\setlength{\extrarowheight}{0pt}\begin{tabular}{l}$b$, $d$, $s$, $k_{2}\mspace{2mu} q_{2}$, $c+k_{1}+w_{1}+w_{2}$, $c\mspace{2mu} k_{1}\mspace{2mu} w_{2}+c\mspace{2mu} w_{1}\mspace{2mu} w_{2}$, \\$k_{1}\mspace{2mu} k_{2}\mspace{2mu} q_{1}+k_{1}\mspace{2mu} k_{2}\mspace{2mu} q_{2}+k_{2}\mspace{2mu} q_{2}\mspace{2mu} w_{1}$, $c\mspace{2mu} k_{1}+c\mspace{2mu} w_{1}+c\mspace{2mu} w_{2}+k_{1}\mspace{2mu} w_{2}+w_{1}\mspace{2mu} w_{2}$.\end{tabular}$}

Result is algebraically independent over $\mathbb{C}$: yes.

\item \exname{Influenza\_MB1} \cite[Figure 3 and Appendix A]{Influenza}.

Original generating set information: 7 indeterminates; 27 non-constant functions; maximal total degrees of numerator and denominator are~$(4, 0)$; 403~bytes total in string representation.

Indeterminates: $c$, $p$, $r$, $d_{I}$, $k_{T}$, $s_{T}$, $\beta$.

Original generating set: 

{\footnotesize$\setlength{\extrarowheight}{8pt}\begin{tabular}{l}$c$, $d_{I}$, $k_{T}$, $k_{T}$, $-d_{I}$, $-k_{T}$, $-s_{T}$, $-\beta$, $c\mspace{2mu} k_{T}$, $-c\mspace{2mu} k_{T}$, $-\beta\mspace{2mu} c$, $d_{I}\mspace{2mu} k_{T}$, $-d_{I}\mspace{2mu} k_{T}$, $-k_{T}\mspace{2mu} s_{T}$, $-\beta\mspace{2mu} d_{I}$, $-c\mspace{2mu} d_{I}\mspace{2mu} s_{T}$,\\ $\dfrac{11}{1000}\mspace{2mu} k_{T}$, $-d_{I}\mspace{2mu} k_{T}\mspace{2mu} s_{T}$, $-r+\dfrac{11}{1000}$, $-\beta\mspace{2mu} k_{T}-c$, $-c\mspace{2mu} d_{I}\mspace{2mu} k_{T}\mspace{2mu} s_{T}$, $\dfrac{11}{1000}\mspace{2mu} d_{I}\mspace{2mu} k_{T}$, $\dfrac{11}{1000}\mspace{2mu} c\mspace{2mu} d_{I}\mspace{2mu} k_{T}$,\\ $-\beta\mspace{2mu} c\mspace{2mu} k_{T}-d_{I}\mspace{2mu} s_{T}$, $-\beta\mspace{2mu} d_{I}\mspace{2mu} k_{T}-d_{I}\mspace{2mu} r+\dfrac{11}{1000}\mspace{2mu} d_{I}$, $-\beta\mspace{2mu} c\mspace{2mu} d_{I}\mspace{2mu} k_{T}-c\mspace{2mu} d_{I}\mspace{2mu} r+\dfrac{11}{1000}\mspace{2mu} c\mspace{2mu} d_{I}$, $-\beta\mspace{2mu} c\mspace{2mu} d_{I}$.\end{tabular}$}

Result of our algorithm: 

{\footnotesize$\setlength{\extrarowheight}{0pt}\begin{tabular}{l}$c$, $r$, $d_{I}$, $k_{T}$, $s_{T}$, $\beta$.\end{tabular}$}

Result is algebraically independent over $\mathbb{C}$: yes.

\item \exname{Influenza\_MB2} \cite[Figure 3 and Appendix A]{Influenza}.

Original generating set information: 8 indeterminates; 50 non-constant functions; maximal total degrees of numerator and denominator are~$(4, 0)$; 817~bytes total in string representation.

Indeterminates: $c$, $p$, $r$, $c_{T}$, $d_{I}$, $k_{T}$, $s_{T}$, $\beta$.

Original generating set: too large to be listed.

Result of our algorithm: 

{\footnotesize$\setlength{\extrarowheight}{0pt}\begin{tabular}{l}$c$, $r$, $c_{T}$, $d_{I}$, $k_{T}$, $s_{T}$, $\beta$.\end{tabular}$}

Result is algebraically independent over $\mathbb{C}$: yes.

\item \exname{Influenza\_MB3} \cite[Figure 3 and Appendix A]{Influenza}.

Original generating set information: 8 indeterminates; 27 non-constant functions; maximal total degrees of numerator and denominator are~$(4, 0)$; 372~bytes total in string representation.

Indeterminates: $c$, $p$, $r$, $d_{I}$, $d_{T}$, $k_{T}$, $s_{T}$, $\beta$.

Original generating set: 

{\footnotesize$\setlength{\extrarowheight}{0pt}\begin{tabular}{l}$c$, $d_{I}$, $k_{T}$, $k_{T}$, $-d_{I}$, $-k_{T}$, $-s_{T}$, $-\beta$, $c\mspace{2mu} k_{T}$, $-c\mspace{2mu} k_{T}$, $-\beta\mspace{2mu} c$, $d_{I}\mspace{2mu} k_{T}$, $d_{T}-r$, $d_{T}\mspace{2mu} k_{T}$, $-d_{I}\mspace{2mu} k_{T}$, \\$-k_{T}\mspace{2mu} s_{T}$, $-\beta\mspace{2mu} d_{I}$, $-c\mspace{2mu} d_{I}\mspace{2mu} s_{T}$, $-\beta\mspace{2mu} c\mspace{2mu} d_{I}$, $d_{I}\mspace{2mu} d_{T}\mspace{2mu} k_{T}$, $-d_{I}\mspace{2mu} k_{T}\mspace{2mu} s_{T}$, $-\beta\mspace{2mu} k_{T}-c$, $c\mspace{2mu} d_{I}\mspace{2mu} d_{T}\mspace{2mu} k_{T}$, $-c\mspace{2mu} d_{I}\mspace{2mu} k_{T}\mspace{2mu} s_{T}$, \\$-\beta\mspace{2mu} c\mspace{2mu} k_{T}-d_{I}\mspace{2mu} s_{T}$, $-\beta\mspace{2mu} d_{I}\mspace{2mu} k_{T}+d_{I}\mspace{2mu} d_{T}-d_{I}\mspace{2mu} r$, $-\beta\mspace{2mu} c\mspace{2mu} d_{I}\mspace{2mu} k_{T}+c\mspace{2mu} d_{I}\mspace{2mu} d_{T}-c\mspace{2mu} d_{I}\mspace{2mu} r$.\end{tabular}$}

Result of our algorithm: 

{\footnotesize$\setlength{\extrarowheight}{0pt}\begin{tabular}{l}$c$, $r$, $d_{I}$, $d_{T}$, $k_{T}$, $s_{T}$, $\beta$.\end{tabular}$}

Result is algebraically independent over $\mathbb{C}$: yes.

\item \exname{Influenza\_MB4} \cite[Figure 3 and Appendix A]{Influenza}.

Original generating set information: 9 indeterminates; 50 non-constant functions; maximal total degrees of numerator and denominator are~$(4, 0)$; 755~bytes total in string representation.

Indeterminates: $c$, $p$, $r$, $c_{T}$, $d_{I}$, $d_{T}$, $k_{T}$, $s_{T}$, $\beta$.

Original generating set: too large to be listed.

Result of our algorithm: 

{\footnotesize$\setlength{\extrarowheight}{0pt}\begin{tabular}{l}$c$, $r$, $c_{T}$, $d_{I}$, $d_{T}$, $k_{T}$, $s_{T}$, $\beta$.\end{tabular}$}

Result is algebraically independent over $\mathbb{C}$: yes.

\item \exname{Influenza\_MD1} \cite[Figure 3 and Appendix A]{Influenza}.

Original generating set information: 6 indeterminates; 14 non-constant functions; maximal total degrees of numerator and denominator are~$(3, 0)$; 206~bytes total in string representation.

Indeterminates: $c$, $p$, $r$, $d_{I}$, $s_{T}$, $\beta$.

Original generating set: 

{\footnotesize$\setlength{\extrarowheight}{8pt}\begin{tabular}{l}$c$, $-c$, $d_{I}$, $-d_{I}$, $\beta$, $\beta\mspace{2mu} c$, $d_{I}\mspace{2mu} s_{T}$, $c\mspace{2mu} d_{I}\mspace{2mu} s_{T}$, $\dfrac{-1000}{11}\mspace{2mu} r$, $\dfrac{-1000}{11}\mspace{2mu} s_{T}$, $\dfrac{-11}{1000}\mspace{2mu} d_{I}$, $\dfrac{-11}{1000}\mspace{2mu} c\mspace{2mu} d_{I}$, $\beta\mspace{2mu} d_{I}+d_{I}\mspace{2mu} r$, $\beta\mspace{2mu} c\mspace{2mu} d_{I}+c\mspace{2mu} d_{I}\mspace{2mu} r$.\end{tabular}$}

Result of our algorithm: 

{\footnotesize$\setlength{\extrarowheight}{0pt}\begin{tabular}{l}$c$, $r$, $d_{I}$, $s_{T}$, $\beta$.\end{tabular}$}

Result is algebraically independent over $\mathbb{C}$: yes.

\item \exname{Influenza\_MD2} \cite[Figure 3 and Appendix A]{Influenza}.

Original generating set information: 7 indeterminates; 28 non-constant functions; maximal total degrees of numerator and denominator are~$(3, 0)$; 396~bytes total in string representation.

Indeterminates: $c$, $p$, $r$, $c_{T}$, $d_{I}$, $s_{T}$, $\beta$.

Original generating set: 

{\footnotesize$\setlength{\extrarowheight}{8pt}\begin{tabular}{l}$c$, $c$, $-c$, $-c$, $d_{I}$, $d_{I}$, $-d_{I}$, $-d_{I}$, $\beta$, $\beta$, $\beta\mspace{2mu} c$, $\beta\mspace{2mu} c$, $\dfrac{1000}{11}$, $\dfrac{1000}{11}$, $d_{I}\mspace{2mu} s_{T}$, $c\mspace{2mu} d_{I}\mspace{2mu} s_{T}$, $\dfrac{-1000}{11}\mspace{2mu} r$, $\dfrac{-1000}{11}\mspace{2mu} r$, \\$\dfrac{-11}{1000}\mspace{2mu} d_{I}$, $\dfrac{-11}{1000}\mspace{2mu} d_{I}$, $\dfrac{-11}{1000}\mspace{2mu} c\mspace{2mu} d_{I}$, $\dfrac{-11}{1000}\mspace{2mu} c\mspace{2mu} d_{I}$, $\beta\mspace{2mu} d_{I}+d_{I}\mspace{2mu} r$, $\beta\mspace{2mu} d_{I}+d_{I}\mspace{2mu} r$, $-c_{T}\mspace{2mu} d_{I}+d_{I}\mspace{2mu} s_{T}$, \\$\beta\mspace{2mu} c\mspace{2mu} d_{I}+c\mspace{2mu} d_{I}\mspace{2mu} r$, $-c\mspace{2mu} c_{T}\mspace{2mu} d_{I}+c\mspace{2mu} d_{I}\mspace{2mu} s_{T}$, $\dfrac{1000}{11}\mspace{2mu} c_{T}-\dfrac{1000}{11}\mspace{2mu} s_{T}$, $\dfrac{-1000}{11}\mspace{2mu} s_{T}$, $\beta\mspace{2mu} c\mspace{2mu} d_{I}+c\mspace{2mu} d_{I}\mspace{2mu} r$.\end{tabular}$}

Result of our algorithm: 

{\footnotesize$\setlength{\extrarowheight}{0pt}\begin{tabular}{l}$c$, $r$, $c_{T}$, $d_{I}$, $s_{T}$, $\beta$.\end{tabular}$}

Result is algebraically independent over $\mathbb{C}$: yes.

\item \exname{Influenza\_MD3} \cite[Figure 3 and Appendix A]{Influenza}.

Original generating set information: 7 indeterminates; 15 non-constant functions; maximal total degrees of numerator and denominator are~$(3, 0)$; 183~bytes total in string representation.

Indeterminates: $c$, $p$, $r$, $d_{I}$, $d_{T}$, $s_{T}$, $\beta$.

Original generating set: 

{\footnotesize$\setlength{\extrarowheight}{0pt}\begin{tabular}{l}$c$, $-c$, $-r$, $d_{I}$, $d_{T}$, $-d_{I}$, $-s_{T}$, $\beta$, $\beta\mspace{2mu} c$, $d_{I}\mspace{2mu} s_{T}$, $-d_{I}\mspace{2mu} d_{T}$, $c\mspace{2mu} d_{I}\mspace{2mu} s_{T}$, $-c\mspace{2mu} d_{I}\mspace{2mu} d_{T}$, $\beta\mspace{2mu} d_{I}+d_{I}\mspace{2mu} r$, $\beta\mspace{2mu} c\mspace{2mu} d_{I}+c\mspace{2mu} d_{I}\mspace{2mu} r$.\end{tabular}$}

Result of our algorithm: 

{\footnotesize$\setlength{\extrarowheight}{0pt}\begin{tabular}{l}$c$, $r$, $d_{I}$, $d_{T}$, $s_{T}$, $\beta$.\end{tabular}$}

Result is algebraically independent over $\mathbb{C}$: yes.

\item \exname{Influenza\_MD4} \cite[Figure 3 and Appendix A]{Influenza}.

Original generating set information: 8 indeterminates; 30 non-constant functions; maximal total degrees of numerator and denominator are~$(3, 0)$; 341~bytes total in string representation.

Indeterminates: $c$, $p$, $r$, $c_{T}$, $d_{I}$, $d_{T}$, $s_{T}$, $\beta$.

Original generating set: 

{\footnotesize$\setlength{\extrarowheight}{0pt}\begin{tabular}{l}$c$, $c$, $-1$, $-1$, $-c$, $-c$, $-r$, $-r$, $d_{I}$, $d_{I}$, $d_{T}$, $d_{T}$, $-d_{I}$, $-d_{I}$, $-s_{T}$, $\beta$, $\beta$, $\beta\mspace{2mu} c$, \\$\beta\mspace{2mu} c$, $d_{I}\mspace{2mu} s_{T}$, $-d_{I}\mspace{2mu} d_{T}$, $-d_{I}\mspace{2mu} d_{T}$, $c\mspace{2mu} d_{I}\mspace{2mu} s_{T}$, $c_{T}-s_{T}$, $-c\mspace{2mu} d_{I}\mspace{2mu} d_{T}$, $-c\mspace{2mu} d_{I}\mspace{2mu} d_{T}$, $\beta\mspace{2mu} d_{I}+d_{I}\mspace{2mu} r$, \\$\beta\mspace{2mu} d_{I}+d_{I}\mspace{2mu} r$, $-c_{T}\mspace{2mu} d_{I}+d_{I}\mspace{2mu} s_{T}$, $\beta\mspace{2mu} c\mspace{2mu} d_{I}+c\mspace{2mu} d_{I}\mspace{2mu} r$, $\beta\mspace{2mu} c\mspace{2mu} d_{I}+c\mspace{2mu} d_{I}\mspace{2mu} r$, $-c\mspace{2mu} c_{T}\mspace{2mu} d_{I}+c\mspace{2mu} d_{I}\mspace{2mu} s_{T}$.\end{tabular}$}

Result of our algorithm: 

{\footnotesize$\setlength{\extrarowheight}{0pt}\begin{tabular}{l}$c$, $r$, $c_{T}$, $d_{I}$, $d_{T}$, $s_{T}$, $\beta$.\end{tabular}$}

Result is algebraically independent over $\mathbb{C}$: yes.

\item \exname{JAK-STAT} \cite[Model L1236 in Supplementary materials]{Raia2011} with modifications introduced in~\cite{ReyBarreiro2023}.


Original generating set information: 22 indeterminates; 1300 non-constant functions; maximal total degrees of numerator and denominator are~$(14, 6)$; 141.8~KB total in string representation.

Indeterminates: $t_{1}$, $t_{2}$, $t_{3}$, $t_{4}$, $t_{5}$, $t_{6}$, $t_{7}$, $t_{8}$, $t_{9}$, $t_{10}$, $t_{11}$, $t_{12}$, $t_{13}$, $t_{14}$, $t_{15}$, $t_{16}$, $t_{17}$, $t_{18}$, $t_{19}$, $t_{20}$, $t_{21}$, $t_{22}$.

Original generating set: too large to be listed.

Result of our algorithm: 

{\footnotesize$\setlength{\extrarowheight}{0pt}\begin{tabular}{l}$t_{1}$, $t_{2}$, $t_{3}$, $t_{4}$, $t_{5}$, $t_{6}$, $t_{7}$, $t_{8}$, $t_{9}$, $t_{10}$, $t_{12}$, $t_{13}$, $t_{14}$, $t_{16}$, $t_{18}$, $t_{19}$, $t_{20}$, $t_{11} t_{21}$, $t_{15} t_{21}$, $t_{17} t_{22}$.\end{tabular}$}

Result is algebraically independent over $\mathbb{C}$: yes.

\item \exname{KD1999} \cite{SedoglavicLocal}, \cite[Eq. (7.1)]{Kumar1999-dl}.


Original generating set information: 14 indeterminates; 129 non-constant functions; maximal total degrees of numerator and denominator are~$(9, 3)$; 3.8~KB total in string representation.

Indeterminates: $E$, $R$, $V$, $\Delta$, $T_a$, $T_h$, $U$, $V_h$, $c_p$, $k_{0}$, $\rho$, $C{a_{0}}$, $c_{ph}$, $\rho_h$.

Original generating set: too large to be listed.

Result of our algorithm: 

{\footnotesize$\setlength{\extrarowheight}{8pt}\begin{tabular}{l}$V$, $T_a$, $T_h$, $V_h$, $C_{a_{0}}$, $\dfrac{E}{R}$, $\dfrac{\Delta}{U}$, $\dfrac{c_p \rho}{U}$, $\dfrac{c_{ph} \rho_h}{U}$.\end{tabular}$}

Result is algebraically independent over $\mathbb{C}$: yes.

\item \exname{LLW} \cite[Example 1]{LLW}.

Original generating set information: 4 indeterminates; 9 non-constant functions; maximal total degrees of numerator and denominator are~$(3, 0)$; 248~bytes total in string representation.

Indeterminates: $p_{1}$, $p_{2}$, $p_{3}$, $p_{4}$.

Original generating set: 

{\footnotesize$\setlength{\extrarowheight}{8pt}\begin{tabular}{l}$-1$, $\dfrac{1}{2}$, $\dfrac{-1}{2}$, $p_{2} p_{4}$, $-p_{1}-p_{3}$, $-p_{1}-p_{3}$, $\dfrac{1}{2} p_{1}+\dfrac{1}{2} p_{3}$, $\dfrac{3}{2} p_{1}+\dfrac{3}{2} p_{3}$, $\dfrac{-1}{2} p_{1}^{2} p_{3}-\dfrac{1}{2} p_{1} p_{3}^{2}$, \\$\dfrac{1}{2} p_{1} p_{2} p_{4}+\dfrac{1}{2} p_{2} p_{3} p_{4}$, $\dfrac{1}{2} p_{1}^{2}+p_{1} p_{3}+\dfrac{1}{2} p_{3}^{2}$, $\dfrac{-1}{2} p_{1}^{2}-\dfrac{3}{2} p_{1} p_{3}-\dfrac{1}{2} p_{3}^{2}$.\end{tabular}$}

Result of our algorithm: 

{\footnotesize$\setlength{\extrarowheight}{0pt}\begin{tabular}{l}$p_{1} p_{3}$, $p_{2} p_{4}$, $p_{1}+p_{3}$.\end{tabular}$}

Result is algebraically independent over $\mathbb{C}$: yes.

\item \exname{Lotka-Volterra}.

Original generating set information: 4 indeterminates; 4 non-constant functions; maximal total degrees of numerator and denominator are~$(3, 0)$; 90~bytes total in string representation.

Indeterminates: $a$, $b$, $c$, $d$.

Original generating set: 

{\footnotesize$\setlength{\extrarowheight}{0pt}\begin{tabular}{l}$d$, $-1$, $-a b$, $-a d-b d$, $a^{2} b+a b^{2}$.\end{tabular}$}

Result of our algorithm: 

{\footnotesize$\setlength{\extrarowheight}{0pt}\begin{tabular}{l}$d$, $a b$, $a+b$.\end{tabular}$}

Result is algebraically independent over $\mathbb{C}$: yes.

\item \exname{Lincomp1} \cite[Appendix A]{ahmed2025identifiabilitydirectedcyclecatenarylinear}.

Original generating set information: 7 indeterminates; 6 non-constant functions; maximal total degrees of numerator and denominator are~$(5, 0)$; 1.4~KB total in string representation.

Indeterminates: $a_{02}$, $a_{03}$, $a_{15}$, $a_{21}$, $a_{32}$, $a_{43}$, $a_{54}$.

Original generating set: too large to be listed.

Result of our algorithm: 

{\footnotesize$\setlength{\extrarowheight}{0pt}\begin{tabular}{l}$a_{21} a_{32} a_{43} a_{54}$, $a_{21}+a_{32}+a_{43}+a_{54}+a_{15}+a_{02}+a_{03}$, \\$a_{21} a_{32} a_{54} a_{15} a_{03}+a_{21} a_{43} a_{54} a_{15} a_{02}+a_{21} a_{54} a_{15} a_{02} a_{03}$, \\$a_{21} a_{32}+a_{21} a_{43}+a_{21} a_{54}+a_{21} a_{15}+a_{21} a_{02} + \text{14 more terms}$, \\$a_{21} a_{32} a_{43} a_{15}+a_{21} a_{32} a_{54} a_{15}+a_{21} a_{32} a_{54} a_{03}+a_{21} a_{32} a_{15} a_{03}+a_{21} a_{43} a_{54} a_{15} + \text{10 more terms}$, \\$a_{21} a_{32} a_{43}+a_{21} a_{32} a_{54}+a_{21} a_{32} a_{15}+a_{21} a_{32} a_{03}+a_{21} a_{43} a_{54} + \text{20 more terms}$.\end{tabular}$}

Result is algebraically independent over $\mathbb{C}$: yes.

\item \exname{Lincomp2} \cite[Appendix A]{ahmed2025identifiabilitydirectedcyclecatenarylinear}.

Original generating set information: 7 indeterminates; 6 non-constant functions; maximal total degrees of numerator and denominator are~$(5, 0)$; 1.4~KB total in string representation.

Indeterminates: $a_{02}$, $a_{04}$, $a_{15}$, $a_{21}$, $a_{32}$, $a_{43}$, $a_{54}$.

Original generating set: too large to be listed.

Result of our algorithm: 

{\footnotesize$\setlength{\extrarowheight}{0pt}\begin{tabular}{l}$a_{21} a_{32} a_{43} a_{54}$, $a_{21}+a_{32}+a_{43}+a_{54}+a_{15}+a_{02}+a_{04}$, \\$a_{21} a_{32} a_{43} a_{15} a_{04}+a_{21} a_{43} a_{54} a_{15} a_{02}+a_{21} a_{43} a_{15} a_{02} a_{04}$, \\$a_{21}^{2}+a_{32}^{2}+2 a_{32} a_{02}+a_{43}^{2}+a_{54}^{2}+2 a_{54} a_{04}+a_{15}^{2}+a_{02}^{2}+a_{04}^{2}$, \\$a_{21} a_{32} a_{43} a_{15}+a_{21} a_{32} a_{43} a_{04}+a_{21} a_{32} a_{54} a_{15}+a_{21} a_{32} a_{15} a_{04}+a_{21} a_{43} a_{54} a_{15} + \text{10 more terms}$, \\$a_{21} a_{32} a_{43}+a_{21} a_{32} a_{54}+a_{21} a_{32} a_{15}+a_{21} a_{32} a_{04}+a_{21} a_{43} a_{54} + \text{20 more terms}$.\end{tabular}$}

Result is algebraically independent over $\mathbb{C}$: yes.

\item \exname{Lipolysis} \cite[Eq. (1)-(5)]{MOATE2008731}.

Original generating set information: 3 indeterminates; 12 non-constant functions; maximal total degrees of numerator and denominator are~$(3, 0)$; 155~bytes total in string representation.

Indeterminates: $k_{2}$, $k_{3}$, $k_{4}$.

Original generating set: 

{\footnotesize$\setlength{\extrarowheight}{8pt}\begin{tabular}{l}$k_{2}$, $k_{2}$, $k_{3}$, $-k_{2}$, $2 k_{4}$, $-k_{4}^{2}$, $k_{2} k_{3}$, $\dfrac{1}{3} k_{2}$, $2 k_{2} k_{4}$, $-k_{2} k_{4}^{2}$, $\dfrac{-1}{3} k_{3}+k_{4}$, $\dfrac{-1}{3} k_{2} k_{3}+k_{2} k_{4}$.\end{tabular}$}

Result of our algorithm: 

{\footnotesize$\setlength{\extrarowheight}{0pt}\begin{tabular}{l}$k_{2}$, $k_{3}$, $k_{4}$.\end{tabular}$}

Result is algebraically independent over $\mathbb{C}$: yes.

\item \exname{MAPK-5} \cite{MANRAI20085533}, \cite[Example MAPK pathway and Eq. (6.7) in Appendix]{stident}.

Original generating set information: 22 indeterminates; 510 non-constant functions; maximal total degrees of numerator and denominator are~$(8, 6)$; 633.8~KB total in string representation.

Indeterminates: $a_{00}$, $a_{01}$, $a_{10}$, $b_{00}$, $b_{01}$, $b_{10}$, $c_{0001}$, $c_{0010}$, $c_{0011}$, $c_{0111}$, $c_{1011}$, $\beta_{01}$, $\beta_{10}$, $\beta_{11}$, $\alpha_{01}$, $\alpha_{10}$, $\alpha_{11}$, $\gamma_{0100}$, $\gamma_{1000}$, $\gamma_{1100}$, $\gamma_{1101}$, $\gamma_{1110}$.

Original generating set: too large to be listed.

Result of our algorithm: 

{\footnotesize$\setlength{\extrarowheight}{0pt}\begin{tabular}{l}$a_{00}$, $a_{01}$, $a_{10}$, $b_{00}$, $b_{01}$, $b_{10}$, $c_{0001}$, $c_{0010}$, $c_{0011}$, $c_{0111}$, $c_{1011}$, $\beta_{01}$, \\$\beta_{10}$, $\beta_{11}$, $\alpha_{01}$, $\alpha_{10}$, $\alpha_{11}$, $\gamma_{0100}$, $\gamma_{1000}$, $\gamma_{1100}$, $\gamma_{1101}$, $\gamma_{1110}$.\end{tabular}$}

Result is algebraically independent over $\mathbb{C}$: yes.

\item \exname{MAPK-6} \cite{MANRAI20085533}, \cite[Example MAPK pathway and Eq. (6.6) in Appendix]{stident}.

Original generating set information: 22 indeterminates; 97 non-constant functions; maximal total degrees of numerator and denominator are~$(5, 3)$; 17.8~KB total in string representation.

Indeterminates: $a_{00}$, $a_{01}$, $a_{10}$, $b_{00}$, $b_{01}$, $b_{10}$, $c_{0001}$, $c_{0010}$, $c_{0011}$, $c_{0111}$, $c_{1011}$, $\beta_{01}$, $\beta_{10}$, $\beta_{11}$, $\alpha_{01}$, $\alpha_{10}$, $\alpha_{11}$, $\gamma_{0100}$, $\gamma_{1000}$, $\gamma_{1100}$, $\gamma_{1101}$, $\gamma_{1110}$.

Original generating set: too large to be listed.

Result of our algorithm: 

{\footnotesize$\setlength{\extrarowheight}{0pt}\begin{tabular}{l}$a_{00}$, $a_{01}$, $a_{10}$, $b_{00}$, $b_{01}$, $b_{10}$, $c_{0001}$, $c_{0010}$, $c_{0011}$, $c_{0111}$, $c_{1011}$, $\beta_{01}$, \\$\beta_{10}$, $\beta_{11}$, $\alpha_{01}$, $\alpha_{10}$, $\alpha_{11}$, $\gamma_{0100}$, $\gamma_{1000}$, $\gamma_{1100}$, $\gamma_{1101}$, $\gamma_{1110}$.\end{tabular}$}

Result is algebraically independent over $\mathbb{C}$: yes.

\item \exname{Ovarian\_follicle} \cite[Section 3.2]{ovarian}.

Original generating set information: 22 indeterminates; 1684 non-constant functions; maximal total degrees of numerator and denominator are~$(21, 10)$; 44.2~MB total in string representation.

Indeterminates: {\footnotesize$n$, $D_{0}$, $D_{1}$, $D_{2}$, $D_{3}$, $D_{4}$, $E_{0}$, $E_{1}$, $E_{2}$, $E_{3}$, $E_{4}$, $J_{1}$, $J_{2}$, $J_{3}$, $f_{1}$, $m_{1}$, $m_{2}$, $m_{3}$, $n_{1}$, $n_{2}$, $n_{3}$, $\tau$}.

Original generating set: too large to be listed.

Result of our algorithm: 

{\footnotesize$\setlength{\extrarowheight}{0pt}\begin{tabular}{l}$D_{2}$, $D_{3}$, $D_{4}$, $E_{0}$, $E_{1}$, $E_{2}$, $E_{3}$, $E_{4}$, $J_{1}$, $J_{2}$, $J_{3}$, $m_{1}$, $m_{2}$, $m_{3}$, $n_{2}$, $n_{3}$, $\tau$, $n^{2}$, $f_{1} n$, $D_{0} D_{1}$, $D_{0}+D_{1}$, $f_{1}+2 n_{1}$.\end{tabular}$}

Result is algebraically independent over $\mathbb{C}$: yes.

\item \exname{Pharm} \cite{Demignot1987Pharm}, \cite[Example 6.4]{HOPY}.

Original generating set information: 7 indeterminates; 3533 non-constant functions; maximal total degrees of numerator and denominator are~$(26, 5)$; 35.4~MB total in string representation.

Indeterminates: $n$, $a_{1}$, $a_{2}$, $b_{1}$, $b_{2}$, $k_a$, $k_c$.

Original generating set: too large to be listed.

Result of our algorithm: 

{\footnotesize$\setlength{\extrarowheight}{0pt}\begin{tabular}{l}$n$, $a_{1}$, $a_{2}$, $b_{1}$, $b_{2}$, $k_a$, $k_c$.\end{tabular}$}

Result is algebraically independent over $\mathbb{C}$: yes.

\item \exname{Pitavastatin} \cite[Eq. (22)-(24)]{pitavastatin}.

Original generating set information: 8 indeterminates; 249 non-constant functions; maximal total degrees of numerator and denominator are~$(29, 5)$; 2.1~MB total in string representation.

Indeterminates: $k$, $T_{0}$, $k_{1}$, $k_{2}$, $k_{3}$, $k_{4}$, $r_{1}$, $r_{3}$.

Original generating set: too large to be listed.

Result of our algorithm: 

{\footnotesize$\setlength{\extrarowheight}{0pt}\begin{tabular}{l}$k_{2}$, $k_{3}$, $k_{4}$, $r_{1}$, $r_{3}$, $T_{0} k$, $T_{0} k_{1}$.\end{tabular}$}

Result is algebraically independent over $\mathbb{C}$: yes.

\item \exname{SEAIJRC} \cite[Example SEAIJRC in Appendix]{stident}.

Original generating set information: 7 indeterminates; 1456 non-constant functions; maximal total degrees of numerator and denominator are~$(31, 8)$; 26.1~MB total in string representation.

Indeterminates: $b$, $k$, $q$, $r$, $g_{1}$, $g_{2}$, $\alpha$.

Original generating set: too large to be listed.

Result of our algorithm: 

{\footnotesize$\setlength{\extrarowheight}{8pt}\begin{tabular}{l}$b$, $k$, $g_{1}$, $g_{2}$, $\alpha$, $\dfrac{q\mspace{2mu} r-q}{r}$.\end{tabular}$}

Result is algebraically independent over $\mathbb{C}$: yes.

\item \exname{SEIR1} \cite[Eq. (2)]{Zha2020}.

Original generating set information: 4 indeterminates; 17 non-constant functions; maximal total degrees of numerator and denominator are~$(7, 1)$; 1015~bytes total in string representation.

Indeterminates: $v$, $\psi$, $\beta$, $\gamma$.

Original generating set: too large to be listed.

Result of our algorithm: 

{\footnotesize$\setlength{\extrarowheight}{8pt}\begin{tabular}{l}$\gamma$, $\dfrac{\beta}{\psi}$, $\psi v-\psi-v$, $\gamma \psi-\psi v$.\end{tabular}$}

Result is algebraically independent over $\mathbb{C}$: yes.

\item \exname{SEIR2T} \cite[Example Simple SEIR]{seir2t}.

Original generating set information: 4 indeterminates; 37 non-constant functions; maximal total degrees of numerator and denominator are~$(5, 2)$; 743~bytes total in string representation.

Indeterminates: $N$, $a$, $b$, $\nu$.

Original generating set: too large to be listed.

Result of our algorithm: 

{\footnotesize$\setlength{\extrarowheight}{0pt}\begin{tabular}{l}$N$, $a$, $b$, $\nu$.\end{tabular}$}

Result is algebraically independent over $\mathbb{C}$: yes.

\item \exname{SEIR34} \cite[Table 2, ID 34]{MASSONIS2021441}.

Original generating set information: 7 indeterminates; 9 non-constant functions; maximal total degrees of numerator and denominator are~$(4, 2)$; 368~bytes total in string representation.

Indeterminates: $k$, $N$, $r$, $\mu$, $\beta$, $\gamma$, $\epsilon$.

Original generating set: 

{\footnotesize$\setlength{\extrarowheight}{8pt}\begin{tabular}{l}$-1$, $-N$, $\dfrac{\beta r}{k N}$, $\dfrac{-\beta \epsilon r}{N}$, $\epsilon+\gamma+3 \mu$, $-\epsilon-\gamma-2 \mu$, $\epsilon \mu+\gamma \mu+2 \mu^{2}$, $\epsilon \gamma \mu+\epsilon \mu^{2}+\gamma \mu^{2}+\mu^{3}$, \\$\dfrac{\beta \epsilon r+\beta \gamma r+2 \beta \mu r}{k N}$, $\dfrac{\beta \epsilon \gamma r+\beta \epsilon \mu r+\beta \gamma \mu r+\beta \mu^{2} r}{k N}$.\end{tabular}$}

Result of our algorithm: 

{\footnotesize$\setlength{\extrarowheight}{0pt}\begin{tabular}{l}$N$, $\mu$, $K \epsilon$, $\epsilon \gamma$, $\beta \epsilon r$, $\epsilon+\gamma$.\end{tabular}$}

Result is algebraically independent over $\mathbb{C}$: yes.

\item \exname{SEIR36\_ref} \cite[Table 2, ID 51]{MASSONIS2021441}.

Original generating set information: 13 indeterminates; 58 non-constant functions; maximal total degrees of numerator and denominator are~$(7, 2)$; 6.9~KB total in string representation.

Indeterminates: $N$, $s$, $\nu$, $\phi$, $s_{d}$, $\beta$, $\gamma$, $\mu_{0}$, $\mu_{d}$, $\mu_{i}$, $\phi_{e}$, $\beta_{d}$, $\gamma_{d}$.

Original generating set: too large to be listed.

Result of our algorithm: 

{\footnotesize$\setlength{\extrarowheight}{0pt}\begin{tabular}{l}$N$, $s$, $\nu$, $\phi$, $s_{d}$, $\beta$, $\gamma$, $\mu_{0}$, $\mu_{d}$, $\mu_{i}$, $\phi_{e}$, $\beta_{d}$, $\gamma_{d}$.\end{tabular}$}

Result is algebraically independent over $\mathbb{C}$: yes.

\item \exname{SEIRT} \cite[Eq. (2.1)]{Sauer2021}.

Original generating set information: 4 indeterminates; 6 non-constant functions; maximal total degrees of numerator and denominator are~$(3, 1)$; 155~bytes total in string representation.

Indeterminates: $N$, $\beta$, $\alpha$, $\lambda$.

Original generating set: 

{\footnotesize$\setlength{\extrarowheight}{8pt}\begin{tabular}{l}$-1$, $-N$, $\dfrac{\beta}{N}$, $\alpha+\lambda$, $-\alpha-\lambda$, $\dfrac{\alpha \beta \lambda}{N}$, $\dfrac{\alpha \beta+\beta \lambda}{N}$.\end{tabular}$}

Result of our algorithm: 

{\footnotesize$\setlength{\extrarowheight}{0pt}\begin{tabular}{l}$N$, $\beta$, $\alpha \lambda$, $\alpha+\lambda$.\end{tabular}$}

Result is algebraically independent over $\mathbb{C}$: yes.

\item \exname{SEUIR} \cite[Eq. (5.2)]{Sauer2021}.

Original generating set information: 5 indeterminates; 126 non-constant functions; maximal total degrees of numerator and denominator are~$(11, 2)$; 5.1~KB total in string representation.

Indeterminates: $N$, $d$, $w$, $z$, $\beta$.

Original generating set: too large to be listed.

Result of our algorithm: 

{\footnotesize$\setlength{\extrarowheight}{8pt}\begin{tabular}{l}$d$, $z$, $\dfrac{N w}{\beta}$.\end{tabular}$}

Result is algebraically independent over $\mathbb{C}$: yes.

\item \exname{SIR21} \cite[Figure 1]{Castro2020}.

Original generating set information: 6 indeterminates; 7 non-constant functions; maximal total degrees of numerator and denominator are~$(3, 1)$; 146~bytes total in string representation.

Indeterminates: $N$, $q$, $r$, $\mu$, $p$, $\beta$.

Original generating set: 

{\footnotesize$\setlength{\extrarowheight}{8pt}\begin{tabular}{l}$q$, $q$, $-N$, $\dfrac{N\mspace{2mu} \mu}{\beta}$, $\dfrac{-N\mspace{2mu} p}{\beta}$, $\dfrac{N\mspace{2mu} \mu\mspace{2mu} p+N\mspace{2mu} \mu\mspace{2mu} q}{\beta}$, $\dfrac{-N\mspace{2mu} \mu\mspace{2mu} p-N\mspace{2mu} p\mspace{2mu} r}{\beta}$.\end{tabular}$}

Result of our algorithm: 

{\footnotesize$\setlength{\extrarowheight}{0pt}\begin{tabular}{l}$N$, $q$, $r$, $\mu$, $p$, $\beta$.\end{tabular}$}

Result is algebraically independent over $\mathbb{C}$: yes.









\item \exname{SIR24} \cite[Eq. (2.1)]{SIR24}.

Original generating set information: 5 indeterminates; 68 non-constant functions; maximal total degrees of numerator and denominator are~$(11, 2)$; 5.7~KB total in string representation.

Indeterminates: $K$, $c$, $\mu$, $\phi$, $\gamma$.

Original generating set: too large to be listed.

Result of our algorithm: 

{\footnotesize$\setlength{\extrarowheight}{0pt}\begin{tabular}{l}$K$, $c \phi$, $\gamma+\mu$.\end{tabular}$}

Result is algebraically independent over $\mathbb{C}$: yes.

\item \exname{SIR6} \cite{zheng2020totalvariationregularizationcompartmental}.

Original generating set information: 4 indeterminates; 4 non-constant functions; maximal total degrees of numerator and denominator are~$(2, 1)$; 92~bytes total in string representation.

Indeterminates: $k$, $N$, $\beta$, $\gamma$.

Original generating set: 

{\footnotesize$\setlength{\extrarowheight}{8pt}\begin{tabular}{l}$-N$, $\gamma$, $\dfrac{k N}{\beta}$, $\dfrac{-k N}{\beta}$.\end{tabular}$}

Result of our algorithm: 

{\footnotesize$\setlength{\extrarowheight}{8pt}\begin{tabular}{l}$N$, $\gamma$, $\dfrac{k}{\beta}$.\end{tabular}$}

Result is algebraically independent over $\mathbb{C}$: yes.

\item \exname{SIRC-forced} \cite[Eq. (7)-(11)]{sirc-forced}.

Original generating set information: 6 indeterminates; 2544 non-constant functions; maximal total degrees of numerator and denominator are~$(15, 0)$; 895.4~KB total in string representation.

Indeterminates: $M$, $g$, $b_{0}$, $b_{1}$, $\mu$, $\nu$.

Original generating set: too large to be listed.

Result of our algorithm: 

{\footnotesize$\setlength{\extrarowheight}{0pt}\begin{tabular}{l}$g$, $b_{0}$, $\mu$, $\nu$, $M^{2}$.\end{tabular}$}

Result is algebraically independent over $\mathbb{C}$: yes.

\item \exname{SIWR-multiout.} \cite[Eq. (3) with extra output]{siwr}.

Original generating set information: 7 indeterminates; 134 non-constant functions; maximal total degrees of numerator and denominator are~$(12, 3)$; 107.5~KB total in string representation.

Indeterminates: $a$, $k$, $\beta_i$, $\beta_w$, $\mu$, $\xi$, $\gamma$.

Original generating set: too large to be listed.

Result of our algorithm: 

{\footnotesize$\setlength{\extrarowheight}{0pt}\begin{tabular}{l}$a$, $k$, $\beta_i$, $\beta_w$, $\mu$, $\xi$, $\gamma$.\end{tabular}$}

Result is algebraically independent over $\mathbb{C}$: yes.

\item \exname{SIWR-orig.} \cite[Eq. (3)]{siwr}.

Original generating set information: 7 indeterminates; 761 non-constant functions; maximal total degrees of numerator and denominator are~$(20, 6)$; 5.7~MB total in string representation.

Indeterminates: $a$, $k$, $\beta_i$, $\beta_w$, $\mu$, $\xi$, $\gamma$.

Original generating set: too large to be listed.

Result of our algorithm: 

{\footnotesize$\setlength{\extrarowheight}{0pt}\begin{tabular}{l}$a$, $k$, $\beta_i$, $\beta_w$, $\mu$, $\xi$, $\gamma$.\end{tabular}$}

Result is algebraically independent over $\mathbb{C}$: yes.

\item \exname{SLIQR} \cite[Eq. (4)]{Dankwa2022}.

Original generating set information: 6 indeterminates; 39 non-constant functions; maximal total degrees of numerator and denominator are~$(6, 0)$; 943~bytes total in string representation.

Indeterminates: $\alpha$, $\beta$, $\eta$, $\gamma$, $\sigma$, $N$.

Original generating set: too large to be listed.

Result of our algorithm: 

{\footnotesize$\setlength{\extrarowheight}{0pt}\begin{tabular}{l}$\beta$, $\sigma$, $N$, $\alpha+\gamma$, $\alpha\mspace{2mu} \eta\mspace{2mu} \gamma$, $\alpha\mspace{2mu} \gamma+\eta\mspace{2mu} \gamma\mspace{2mu} \sigma-\gamma\mspace{2mu} \sigma$.\end{tabular}$}

Result is algebraically independent over $\mathbb{C}$: yes.

\item \exname{Transfection} \cite{Transfection-masters}, equations copied from \cite{genssi-2-0}.


Original generating set information: 5 indeterminates; 25 non-constant functions; maximal total degrees of numerator and denominator are~$(6, 1)$; 524~bytes total in string representation.

Indeterminates: $b$, $d_{1}$, $d_{2}$, $d_{3}$, $k_{\operatorname{TL}}$.

Original generating set: 

{\footnotesize$\setlength{\extrarowheight}{17pt}\begin{tabular}{l}$b$, $b^{3}$, $\dfrac{-1}{3}$, $\dfrac{-2}{3}$, $\dfrac{-2}{3}\mspace{2mu} b$, $\dfrac{-1}{3}\mspace{2mu} b^{2}$, $\dfrac{1}{3}\mspace{2mu} b^{3}\mspace{2mu} d_{3}$, $\dfrac{-b\mspace{2mu} d_{2}}{k_{\operatorname{TL}}}$, $\dfrac{-1}{3}\mspace{2mu} b^{3}\mspace{2mu} d_{3}$, $-b+\dfrac{1}{3}\mspace{2mu} d_{3}$, $\dfrac{\dfrac{-1}{3}\mspace{2mu} d_{2}}{k_{\operatorname{TL}}}$, $\dfrac{-b^{2}\mspace{2mu} d_{2}}{k_{\operatorname{TL}}}$, $\dfrac{2}{3}\mspace{2mu} b-\dfrac{1}{3}\mspace{2mu} d_{3}$, $b^{2}+\dfrac{1}{3}\mspace{2mu} b\mspace{2mu} d_{3}$, \\$-b^{2}+\dfrac{1}{3}\mspace{2mu} b\mspace{2mu} d_{3}$, $\dfrac{-b^{2}\mspace{2mu} d_{1}\mspace{2mu} d_{2}}{k_{\operatorname{TL}}}$, $\dfrac{-b^{3}\mspace{2mu} d_{1}\mspace{2mu} d_{2}}{k_{\operatorname{TL}}}$, $\dfrac{\dfrac{-1}{3}\mspace{2mu} b^{3}\mspace{2mu} d_{2}}{k_{\operatorname{TL}}}$, $\dfrac{\dfrac{-1}{3}\mspace{2mu} b\mspace{2mu} d_{1}\mspace{2mu} d_{2}}{k_{\operatorname{TL}}}$, $\dfrac{1}{3}\mspace{2mu} b^{2}-\dfrac{2}{3}\mspace{2mu} b\mspace{2mu} d_{3}$, $\dfrac{\dfrac{-1}{3}\mspace{2mu} b^{4}\mspace{2mu} d_{1}\mspace{2mu} d_{2}}{k_{\operatorname{TL}}}$, $\dfrac{-1}{3}\mspace{2mu} b^{3}-\dfrac{1}{3}\mspace{2mu} b^{2}\mspace{2mu} d_{3}$, \\$\dfrac{-2}{3}\mspace{2mu} b^{3}+\dfrac{1}{3}\mspace{2mu} b^{2}\mspace{2mu} d_{3}$, $\dfrac{-b^{2}\mspace{2mu} d_{2}-b\mspace{2mu} d_{1}\mspace{2mu} d_{2}}{k_{\operatorname{TL}}}$, $\dfrac{-b^{3}\mspace{2mu} d_{2}-b^{2}\mspace{2mu} d_{1}\mspace{2mu} d_{2}}{k_{\operatorname{TL}}}$, $\dfrac{\dfrac{-1}{3}\mspace{2mu} b\mspace{2mu} d_{2}-\dfrac{1}{3}\mspace{2mu} d_{1}\mspace{2mu} d_{2}}{k_{\operatorname{TL}}}$, $\dfrac{\dfrac{-1}{3}\mspace{2mu} b^{4}\mspace{2mu} d_{2}-\dfrac{1}{3}\mspace{2mu} b^{3}\mspace{2mu} d_{1}\mspace{2mu} d_{2}}{k_{\operatorname{TL}}}$.\end{tabular}$}

Result of our algorithm: 

{\footnotesize$\setlength{\extrarowheight}{8pt}\begin{tabular}{l}$b$, $d_{1}$, $d_{3}$, $\dfrac{d_{2}}{k_{\operatorname{TL}}}$.\end{tabular}$}

Result is algebraically independent over $\mathbb{C}$: yes.

\item \exname{SIRT} \cite[Eq. (2.3)]{Tuncer2018}.


Original generating set information: 5 indeterminates; 11 non-constant functions; maximal total degrees of numerator and denominator are~$(7, 1)$; 519~bytes total in string representation.

Indeterminates: $a$, $b$, $d$, $g$, $\nu$.

Original generating set: 

{\footnotesize$\setlength{\extrarowheight}{8pt}\begin{tabular}{l}$-1$, $\dfrac{b}{g}$, $d\mspace{2mu} g+\nu$, $-d\mspace{2mu} g-\nu$, $\dfrac{2\mspace{2mu} b\mspace{2mu} d\mspace{2mu} g+2\mspace{2mu} b\mspace{2mu} \nu}{g}$, $\dfrac{a\mspace{2mu} b+b\mspace{2mu} g+b\mspace{2mu} \nu}{g}$, $a\mspace{2mu} d\mspace{2mu} g+d\mspace{2mu} g^{2}+d\mspace{2mu} g\mspace{2mu} \nu+\nu^{2}$, $-a\mspace{2mu} d\mspace{2mu} g-d\mspace{2mu} g^{2}-d\mspace{2mu} g\mspace{2mu} \nu-\nu^{2}$, \\$\dfrac{b\mspace{2mu} d^{2}\mspace{2mu} g^{2}+2\mspace{2mu} b\mspace{2mu} d\mspace{2mu} g\mspace{2mu} \nu+b\mspace{2mu} \nu^{2}}{g}$, $\dfrac{2\mspace{2mu} a\mspace{2mu} b\mspace{2mu} d\mspace{2mu} g+3\mspace{2mu} a\mspace{2mu} b\mspace{2mu} \nu+2\mspace{2mu} b\mspace{2mu} d\mspace{2mu} g^{2}+2\mspace{2mu} b\mspace{2mu} d\mspace{2mu} g\mspace{2mu} \nu+3\mspace{2mu} b\mspace{2mu} g\mspace{2mu} \nu+2\mspace{2mu} b\mspace{2mu} \nu^{2}}{g}$, \\$\dfrac{a\mspace{2mu} b\mspace{2mu} d^{2}\mspace{2mu} g^{2}\mspace{2mu} \nu+2\mspace{2mu} a\mspace{2mu} b\mspace{2mu} d\mspace{2mu} g\mspace{2mu} \nu^{2}+a\mspace{2mu} b\mspace{2mu} \nu^{3}+b\mspace{2mu} d^{2}\mspace{2mu} g^{3}\mspace{2mu} \nu+2\mspace{2mu} b\mspace{2mu} d\mspace{2mu} g^{2}\mspace{2mu} \nu^{2}+b\mspace{2mu} g\mspace{2mu} \nu^{3}}{g}$, \\$\dfrac{a\mspace{2mu} b\mspace{2mu} d^{2}\mspace{2mu} g^{2}+4\mspace{2mu} a\mspace{2mu} b\mspace{2mu} d\mspace{2mu} g\mspace{2mu} \nu+3\mspace{2mu} a\mspace{2mu} b\mspace{2mu} \nu^{2}+b\mspace{2mu} d^{2}\mspace{2mu} g^{3}+b\mspace{2mu} d^{2}\mspace{2mu} g^{2}\mspace{2mu} \nu+4\mspace{2mu} b\mspace{2mu} d\mspace{2mu} g^{2}\mspace{2mu} \nu+2\mspace{2mu} b\mspace{2mu} d\mspace{2mu} g\mspace{2mu} \nu^{2}+3\mspace{2mu} b\mspace{2mu} g\mspace{2mu} \nu^{2}+b\mspace{2mu} \nu^{3}}{g}$.\end{tabular}$}

Result of our algorithm: 

{\footnotesize$\setlength{\extrarowheight}{8pt}\begin{tabular}{l}$\dfrac{b}{g}$, $d\mspace{2mu} g+\nu$, $a+g+\nu$, $a\mspace{2mu} \nu+g\mspace{2mu} \nu$.\end{tabular}$}

Result is algebraically independent over $\mathbb{C}$: yes.

\item \exname{cLV1} \cite{cLV}, equations copied from \cite{on-the-origins}.

Original generating set information: 15 indeterminates; 98 non-constant functions; maximal total degrees of numerator and denominator are~$(4, 2)$; 17.9~KB total in string representation.

Indeterminates: $g_{1}$, $g_{2}$, $g_{3}$, $A_{11}$, $A_{12}$, $A_{13}$, $A_{21}$, $A_{22}$, $A_{23}$, $A_{31}$, $A_{32}$, $A_{33}$, $B_{11}$, $B_{21}$, $B_{31}$.

Original generating set: too large to be listed.

Result of our algorithm: 

{\footnotesize$\setlength{\extrarowheight}{8pt}\begin{tabular}{l}$g_{1}$, $g_{2}$, $g_{3}$, $A_{11}$, $A_{12}$, $A_{21}$, $A_{22}$, $A_{31}$, $A_{32}$, $B_{11}$, $B_{21}$, $B_{31}$, $\dfrac{A_{13}}{A_{33}}$, $\dfrac{A_{23}}{A_{33}}$.\end{tabular}$}

Result is algebraically independent over $\mathbb{C}$: yes.

\item \exname{genLV} \cite[Eq. (2.3)]{Remien2021GLVIdentifiability}.

Original generating set information: 6 indeterminates; 6 non-constant functions; maximal total degrees of numerator and denominator are~$(3, 1)$; 335~bytes total in string representation.

Indeterminates: $r_{1}$, $r_{2}$, $\beta_{11}$, $\beta_{12}$, $\beta_{21}$, $\beta_{22}$.

Original generating set: 

{\footnotesize$\setlength{\extrarowheight}{8pt}\begin{tabular}{l}$\dfrac{-\beta_{12}-\beta_{22}}{\beta_{12}}$, $\dfrac{-\beta_{12} r_{2}+2 \beta_{22} r_{1}}{\beta_{12}}$, $\dfrac{\beta_{12} r_{1} r_{2}-\beta_{22} r_{1}^{2}}{\beta_{12}}$, $\dfrac{-\beta_{11}^{2} \beta_{22}+\beta_{11} \beta_{12} \beta_{21}}{\beta_{12}}$, \\$\dfrac{-\beta_{11} \beta_{12}+2 \beta_{11} \beta_{22}-\beta_{12} \beta_{21}}{\beta_{12}}$, $\dfrac{\beta_{11} \beta_{12} r_{2}-2 \beta_{11} \beta_{22} r_{1}+\beta_{12} \beta_{21} r_{1}}{\beta_{12}}$.\end{tabular}$}

Result of our algorithm: 

{\footnotesize$\setlength{\extrarowheight}{8pt}\begin{tabular}{l}$r_{1}$, $r_{2}$, $\beta_{11}$, $\beta_{21}$, $\dfrac{\beta_{12}}{\beta_{22}}$.\end{tabular}$}

Result is algebraically independent over $\mathbb{C}$: yes.

\item \exname{p53} \cite{DiStefano2015DynamicSB} (see also~\cite[Eq. (1)-(4)]{Sin}), equations copied from \cite{ReyBarreiro2023}.


Original generating set information: 23 indeterminates; 41 non-constant functions; maximal total degrees of numerator and denominator are~$(9, 4)$; 611~bytes total in string representation.

Indeterminates: {\footnotesize$p_{1}$, $p_{3}$, $p_{4}$, $p_{5}$, $p_{6}$, $p_{7}$, $p_{8}$, $p_{9}$, $p_{10}$, $p_{11}$, $p_{12}$, $p_{13}$, $p_{14}$, $p_{15}$, $p_{16}$, $p_{17}$, $p_{18}$, $p_{20}$, $p_{21}$, $p_{22}$, $p_{23}$, $p_{24}$, $p_{25}$}.

Original generating set: too large to be listed.

Result of our algorithm: 

{\footnotesize$\setlength{\extrarowheight}{0pt}\begin{tabular}{l}$p_{1}$, $p_{3}$, $p_{4}$, $p_{5}$, $p_{6}$, $p_{7}$, $p_{8}$, $p_{9}$, $p_{10}$, $p_{11}$, $p_{12}$, $p_{13}$, $p_{14}$, $p_{15}$, $p_{16}$, $p_{17}$, $p_{18}$, $p_{20}$, $p_{21}$, $p_{23}$, $p_{24}$, $p_{25}$, $p_{22}^{4}$.\end{tabular}$}

Result is algebraically independent over $\mathbb{C}$: yes.

\end{enumerate}

}

\end{document}